\documentclass{article}

\usepackage[all]{xy}
\usepackage{amssymb}
\usepackage{stmaryrd,bm}
\usepackage{amsmath,amsthm} 
\usepackage{graphicx}
\usepackage{url,soul}

\makeatletter
\let \@sverbatim \@verbatim
\def \@verbatim {\@sverbatim \verbatimplus}
{\catcode`'=13 \gdef \verbatimplus{\catcode`'=13 \chardef '=13 }} 
\makeatother

\usepackage{environ}
\makeatletter
\NewEnviron{Lalign}{\tagsleft@true\begin{align}\BODY\end{align}}
\makeatother

\usepackage{enumerate,tikz}
\usepackage{xcolor}

\usepackage[hidelinks]{hyperref} 

\newtheorem{theorem}{Theorem}[section]
\newtheorem{lemma}[theorem]{Lemma}  
  
\newtheorem{definition}[theorem]{Definition}  
\theoremstyle{definition}
\newtheorem{example}[theorem]{Example}
\newtheorem{comment}[theorem]{Comment}

\newcommand{\ssub}[4]{\ensuremath{[#1 \mapsto #2, #3 \mapsto #4]}}
\newcommand{\FE}{\fe}   
\newcommand{\sub}[2]{\ensuremath{[#1 \mapsto #2]}}
\newcommand{\FT}{\SP}
\newcommand{\FNF}{\ensuremath{\textup{FNF}}}
\newcommand{\nf}{\ensuremath{f}}
\newcommand{\EqFFEL}{\FFELe}
\newcommand{\T}{\NT}
\newcommand{\Tone}{\ensuremath{{\mathcal{T}_{A,\triangle}}}}
\newcommand{\Ttwo}{\ensuremath{\mathcal{T}_{A,1, 2}}}
\newcommand{\cd}{\ensuremath{\textup{cd}}}
\newcommand{\dd}{\ensuremath{\textup{dd}}}
\newcommand{\tsd}{\ensuremath{\textup{tsd}}}
\newcommand{\inv}{\ensuremath{g}}
\newcommand{\tlef}{\unlhd}  
\newcommand{\trig}{\unrhd}  

\colorlet{darkblue}{blue!80!black}
\newcommand{\blauw}[1]{\textcolor{darkblue}{#1}}

\newcommand{\true}{\ensuremath{\textit{true}}}
\newcommand{\false}{\ensuremath{\textit{false}}}

\newcommand{\auxf}{h}

\newcommand{\Nat}{\ensuremath{\mathbb{N}}}

\newcommand{\se}{\ensuremath{\mathit{se}}}

\newcommand{\fe}{\ensuremath{\mathit{fe}}}
\newcommand{\tra}{\ensuremath{\mathit{t}}}
\newcommand{\feu}{\ensuremath{\mathit{fe}^\und}}
\newcommand{\mfe}{\ensuremath{\mathit{mfe}}}
\newcommand{\mse}{\ensuremath{\mathit{mse}}}
\newcommand{\mfeu}{\ensuremath{\mathit{mfe}^\und}}
\newcommand{\clfe}{\ensuremath{\mathit{clfe}}}

\newcommand{\clfeu}{\ensuremath{\mathit{clfe}^\und}}

\newcommand{\sfe}{\ensuremath{\mathit{sfe}}}
\newcommand{\sse}{\ensuremath{\mathit{sse}}}
\newcommand{\memt}{\ensuremath{\mathit{m}}}
\newcommand{\memtu}{\ensuremath{\memt^\und}}

\newcommand{\treeneg}{\text{$\widetilde\neg$}}
\newcommand{\negtext}{\bm{[\textit{\textbf{\small neg}}]}}
\newcommand{\fullandtext}{\bm{[\textit{\textbf{\small fulland} $Y$}]}}
\newcommand{\Lfullandtext}{\bm{[\textit{\textbf{\small fulland} $\Le_a(Y)$}]}}
\newcommand{\Rfullandtext}{\bm{[\textit{\textbf{\small fulland} $\Ri_a(Y)$}]}}

\newcommand{\treefulland}{~
     \mathbin{\setlength{\unitlength}{.9ex}
     \begin{picture}(1.6,1.8)(-.4,0)
     \put(-.8,0){$\widetilde\wedge$}
     \put(-.66,-0.1){\circle*{0.7}}
     \end{picture}
     }}
    
\newcommand{\treefullor}{~
     \mathbin{\setlength{\unitlength}{.9ex}
     \begin{picture}(1.6,1.8)(-.4,0)
     \put(-.8,0){{$\widetilde\vee$}}
     \put(-.46,1.3){\circle*{0.7}}
     \end{picture}
     }}

\newcommand{\axname}[1]{\textup{\ensuremath{\textrm{#1}}}}

\newcommand{\CP}{\axname{CP}}

\newcommand{\SCL}{\axname{SCL}}
\newcommand{\FSCL}{\axname{FSCL}}
\newcommand{\FSCLu}{\ensuremath{\axname{FSCL}^\und}}
\newcommand{\MSCL}{\axname{MSCL}}
\newcommand{\MSCLu}{\ensuremath{\axname{MSCL}^\und}}
\newcommand{\SSCL}{\axname{SSCL}}
\newcommand{\MSCLe}{\axname{EqMSCL}}
\newcommand{\SSCLe}{\ensuremath{\axname{EqSSCL}}}

\newcommand{\FEL}{\axname{FEL}}
\newcommand{\FELu}{\ensuremath{\axname{FEL}^\und}}
\newcommand{\FFEL}{\axname{FFEL}}
\newcommand{\FFELu}{\ensuremath{\axname{FFEL}^\und}}
\newcommand{\MFEL}{\axname{MFEL}}
\newcommand{\MFELu}{\ensuremath{\axname{MFEL}^\und}}
\newcommand{\SFEL}{\axname{SFEL}}

\newcommand{\CLSCL}{\axname{C$\ell$SCL}}
\newcommand{\CLSCLtwo}{\axname{C$\ell$SCL$_2$}}

\newcommand{\CLSCLu}{\ensuremath{\CLSCL}}
\newcommand{\CLFELtwo}{\ensuremath{\axname{C$\ell$FEL$_2$}}}
\newcommand{\CLFELu}{\ensuremath{\axname{C$\ell$FEL}}}
\newcommand{\CLFEL}{\CLFELu}

\newcommand{\alp}[1]{\ensuremath{\{#1\}}}

\newcommand{\SPsc}{\ensuremath{{\mathcal{S}_A}}}
\newcommand{\SPbeta}{\ensuremath{{\mathcal{SP}_{\!A,\beta}}}}

\newcommand{\SP}{\ensuremath{{\mathcal{SP}_{\!A}}}}
\newcommand{\SPu}{\ensuremath{\SP^{\!\und}}}

\newcommand{\SigFEL}{\ensuremath{\Sigma_\FEL(A)}}
\newcommand{\SigFELu}{\ensuremath{\Sigma_\FEL^\und(A)}}

\newcommand{\FFELe}{\ensuremath{\axname{EqFFEL}}}
\newcommand{\FFELeu}{\ensuremath{\axname{EqFFEL}^\und}}
\newcommand{\MFELe}{\axname{EqMFEL}}
\newcommand{\MFELeu}{\ensuremath{\axname{EqMFEL}^\und}}
\newcommand{\CLFELe}{\ensuremath{\axname{EqC$\ell$FEL$_2$}}}
\newcommand{\CLFELeu}{\ensuremath{\axname{EqC$\ell$FEL}^\und}}
\newcommand{\SFELe}{\ensuremath{\axname{EqSFEL}}}

\newcommand{\NT}{\ensuremath{\mathcal{T}_A}}

\newcommand{\NTu}{\ensuremath{{\mathcal{T}_A^\und}}}
\newcommand{\Le}{L}
\newcommand{\Ri}{R}
\newcommand{\Leu}{L^\und}
\newcommand{\Riu}{R^\und}

\newcommand{\lef}{\ensuremath{\scalebox{0.78}{\raisebox{.1pt}[0pt][0pt]{$\;\lhd\;$}}}}
\newcommand{\rig}{\ensuremath{\scalebox{0.78}{\raisebox{.1pt}[0pt][0pt]{$\;\rhd\;$}}}}

\newcommand{\tr}{\ensuremath{{\sf T}}}
\newcommand{\fa}{\ensuremath{{\sf F}}}
\newcommand{\und}{\ensuremath{{\sf U}}}
\newcommand{\undefi}{\ensuremath{\textit{undefined}}}
\newcommand{\wfa}{\widetilde\fa}

\newcommand{\leftand}{~
     \mathbin{\setlength{\unitlength}{1ex}
     \begin{picture}(1.4,1.8)(-.3,0)
     \put(-.6,0){$\wedge$}
     \put(-.54,-0.2){\textcolor{white}{\circle*{0.6}}}
     \put(-.54,-0.2){\circle{0.6}}
     \end{picture}
     }}
     
\newcommand{\fulland}{~
     \mathbin{\setlength{\unitlength}{1ex}
     \begin{picture}(1.4,1.8)(-.3,0)
     \put(-.6,0.2){$\wedge$}
     \put(-.46,-0.){\circle*{0.6}}
     \end{picture}
     }}

\newcommand{\leftor}{~
     \mathbin{\setlength{\unitlength}{1ex}
     \begin{picture}(1.4,1.8)(-.3,0)
     \put(-.6,0){$\vee$}
     \put(-.54,1.54){\textcolor{white}{\circle*{0.6}}}
     \put(-.54,1.54){\circle{0.6}}
     \end{picture}
     }}

\newcommand{\fullor}{~
     \mathbin{\setlength{\unitlength}{1ex}
     \begin{picture}(1.4,1.8)(-.3,0)
     \put(-.6,0){$\vee$}
     \put(-.42,1.4){\circle*{0.6}}
     \end{picture}
     }}
     
\newcommand{\fullimp}{~
     \mathbin{\setlength{\unitlength}{1ex}
     \begin{picture}(1.8,1.8)
     \put(-.0,0){$\rightarrow$}
     \put(-.1,0.57){\circle*{0.6}}
     \end{picture}
     ~}}

\newcommand{\lxor}{~
     \mathbin{\setlength{\unitlength}{1ex}
     \begin{picture}(1.5,1.8)(-.5,0)
     \put(-.8,0){$\oplus$}
     \put(-1.62,0.02){{\small$\leftarrow$}}
     \end{picture}
     }}
\newcommand{\lxoro}{~
     \mathbin{\setlength{\unitlength}{1ex}
     \begin{picture}(1.5,1.8)(-.5,0)
     \put(-.8,0){$\oplus$}
     \put(-1.02,0.56){\circle{0.6}}
     \end{picture}
     }}
     
\newcommand{\lxorC}{~
     \mathbin{\setlength{\unitlength}{1ex}
     \begin{picture}(1.5,1.8)(-.5,0)
     \put(-.8,0){$\oplus$}
     \put(-1.02,0.56){\circle*{0.6}}
     \end{picture}
     }}

\newcommand{\liffo}{~
     \mathbin{\setlength{\unitlength}{1ex}
     \begin{picture}(1.9,1.8)(-.5,0)\linethickness{0.16mm}
     \put(-.8,0){$\leftrightarrow$}
     \put(-.98,0.58){\circle{0.6}}
          \end{picture}
     }}
     
\newcommand{\liff}{~
     \mathbin{\setlength{\unitlength}{1ex}
     \begin{picture}(1.9,1.8)(-.5,0)\linethickness{0.16mm}
     \put(-.8,0){$\leftrightarrow$}
     \put(-1.38,0.02){{\small$\leftarrow$}}
          \end{picture}
     }}

\newcommand{\liffC}{~
     \mathbin{\setlength{\unitlength}{1ex}
     \begin{picture}(1.9,1.8)(-.5,0)\linethickness{0.16mm}
     \put(-.8,0){$\leftrightarrow$}
     \put(-.98,0.58){\circle*{0.6}}
          \end{picture}
     }}

\begin{document}

\title{Fully Evaluated Left-Sequential Logics}
\date{}

\author{
Alban Ponse \& Daan J.C. Staudt\\[2mm]
\small Informatics Institute, University of Amsterdam,\\
\small Science Park 900, 1098 XH, Amsterdam, The Netherlands\\
\small 
 a.ponse@uva.nl \& daan@daanstaudt.nl
}

\maketitle

\begin{abstract}
We consider a family of two-valued ``fully evaluated left-sequential logics" (FELs), of which Free FEL 
(defined by Staudt in 2012) is most distinguishing (weakest) and immune to atomic side effects. 
Next is Memorising FEL, in which evaluations of subexpressions are memorised. The following stronger
logic is Conditional FEL (inspired by Guzmán and Squier's Conditional logic, 1990).
The strongest FEL is static FEL, a sequential version of propositional logic. 
We use evaluation trees as a simple, intuitive semantics and provide complete axiomatisations
for closed terms (left-sequential propositional expressions).

For each FEL except Static FEL, we also define its three-valued version, with a constant U for 
``undefinedness" and again provide complete, independent axiomatisations, each one 
containing two additional axioms for U on top of the axiomatisations
of the two-valued case.
In this setting, the strongest FEL is equivalent to Bochvar's strict logic.
\\[2mm]
\textbf{Keywords \& phrases:}
full left-sequential evaluation, strict evaluation, 
left-sequential connectives, Bochvar's logic
\end{abstract}

{\fontsize{9}{5}\selectfont \tableofcontents}

\section{Introduction}
\label{sec:1}
This paper has its origin in the work of Staudt~\cite{Stau},
in which so called ``Free Fully Evaluated Left-Sequential Logic" (\FFEL) was introduced,
together with ``evaluation trees" as a simple semantics and an equational axiomatisation
of their equality.

We define a family of ``Fully Evaluated Left-Sequential Logics" (\FEL s) that is about full 
left-sequential evaluation (also called strict evaluation) and of which \FFEL\ is the most distinguishing
(weakest).
As in~\cite{Stau}, we use for the Boolean connectives a dedicated notation that prescribes
a full, left-sequential evaluation strategy: we consider terms
(propositional expressions) that are built from
atoms $a, b, c, ...$ (propositional variables) and constants \tr\ and \fa\
for truth and falsehood by composition with negation ($\neg$) and
sequential connectives: left-sequential conjunction, notation $\fulland$, and left-sequential 
disjunction, notation $\fullor$, where the little black dot 
denotes that the left argument
must be evaluated first, and then the right argument.
For two-valued \FEL s we discern a hierarchy starting with
\FFEL, which is immune to (atomic) side effects, 
and ending with a sequential version of propositional logic.

\emph{Evaluation trees} are binary trees with internal nodes labeled by atoms and 
leaves labeled with either \tr\ or \fa, and provide a simple semantics for propositional expressions: 
a path from the root to a leaf models an evaluation. 
The left branch from an internal node indicates that its atom is evaluated \true, and the right 
branch that the atom evaluates to \false.
The leaves of an evaluation tree represent evaluation results. 
Atom $a$ has as its semantics the evaluation tree
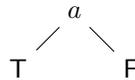
\begin{center}
\begin{tikzpicture}[%
      level distance=7.5mm,
      level 1/.style={sibling distance=15mm},
      level 2/.style={sibling distance=7.5mm},
      baseline=(current bounding box.center)]
      \node (a) {$a$}
        child {node (b1) {$\tr$}
        }
        child {node (b2) {$\fa$}
        };
      \end{tikzpicture}
\end{center}
and for example, $a\fulland b$
has as its semantics the evaluation tree
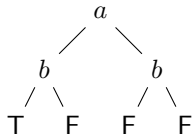
\begin{center}
\begin{tikzpicture}[%
      level distance=7.5mm,
      level 1/.style={sibling distance=15mm},
      level 2/.style={sibling distance=7.5mm},
      baseline=(current bounding box.center)]
      \node (a) {$a$}
        child {node (b1) {$b$}
          child {node (d1) {$\tr$}} 
          child {node (d2) {$\fa$}}
        }
        child {node (b2) {$b$}
          child {node (d1) {$\fa$}} 
          child {node (d2) {$\fa$}}
        };
      \end{tikzpicture}
\end{center}
which is composed from the evaluation trees of atoms $a$ and $b$.
If in $a\fulland b$, atom $a$ evaluates to \true, then atom $b$ is evaluated and determines the 
evaluation result, and if $a$ evaluates to \false, then the evaluation result
is \false, but $b$ is (still) evaluated. 
Note that in the example tree of $a\fulland b$, replacing $b$ by $a$ has an evaluation tree that 
is \emph{not} equal to that of $a$; the evaluation of the first $a$ can have a side effect that 
changes the second evaluation result.

In~\cite{Stau}, evaluation trees were introduced 
as the basis of a relatively simple semantic framework, both for \FFEL\ and 
for so-called \emph{Free Short-Circuit Logic} (FSCL)~\cite{BPS13},\footnote{%
    In ``short-circuit logics" (SCLs), short-circuit evaluation is prescribed by the
    short-circuit connectives ${\leftand}$ and ${\leftor}$; in the journal paper~\cite{PS18}, 
    the FSCL-part of~\cite{Stau} is also discussed.} 
and completeness results for both \FFEL\ and FSCL were provided, i.e.\ equational axiomatisations
of the equality of their evaluation trees. 
For example, \FFEL\ refutes the idempotence axiom $x\fulland x=x$,
but satisfies the associativity of the binary connectives.
In the subsequent paper~\cite{BP15}, transformations on evaluation trees were defined in order to 
provide a semantic basis for other short-circuit logics, see~\cite{PS18,BPS21,BP23}.
In this paper, we introduce three new FELs with associated evaluation trees, each of which has a 
counterpart in short-circuit logic, and give complete, equational axiomatisations:

\begin{description}
 \setlength\itemsep{0mm}
\item[Memorising FEL (\MFEL).]
The logic \MFEL\ characterises the evaluation strategy in which the
evaluation of each atom in an expression is memorised: 
no complete path in a ``memorising evaluation tree" contains multiple occurrences of the same atom.
For example, $a$ and $a\fulland a$ have the same memorising evaluation tree (that of $a$).
In \MFEL, atoms cannot have side effects and the evaluation order of a propositional expression
is prescribed. As an example, the memorising evaluation trees of $a\fulland b$ and $b\fulland a$ are 
different.
\MFEL\ is axiomatised by adding one equational axiom to those of \FFEL.

\item[Conditional FEL (C$\bm{\ell}$FEL$_2$).]
The name \CLFELtwo\ refers to Conditional logic, defined by Guzmán and Squier~\cite{GS90},
and the subscript 2 refers to its two-valued version. In \CLFELtwo,
${\fulland}$ and ${\fullor}$ are taken to be commutative. 
In~\cite{BP23}, we show that this
commutativity is a consequence of Conditional logic.
\CLFELtwo-evaluation trees are memorising evaluation trees in which the order of 
commuting atoms can be changed while preserving equivalence.
\CLFELtwo\ is axiomatised by the axioms of \MFEL\ and $x\fulland y=y\fulland x$.

\item[Static FEL (\SFEL).] This logic is just a sequential version of propositional logic. 
\SFEL\ is axiomatised by adding the axiom $x\fulland\fa=\fa$ to those of \CLFELtwo.
\end{description}

Evaluation trees for full left-sequential evaluation can be easily extended to a three-valued 
case by including
\und\ as a leaf, which  represents the truth value \undefi.
Evaluation trees with leaves in $\{\tr,\fa,\und\}$ 
were introduced in~\cite{BPS21}, where they serve as a semantics for 
short-circuit logic with undefinedness.
In both \FEL\ and \SCL, atom $a$  has as its semantics the left tree below, which shows that
the evaluation of each atom can be \undefi. Another example in  \FEL\ is the 
evaluation tree of
$a\fulland (b\fulland\und)$, shown on the right:
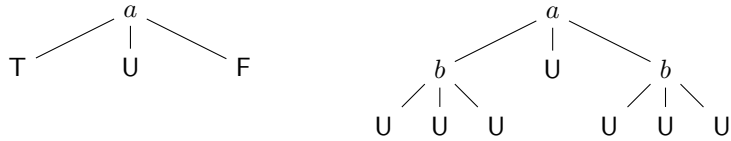
\[
\begin{array}{ll}
\begin{array}{l}
\begin{tikzpicture}[%
level distance=7.5mm,
level 1/.style={sibling distance=15mm},
level 2/.style={sibling distance=7.5mm},
level 3/.style={sibling distance=3.75mm}
]
\node (root) {$a$}
  child {node (lef1) {$\tr$}
  }
  child {node (mid1) {$\und$}
  }
  child {node (rig1) {$\fa$}
  };
\end{tikzpicture}
\\[8mm]
\end{array}
\qquad
&
\begin{array}{l}
\begin{tikzpicture}[%
level distance=7.5mm,
level 1/.style={sibling distance=15mm},
level 2/.style={sibling distance=7.5mm},
level 3/.style={sibling distance=3.75mm}
]
\node (root) {$a$}
  child {node (lef1) {$b$}
    child {node (lef2) {$\und$}
    }
    child {node (mid2) {$\und$}
    }
    child {node (rig2) {$\und$}
    }
  }
  child {node (mid1) {$\und$}
  }
  child {node (rig1) {$b$}
    child {node (riglef2) {$\und$}
    }
    child {node (rigmid2) {$\und$}
    }
    child {node (rigrig2) {$\und$}
    }
  };
\end{tikzpicture}
\end{array}
\end{array}
\]
The general idea here is that in propositional expressions, \und\ aborts further evaluation.
This property can be equationally axiomatised by $\neg\und=\und$ and $\und\fulland x=\und$ 
(and by duality it follows that 
$\und\fullor x=\und$).
A programming-oriented example is a condition $\texttt{[y != 0]}\:\fulland\:\texttt{[x/y > 17]}$,
where the atom \texttt{[x/y~>~17]} can be undefined.

For each \FEL\ except \SFEL, we define the three-valued version and give  
equational axiomatisations consisting of those for  
the two-valued case and the two mentioned axioms for \und. 
In this family, \CLFEL\ with undefinedness 
is equivalent to Bochvar’s strict logic~\cite{Boc38} and \und\ is fully absorptive. 
Note that \SFEL\ cannot be extended with \und: $\fa=\und\fulland\fa=\und$.

We then present alternative, independent axiomatisations for these equational axiomatisations
that are simpler and more concise.
We end the paper with some remarks on our specific notation for left-sequential connectives,
expressiveness, and related work.

\paragraph{\textit{Prover9} and \emph{Mace4}.}
Apart from inductive proofs, all presented derivabilities from
equational axiomatisations were found by or 
checked with the theorem 
prover \emph{Prover9}, and finite (counter) models were generated with the tool \emph{Mace4}. 
For both tools, see~\cite{Prover9}. We used these tools on a Macbook Pro with 
a 2.4 GHz dual-core Intel Core i5 processor and 4GB of RAM. 
Average running times are mostly
given in mere seconds and rounded up (for example, 2s).

\paragraph{\emph{Structure of the paper.}}
In Section~\ref{sec:2}, we review Free FEL (\FFEL) and its equational axiomatisation.
In Section~\ref{sec:3}, we define \FFELu, the extension of \FFEL\ with undefinedness.
In Section~\ref{sec:4} we define Memorising FEL (\MFEL) and provide an equational axiomatisation,
and in Section~\ref{sec:5} we extend \MFEL\ to \MFELu\ with undefinedness.
In Section~\ref{sec:6}, we define two-valued Conditional FEL (\CLFELtwo) and its extension 
\CLFELu\ with undefinedness.  
In Section~\ref{sec:7}, we briefly consider Static FEL (\SFEL), and then 
provide for each \FEL\ a set of independent axioms.
Section~\ref{sec:8} contains some discussion and conclusions.

\section{Free FEL (\FFEL) and an equational axiomatisation} 
\label{sec:2}
In this section, we review evaluation trees and the logic \FFEL.
We repeat the main results about \FFEL, that is, its evaluation trees and axiomatisation, and 
the completeness proof as presented in~\cite{Stau}.

From this point on, we assume that $A$ is a countable set of atoms with typical elements $a$, $b$, $c$. 
We start with a formal definition of evaluation trees. 

\begin{definition}
\label{def:treesF}
The set \NT\ of \textbf{evaluation trees} over $A$ with leaves in 
$\{\tr, \fa\}$ is defined inductively by
\[\tr\in\NT,\quad\fa\in\NT, \quad (X\unlhd a\unrhd Y)\in\NT 
~\text{ for any }X,Y \in \NT \text{  and } a\in A.\]
The operator $\_\unlhd a\unrhd\_$ is called 
\textbf{tree composition over $a$}.
In the evaluation tree $X \unlhd a \unrhd Y$, 
the root is represented by $a$,
the left branch by $X$ and the right branch by $Y$. 
\\[2mm]
The \textbf{depth} $d:\NT\to\Nat$ of an evaluation tree is defined by 
$d(\tr) = d(\fa) = 0$ and
$d(Y \unlhd a \unrhd Z) = 1 + \max(d(Y ), d(Z))$.
\end{definition}

The leaves of an evaluation tree represent evaluation results (we use 
the constants \tr\ and \fa\ for \true\ and \false). 
Next to the formal notation for evaluation
trees we also use a more pictorial representation. For example,
the tree
\[
(\fa\unlhd a\unrhd\fa)\unlhd b\unrhd(\tr\unlhd a\unrhd\fa)
\]
can be represented as follows, where $\unlhd$ yields a left branch, and $\unrhd$ a right branch:
\begin{equation}
\label{plaatje}
\tag{Picture 1}
\hspace{-12mm}
\begin{tikzpicture}[%
      level distance=7.5mm,
      level 1/.style={sibling distance=15mm},
      level 2/.style={sibling distance=7.5mm},
      baseline=(current bounding box.center)]
      \node (a) {$b$}
        child {node (b1) {$a$}
          child {node (d1) {$\fa$}} 
          child {node (d2) {$\fa$}}
        }
        child {node (b2) {$a$}
          child {node (d1) {$\tr$}} 
          child {node (d2) {$\fa$}}
        };
      \end{tikzpicture}
\end{equation}

In order to define a semantics for full evaluation of negation and the left-sequential 
connectives, we first define the \emph{leaf replacement} operator, 
`replacement' for short, on trees in \NT\ as follows. 
For $X\in\NT$, the replacement of \tr\ with $Y$ and $\fa$ with $Z$ in $X$, denoted
\[X[\tr\mapsto Y, \fa \mapsto Z]\]
is defined recursively by 
\begin{align*}
\tr[\tr\mapsto Y,\fa\mapsto Z]&= Y,\\
\fa[\tr\mapsto Y,\fa\mapsto Z]&= Z,\\
(X_1\unlhd a\unrhd X_2)[\tr\mapsto Y,\fa\mapsto Z]
&=X_1[\tr\mapsto Y,\fa\mapsto Z]\unlhd a\unrhd X_2[\tr\mapsto Y,\fa\mapsto Z].
\end{align*}
We note that the order in which the replacements of leaves of 
$X$ is listed is irrelevant and adopt the convention of not listing  
identities inside the brackets, e.g., 
$X[\fa\mapsto Z]=X[\tr\mapsto \tr,\fa\mapsto Z]$.
By structural induction it follows that repeated replacements satisfy 
\[
X[\tr\mapsto Y_1,\fa\mapsto Z_1][\tr\mapsto Y_2,\fa\mapsto Z_2]=
X[\tr\mapsto Y_1[\tr\mapsto Y_2,\fa\mapsto Z_2],~
\fa\mapsto Z_1[\tr\mapsto Y_2,\fa\mapsto Z_2]].
\]

The set \SP\ of (left-sequential) propositional expressions over $A$
that prescribe full left-sequential evaluation is defined
by the following grammar ($a\in A$):
\[P ::= \tr\mid\fa\mid a\mid\neg P\mid P\fulland P\mid P\fullor P,\]
and we refer to its signature by
\[\SigFEL=\{\fulland,\fullor,\neg,\tr,\fa,a\mid a\in A\}.\]
We interpret expressions in \SP\ as evaluation trees
by a function \fe\ (abbreviating  full evaluation).

\begin{definition}
\label{def:fe}
The unary \textbf{full evaluation} function $\fe : \SP \to\NT$ 
is defined as
follows, where $a\in A$:
\begin{align*}
\fe(\tr) &= \tr,
&\fe(\neg P)&=\fe(P)[\tr\mapsto \fa,\fa\mapsto \tr],\\
\fe(\fa) &= \fa,
&\fe(P \fulland Q)&= \fe(P)[\tr\mapsto \fe(Q),\fa\mapsto\fe(Q)[\tr\mapsto\fa]],\\
\fe(a)&=\tr\unlhd a\unrhd \fa,
&\fe(P \fullor Q)&= \fe(P)[\tr\mapsto\fe(Q)[\fa\mapsto\tr], \fa\mapsto \fe(Q)].
\end{align*}
\end{definition}

The overloading of the symbol \tr\ in $\fe(\tr)=\tr$ will not
cause confusion (and similarly for \fa).
As a simple example we derive the evaluation tree of $\neg b\fulland a$:
\begin{align*}
\fe(\neg b\fulland a)&=\fe(\neg b)[\tr\mapsto \fe(a),\fa\mapsto\fe(a)[\tr\mapsto\fa]]\\
&=(\fa\unlhd b\unrhd\tr)[\tr\mapsto \fe(a),\fa\mapsto\fe(a)[\tr\mapsto\fa]]\\
&=(\fa\unlhd a\unrhd\fa)\unlhd b\unrhd(\tr\unlhd a\unrhd\fa),
\end{align*}
which can be visualized as in~\ref{plaatje}.
Also, $\fe(\neg(b\fullor\neg a))=(\fa\unlhd a\unrhd\fa)\unlhd b\unrhd(\tr\unlhd a\unrhd\fa)$.
An evaluation tree $\fe(P)$ represents full evaluation in a way that can be
compared to the notion of a truth table for propositional logic in that it 
represents each possible evaluation of $P$. However, there are some important differences with
truth tables: in $\fe(P)$, the sequentiality
of $P$'s evaluation is represented, 
the same atom may occur multiple times in $\fe(P)$ with different evaluation values, and all atoms
are evaluated, thus yielding 
evaluation trees that are perfect. Evaluation trees that are
not perfect occur in the setting with 
\emph{short-circuit} connectives; in that setting, the connectives ${\fulland}$ and ${\fullor}$
are definable (we return to this in Sections~\ref{sec:7} and~\ref{sec:8}).

\begin{table}[t]
{
\centering
\hrule
~\\[-5mm]
\begin{align}
\label{FFEL1}
\tag{FFEL1}
\fa&=\neg\tr
\\
\label{FFEL2}
\tag{FFEL2}
x\fullor y
&=\neg(\neg x\fulland\neg y)
\\ 
\label{FFEL3}
\tag{FFEL3}
\neg\neg x&=x
\\
\label{FFEL4}
\tag{FFEL4}
(x\fulland y)\fulland z&=x\fulland (y\fulland z)
\\
\label{FFEL5}
\tag{FFEL5}
\tr\fulland x&=x
\\
\label{FFEL6}
\tag{FFEL6}
x\fulland \tr&= x
\\
\label{FFEL7}
\tag{FFEL7}
x\fulland \fa&=\fa\fulland x
\\
\label{FFEL8}
\tag{FFEL8}
\neg x\fulland \fa&=x\fulland \fa
\\
\label{FFEL9}
\tag{FFEL9}
(x\fulland \fa)\fullor y
&=(x\fullor\tr)\fulland y
\\
\label{FFEL10}
\tag{FFEL10}
x\fullor (y\fulland \fa)&=x\fulland (y\fullor\tr)
\end{align}
\hrule
}
\caption{\FFELe, axioms for \FFEL}
\label{tab:FFELe}
\end{table}

\begin{definition}
\label{def:congruence}
The binary relation $=_{\fe}$ 
on \SP\ defined by $P=_{\fe} Q\iff\fe(P)=\fe(Q)$
is called \textbf{free full valuation congruence}.
\end{definition}

\begin{lemma}
\label{la:Fcongruence}
The relation $=_{\fe}$ is  a congruence.
\end{lemma}

\begin{proof}
It is immediately clear that identity, symmetry and transitivity are preserved.
For congruence we show only that for all $P, Q, R \in \SP$, $P =_\fe
Q$ implies $R \fulland P =_\fe R \fulland Q$. The other cases proceed
in a similar fashion. If $\fe(P) = \fe(Q)$ then
$\fe(P)[\tr\mapsto\fa] = \fe(Q)[\tr\mapsto\fa]$, so
\begin{equation*}
\fe(R)[\tr\mapsto\fe(P),\fa\mapsto\fe(P)[\tr\mapsto\fa]] =
\fe(R)[\tr\mapsto\fe(Q),\fa\mapsto\fe(Q)[\tr\mapsto\fa]].
\end{equation*}
Therefore, by definition of $\fe$, $R \fulland P =_\fe R \fulland Q$.
\end{proof}

\begin{definition}
\label{def:FFEL}
A \textbf{Fully Evaluated Left-Sequential Logic \emph{(\FEL)}} is a logic that satisfies the 
consequences of \fe-equality. 
\textbf{Free Fully Evaluated Left-Sequential Logic \emph{(\FFEL)}} is the fully evaluated 
left-sequential logic that satisfies no more consequences than those of \fe-equality, i.e., 
for all $P, Q \in\SP$,
\[\FEL\models P=Q \Longleftarrow P=_\fe Q \quad\text{and}\quad \FFEL\models P=Q \iff P=_\fe Q.
\]
\end{definition}

In~\cite{Stau}, it is proven that the set \FFELe\ of equational axioms in Table~\ref{tab:FFELe} 
(completely) axiomatises equality of evaluation trees; this proof is based on normal forms for \SP.
Below, we repeat the completeness proof of~\cite{Stau} almost verbatim.
In the remainder of this paper we define several equational logics that model sequential 
full evaluation according to various evaluation strategies and the
corresponding adaptations of evaluation trees. A general reference to equational
logics is~\cite{Burris}.

The following lemma shows some useful equations illustrating the special
properties of terms of the form $x \fulland \fa$ and $x \fullor \tr$. The
first is an `extension' of axiom~\eqref{FFEL8} and the others show two different
ways how terms of the form $x \fullor \tr$, and by duality terms of the form
$x \fulland \fa$, can change the main connective of a term.

\begin{lemma}
\label{la:feqs}
The following equations are consequences of $\FFELe$.
\textup{
\begin{enumerate}
\setlength\itemsep{5pt}
\item $x \fulland (y \fulland \fa) = \neg x \fulland (y \fulland \fa)$,
  \label{eq:a1}
\item $(x \fullor \tr) \fulland y = \neg(x \fullor \tr) \fullor y$,
  \label{eq:a5}
\item $x \fullor (y \fulland (z \fullor \tr)) = (x \fullor y) \fulland
  (z \fullor \tr)$.
  \label{eq:a2}
\end{enumerate}
}
\end{lemma}

\begin{proof}
We derive the equations in order.
\begin{align*}
x \fulland (y \fulland \fa)
&= (x \fulland \fa) \fulland y
&&\textrm{by \eqref{FFEL7} and \eqref{FFEL4}} \\
&= (\neg x \fulland \fa) \fulland y
&&\textrm{by \eqref{FFEL8}} \\
&= \neg x \fulland (y \fulland \fa),
&&\textrm{by \eqref{FFEL7} and \eqref{FFEL4}} 
\\[5pt]
(x \fullor \tr) \fulland y
&= (x \fulland \fa) \fullor y
&&\textrm{by \eqref{FFEL9}} \\
&= (\neg x \fulland \fa) \fullor y
&&\textrm{by \eqref{FFEL8}} \\
&= (\neg x \fulland \neg \tr) \fullor y
&&\textrm{by \eqref{FFEL1}} \\
&= \neg(x \fullor \tr) \fullor y,
&&\textrm{by \eqref{FFEL3} and \eqref{FFEL2}} 
\\[5pt]
x \fullor (y \fulland (z \fullor \tr))
&= x \fullor (y \fullor (z \fulland \fa))
&&\textrm{by \eqref{FFEL10}} \\
&= (x \fullor y) \fullor (z \fulland \fa)
&&\textrm{by the dual of \eqref{FFEL4}} \\
&= (x \fullor y) \fulland (z \fullor \tr).
&&\textrm{by \eqref{FFEL10}} 
\end{align*}
\end{proof}

We now turn to the equational logic defined by $\FFELe$, which we will show is an
axiomatization of $\FFEL$. This set of equations was first presented by Blok in
\cite{Blo11}.

If two $\FEL$-terms $s$ and $t$, possibly containing variables, are
derivable in $\FFELe$, we write $\FFELe \vdash
s = t$ and say that $s$ and $t$ are derivably equal. By virtue of
\eqref{FFEL1} through \eqref{FFEL3}, ${\fulland}$ is the dual of ${\fullor}$
and \tr\ is the dual of \fa,
and hence the duals of the equations in $\FFELe$ are also derivable. We will
use this fact implicitly throughout our proofs.

\begin{lemma}[Soundness]
\label{la:Fsound}
For all $P, Q \in \SP$, if $\FFELe \vdash P = Q$ then $\FFEL \models P = Q$.
\end{lemma}

\begin{proof}
By Lemma~\ref{la:Fcongruence}, $=_\fe$ is a congruence, so it suffices to prove the validity of 
the equations in $\FFELe$, which is easily verified. As an example we show this for \eqref{FFEL8}. 
\begin{align*}
\FE(P \fulland \fa) &=
\FE(P)\ssub{\tr}{\fa}{\fa}{\fa\sub{\tr}{\fa}}
&&\textrm{by definition} \\
&= \FE(P)\sub{\tr}{\fa}
&&\textrm{because $\fa\sub{\tr}{\fa} = \fa$} \\
&= \FE(P)\ssub{\tr}{\fa}{\fa}{\tr}\sub{\tr}{\fa}
&&\textrm{by induction} \\
&= \FE(\neg P \fulland \fa),
\end{align*}
where the induction that proves the third equality is on the structure of
evaluation trees.
\end{proof}

\paragraph{FEL Normal Form.}
To aid in our completeness proof we define a normal form for $\FEL$-terms. Due
to the possible presence of side effects, $\FFEL$ does not identify terms which
contain different atoms or the same atoms in a different order. Because of
this, common normal forms for propositional logic are not normal forms for $\FEL$-terms. For
example, rewriting a term to Conjunctive Normal Form or Disjunctive Normal Form
may require duplicating some of the atoms in the term, thus yielding a term
that is not derivably equal to the original. We first present the grammar for
our normal form, before motivating it. The normal form we present here is an
adaptation of a normal form proposed by Blok in \cite{Blo11}.

\begin{definition}
\label{def:fnf}
A term $P \in \FT$ is said to be in \textbf{$\FEL$ Normal Form $(\FNF)$} if it
is generated by the following grammar.
\begin{align*}
P \in \FNF &::= P^\tr ~\mid~ P^\fa ~\mid~ P^\tr \fulland P^* \\
P^* &::= P^c ~\mid~ P^d \\
P^c &::= P^\ell ~\mid~ P^* \fulland P^d \\
P^d &::= P^\ell ~\mid~ P^* \fullor P^c \\
P^\ell &::= a \fulland P^\tr ~\mid~ \neg a \fulland P^\tr \\
P^\tr &::= \tr ~\mid~ a \fullor P^\tr \\
P^\fa &::= \fa ~\mid~ a \fulland P^\fa,
\end{align*}
where $a \in A$. We refer to $P^*$-forms as $*$-terms, to $P^\ell$-forms as
$\ell$-terms, to $P^\tr$-forms as $\tr$-terms and to $P^\fa$-forms as
$\fa$-terms. A term of the form $P^\tr \fulland P^*$ is referred to as a
$\tr$-$*$-term.
\end{definition}

We immediately note that if it were not for the presence of $\tr$ and
$\fa$ we could define a much simpler normal form. In that case it would
suffice to `push in' or `push down' the negations, thus obtaining a Negation
Normal Form (NNF), as exists for propositional logic. Naturally if our set $A$ of atoms is empty, the
truth value constants would be a normal form.

When considering the image of $\FE$ we note that some trees only have
$\tr$-leaves, some only have $\fa$-leaves and some have both
$\tr$-leaves and $\fa$-leaves. For any $\FEL$-term $P$, $\FE(P \fullor
\tr)$ is a tree with only $\tr$-leaves, as can easily be seen from the
definition of $\FE$. All terms $P$ such that $\FE(P)$ only has $\tr$-leaves
are rewritten to $\tr$-terms.  Similarly $\FE(P \fulland \fa)$ is a tree
with only $\fa$-leaves. All terms $P$ such that $\FE(P)$ only has
$\fa$-leaves are rewritten to $\fa$-terms. The simplest trees in the
image of $\FE$ that have both $\tr$-leaves and $\fa$-leaves are $\FE(a)$
for $a \in A$. Any (occurrence of an) atom that determines (in whole or in part)
the evaluation result of a term, such as $a$ in this example, is referred to as a
determinative (occurrence of an) atom. This as opposed to a non-determinative
(occurrence of an) atom, such as the $a$ in $a \fullor \tr$, which does not
determine (either in whole or in part) the evaluation result of the term.  Note that a term
$P$ such that $\FE(P)$ contains both $\tr$ and $\fa$ must contain at least
one determinative atom.

Terms that contain at least one determinative atom will be rewritten to
$\tr$-$*$-terms. In $\tr$-$*$-terms we encode each determinative atom
together with the non-determinative atoms that occur between it and the next
determinative atom in the term (reading from left to right) as an $\ell$-term.
Observe that the first atom in an $\ell$-term is the (only) determinative atom
in that $\ell$-term and that determinative atoms only occur in $\ell$-terms.
Also observe that the evaluation result of an $\ell$-term is that of its determinative
atom. This is intuitively convincing, because the remainder of the atoms in any
$\ell$-term are non-determinative and hence do not contribute to its evaluation result. 
The non-determinative atoms that may occur before the first determinative atom are
encoded as a $\tr$-term. A $\tr$-$*$-term is the conjunction of a
$\tr$-term encoding such atoms and a $*$-term, which contains only
conjunctions and disjunctions of $\ell$-terms. We could also have encoded such
atoms as an $\fa$-term and then taken the disjunction with a $*$-term to
obtain a term with the same semantics. We consider $\ell$-terms to be `basic'
in $*$-terms in the sense that they are the smallest grammatical unit that
influences the evaluation result of the $*$-term.

In the following, $P^\tr, P^\ell$, etc.~are used both to denote grammatical
categories and as variables for terms in those categories. The remainder of
this section is concerned with defining and proving correct the normalization
function $\nf: \FT \to \FNF$.  We will define $\nf$ recursively using the
functions
\begin{equation*}
\nf^n: \FNF \to \FNF \quad\text{and}\quad
\nf^c: \FNF \times \FNF \to \FNF.
\end{equation*}
The first of these will be used to rewrite negated $\FNF$-terms to $\FNF$-terms
and the second to rewrite the conjunction of two $\FNF$-terms to an
$\FNF$-term. By \eqref{FFEL2} we have no need for a dedicated function that
rewrites the disjunction of two $\FNF$-terms to an $\FNF$-term.

We start by defining $\nf^n$. Analysing the semantics of $\tr$-terms and
$\fa$-terms together with the definition of $\FE$ on negations, it becomes
clear that $\nf^n$ must turn $\tr$-terms into $\fa$-terms and vice versa.
We also remark that $\nf^n$ must preserve the left-associativity of the
$*$-terms in $\tr$-$*$-terms, modulo the associativity within $\ell$-terms.
We define $\nf^n: \FNF \to \FNF$ as follows, using the auxiliary function
$\nf^n_1: P^* \to P^*$ to `push down' or `push in' the negation symbols when
negating a $\tr$-$*$-term. We note that there is no ambiguity between the
different grammatical categories present in an $\FNF$-term, i.e., any
$\FNF$-term is in exactly one of the grammatical categories identified in
Definition \ref{def:fnf}.
\begin{align}
\nf^n(\tr) &= \fa
  \label{eq:nfn1} \\
\nf^n(a \fullor P^\tr) &= a \fulland \nf^n(P^\tr)
  \label{eq:nfn2} \\
\nf^n(\fa) &= \tr
  \label{eq:nfn3} \\
\nf^n(a \fulland P^\fa) &= a \fullor \nf^n(P^\fa).
  \label{eq:nfn4} \\
\nf^n(P^\tr \fulland Q^*) &= P^\tr \fulland \nf^n_1(Q^*)
  \label{eq:nfn5} \\ \displaybreak[0]
\nf^n_1(a \fulland P^\tr) &= \neg a \fulland P^\tr
  \label{eq:nfn6} \\
\nf^n_1(\neg a \fulland P^\tr) &= a \fulland P^\tr
  \label{eq:nfn7} \\
\nf^n_1(P^* \fulland Q^d) &= \nf^n_1(P^*) \fullor \nf^n_1(Q^d)
  \label{eq:nfn8} \\
\nf^n_1(P^* \fullor Q^c) &= \nf^n_1(P^*) \fulland \nf^n_1(Q^c)
  \label{eq:nfn9}
\end{align}

Now we turn to defining $\nf^c$. These definitions have a great deal of
inter-dependence so we first present the definition for $\nf^c$ when the first
argument is a $\tr$-term. We see that the conjunction of a $\tr$-term with
another term always yields a term of the same grammatical category as the
second conjunct.
\begin{align}
\nf^c(\tr, P) &= P
  \label{eq:nfc1} \\
\nf^c(a \fullor P^\tr, Q^\tr) &= a \fullor \nf^c(P^\tr, Q^\tr)
  \label{eq:nfc2} \\
\nf^c(a \fullor P^\tr, Q^\fa) &= a \fulland \nf^c(P^\tr, Q^\fa)
  \label{eq:nfc3} \\
\nf^c(a \fullor P^\tr, Q^\tr \fulland R^*) &= \nf^c(a \fullor P^\tr,
  Q^\tr) \fulland R^*
  \label{eq:nfc4}
\end{align}

For defining $\nf^c$ where the first argument is an $\fa$-term we make use
of \eqref{FFEL7} when dealing with conjunctions of $\fa$-terms with
$\tr$-$*$-terms. The definition of $\nf^c$ for the arguments used in the
right hand side of \eqref{eq:nfc7} starts at \eqref{eq:nfc14}. We note that
despite the high level of inter-dependence in these definitions, this does not
create a circular definition. We also note that the conjunction of an
$\fa$-term with another term is always itself an $\fa$-term. 
\begin{align}
\nf^c(\fa, P^\tr) &= \nf^n(P^\tr)
  \label{eq:nfc5} \\
\nf^c(\fa, P^\fa) &= P^\fa
  \label{eq:nfc6} \\
\nf^c(\fa, P^\tr \fulland Q^*) &= \nf^c(P^\tr \fulland Q^*, \fa)
  \label{eq:nfc7} \\
\nf^c(a \fulland P^\fa, Q) &= a \fulland \nf^c(P^\fa, Q)
  \label{eq:nfc8}
\end{align}

The case where the first conjunct is a $\tr$-$*$-term is the most
complicated. Therefore we first consider the case where the second conjunct is
a $\tr$-term. In this case we must make the $\tr$-term part of the last
(rightmost) $\ell$-term in the $\tr$-$*$-term, so that the result will
again be a $\tr$-$*$-term. For this `pushing in' of the second conjunct we
define an auxiliary function $\nf^c_1: P^* \times P^\tr \to P^*$.
\begin{align}
\nf^c(P^\tr \fulland Q^*, R^\tr) &= P^\tr \fulland \nf^c_1(Q^*, R^\tr)
  \label{eq:nfc9} \\
\nf^c_1(a \fulland P^\tr, Q^\tr) &= a \fulland \nf^c(P^\tr, Q^\tr)
  \label{eq:nfc10} \\
\nf^c_1(\neg a \fulland P^\tr, Q^\tr) &= \neg a \fulland \nf^c(P^\tr,
  Q^\tr)
  \label{eq:nfc11} \\
\nf^c_1(P^* \fulland Q^d, R^\tr) &= P^* \fulland \nf^c_1(Q^d, R^\tr)
  \label{eq:nfc12} \\
\nf^c_1(P^* \fullor Q^c, R^\tr) &= P^* \fullor \nf^c_1(Q^c, R^\tr)
  \label{eq:nfc13}
\end{align}

When the second conjunct is an $\fa$-term, the result will naturally be an
$\fa$-term itself. So we need to convert the $\tr$-$*$-term to an
$\fa$-term. Using \eqref{FFEL4} we reduce this problem to converting a
$*$-term to an $\fa$-term, for which we use the auxiliary function
$\nf^c_2: P^* \times P^\fa \to P^\fa$.
\begin{align}
\nf^c(P^\tr \fulland Q^*, R^\fa) &= \nf^c(P^\tr, \nf^c_2(Q^*, R^\fa)) 
  \label{eq:nfc14} \\
\nf^c_2(a \fulland P^\tr, R^\fa) &= a \fulland \nf^c(P^\tr, R^\fa)
  \label{eq:nfc15} \\
\nf^c_2(\neg a \fulland P^\tr, R^\fa) &= a \fulland \nf^c(P^\tr,
  R^\fa)
  \label{eq:nfc16} \\
\nf^c_2(P^* \fulland Q^d, R^\fa) &= \nf^c_2(P^*, \nf^c_2(Q^d, R^\fa))
  \label{eq:nfc17} \\
\nf^c_2(P^* \fullor Q^c, R^\fa) &= \nf^c_2(P^*, \nf^c_2(Q^c, R^\fa))
  \label{eq:nfc18}
\end{align}

Finally we are left with conjunctions and disjunctions of two
$\tr$-$*$-terms, thus completing the definition of $\nf^c$.  We use the
auxiliary function $\nf^c_3: P^* \times P^\tr \fulland P^* \to P^*$ to ensure
that the result is a $\tr$-$*$-term.
\begin{align}
\nf^c(P^\tr \fulland Q^*, R^\tr \fulland S^*) &= P^\tr \fulland 
  \nf^c_3(Q^*, R^\tr \fulland S^*)
  \label{eq:nfc19} \\
\nf^c_3(P^*, Q^\tr \fulland R^\ell) &= \nf^c_1(P^*, Q^\tr) \fulland R^\ell
  \label{eq:nfc20} \\
\nf^c_3(P^*, Q^\tr \fulland (R^* \fulland S^d)) &= \nf^c_3(P^*, Q^\tr
  \fulland R^*) \fulland S^d
  \label{eq:nfc21} \\
\nf^c_3(P^*, Q^\tr \fulland (R^* \fullor S^c)) &= \nf^c_1(P^*, Q^\tr)
  \fulland (R^* \fullor S^c)
  \label{eq:nfc22}
\end{align}

As promised, we now define the normalization function $\nf: \FT \to \FNF$
recursively, using $\nf^n$ and $\nf^c$, as follows.
\begin{align}
\nf(a) &= \tr \fulland (a \fulland \tr)
  \label{eq:nf1} \\
\nf(\tr) &= \tr
  \label{eq:nf2} \\
\nf(\fa) &= \fa
  \label{eq:nf3} \\
\nf(\neg P) &= \nf^n(\nf(P))
  \label{eq:nf4} \\
\nf(P \fulland Q) &= \nf^c(\nf(P), \nf(Q))
  \label{eq:nf5} \\
\nf(P \fullor Q) &= \nf^n(\nf^c(\nf^n(\nf(P)), \nf^n(\nf(Q))))
  \label{eq:nf6}
\end{align}

\begin{theorem}
\label{thm:nf}
For any $P \in \FT$, $\nf(P)$ terminates, $\nf(P) \in \FNF$ and $\EqFFEL \vdash
\nf(P) = P$.
\end{theorem}

In Appendix~\ref{app:A2} we first prove a number of lemmas showing that the
definitions $\nf^n$ and $\nf^c$ are correct and use those to prove the theorem.
The main reason to use a normalization function rather than a
term rewriting system to prove the correctness of $\FNF$ is that this relieves us 
of the need to prove the confluence of the
induced rewriting system, thus simplifying the proof.

\paragraph{Tree Structure.}
Below we prove that $\EqFFEL$ axiomatises $\FFEL$ by
showing that for $P \in \FNF$ we can invert $\FE(P)$. To do this we need to
prove several structural properties of the trees in the image of $\FE$. In the
definition of $\FE$ we can see how $\FE(P \fulland Q)$ is assembled from
$\FE(P)$ and $\FE(Q)$ and similarly for $\FE(P \fullor Q)$. To decompose
these trees we introduce some notation. The trees in the image of $\FE$
are all finite binary trees over $A$ with leaves in $\{\tr, \fa\}$, i.e.,
$\FE[\FT] \subseteq \T$. We will now also consider the set $\Tone$ of binary
trees over $A$ with leaves in $\{\tr, \fa, \triangle\}$. 
Similarly we consider $\Ttwo$, the set of binary
trees over $A$ with leaves in $\{\tr, \fa, \triangle_1, \triangle_2\}$. The $\triangle$,
$\triangle_1$ and $\triangle_2$ will be used as placeholders when composing or
decomposing trees.  Replacement of the leaves of trees in $\Tone$ and $\Ttwo$
by trees (either in $\T$, $\Tone$ or $\Ttwo$) is defined analogous to
replacement for trees in $\T$, adopting the same notational conventions.

For example we have by definition of $\FE$ that $\FE(P \fulland Q)$ can be
decomposed as
\begin{equation*}
\FE(P)\ssub{\tr}{\triangle_1}{\fa}{\triangle_2}
\ssub{\triangle_1}{\FE(Q)}{\triangle_2}{\FE(Q)\sub{\tr}{\fa}},
\end{equation*}
where $\FE(P)\ssub{\tr}{\triangle_1}{\fa}{\triangle_2} \in \Ttwo$ and $\FE(Q)$ and
$\FE(Q)\sub{\tr}{\fa}$ are in $\T$.  We note that this only works because
the trees in the image of $\FE$, or more general, in $\T$, do not contain any
triangles. Similarly, as we discussed previously, $\FE(P \fulland \fa) =
\FE(P)\sub{\tr}{\fa}$, which we can write as
$\FE(P)\sub{\tr}{\triangle}\sub{\triangle}{\fa}$. We start by analysing the
$\FE$-image of $\ell$-terms.

\begin{lemma}[Structure of $\ell$-terms]
\label{lem:litstf}
There is no $\ell$-term $P$ such that $\FE(P)$ can be decomposed as
$X\sub{\triangle}{Y}$ with $X \in \Tone$ and $Y \in \T$, where $X \neq \triangle$, but
does contain $\triangle$, and $Y$ contains occurrences of both $\tr$ and
$\fa$.
\end{lemma}
\begin{proof}
Let $P$ be some $\ell$-term. When we analyse the grammar of $P$ we find that
one branch from the root of $\FE(P)$ will only contain $\tr$ and not $\fa$
and the other branch vice versa.  Hence if $\FE(P) = X\sub{\triangle}{Y}$ and $Y$
contains occurrences of both $\tr$ and $\fa$, then $Y$ must contain the
root and hence $X = \triangle$.
\end{proof}

By definition a $*$-term contains at least one $\ell$-term and hence for any
$*$-term $P$, $\FE(P)$ contains both $\tr$ and $\fa$. The following lemma
provides the $\FE$-image of the rightmost $\ell$-term in a $*$-term to witness
this fact.

\begin{lemma}[Determinativeness]
\label{lem:pert}
For all $*$-terms $P$, $\FE(P)$ can be decomposed as $X\sub{\triangle}{Y}$ with $X
\in \Tone$ and $Y \in \T$ such that $X$ contains $\triangle$ and $Y = \FE(Q)$ for
some $\ell$-term $Q$. Note that $X$ may be $\triangle$. We will refer to $Y$ as the
witness for this lemma for $P$.
\end{lemma}
\begin{proof}
By induction on the complexity of $*$-terms $P$ modulo the complexity of
$\ell$-terms. In the base case $P$ is an $\ell$-term and $\FE(P) =
\triangle\sub{\triangle}{\FE(P)}$ is the desired decomposition by Lemma \ref{lem:litstf}.
For the induction we have to consider both $\FE(P \fulland Q)$ and $\FE(P
\fullor Q)$.

We treat only the case for $\FE(P \fulland Q)$, the case for $\FE(P \fullor Q)$
is analogous. Let $X\sub{\triangle}{Y}$ be the decomposition for $\FE(Q)$ which we
have by induction hypothesis. Since by definition of $\FE$ on ${\fulland}$ we
have
\begin{equation*}
\FE(P \fulland Q) = \FE(P)\ssub{\tr}{\FE(Q)}{\fa}{\FE(Q)
\sub{\tr}{\fa}},
\end{equation*}
we also have
\begin{align*}
\FE(P \fulland Q) &= \FE(P)\ssub{\tr}{X\sub{\triangle}{Y}}{\fa}{\FE(Q)
  \sub{\tr}{\fa}} \\
&= \FE(P)\ssub{\tr}{X}{\fa}{\FE(Q) \sub{\tr}{\fa}}\sub{\triangle}{Y},
\end{align*}
where the second equality is due to the fact that the only triangles in
\begin{equation*}
\FE(P)\ssub{\tr}{X}{\fa}{\FE(Q) \sub{\tr}{\fa}}
\end{equation*}
are those occurring in $X$. This gives our desired decomposition.
\end{proof}

The following lemma illustrates another structural property of trees in the
image of $*$-terms under $\FE$, namely that the left branch of any
determinative atom in such a tree is different from its right branch.

\begin{lemma}[Non-decomposition]
\label{lem:nondectf}
There is no $*$-term $P$ such that $\FE(P)$ can be decomposed as
$X\sub{\triangle}{Y}$ with $X \in \Tone$ and $Y \in \T$, where $X \neq \triangle$ and $X$
contains $\triangle$, but not $\tr$ or $\fa$.
\end{lemma}
\begin{proof}
By induction on $P$ modulo the complexity of $\ell$-terms. The base case covers
$\ell$-terms and follows immediately from Lemma \ref{lem:pert} ($\FE(P)$
contains occurrences of both $\tr$ and $\fa$) and Lemma~\ref{lem:litstf}
(no non-trivial decomposition exists that contains both). For the induction
we assume that the lemma holds for all $*$-terms with lesser complexity than $P
\fulland Q$ and $P \fullor Q$.

We start with the case for $\FE(P \fulland Q)$. Suppose for contradiction that
$\FE(P \fulland Q) = X\sub{\triangle}{Y}$ with $X \neq \triangle$ and $X$ not containing
any occurrences of $\tr$ or $\fa$. Let $R$ be a witness of Lemma
\ref{lem:pert} for $P$. Now note that $\FE(P \fulland Q)$ has a subtree
\begin{equation*}
R\ssub{\tr}{\FE(Q)}{\fa}{\FE(Q)\sub{\tr}{\fa}}.
\end{equation*}
Because $Y$ must contain both the occurrences of $\fa$ in the one branch
from the root of this subtree as well as the occurrences of $\FE(Q)$ in the
other (because they contain $\tr$ and $\fa$), Lemma \ref{lem:litstf}
implies that $Y$ must (strictly) contain $\FE(Q)$ and
$\FE(Q)\sub{\tr}{\fa}$. Hence there is a $Z \in \T$ such that $\FE(P) =
X\sub{\triangle}{Z}$, which violates the induction hypothesis. The case for $\FE(P
\fullor Q)$ proceeds analogously.
\end{proof}

We now arrive at two crucial definitions for our completeness proof. When
considering $*$-terms we already know that $\FE(P \fulland Q)$ can be
decomposed as
\begin{equation*}
\FE(P)\ssub{\tr}{\triangle_1}{\fa}{\triangle_2}
\ssub{\triangle_1}{\FE(Q)}{\triangle_2}{\FE(Q)\sub{\tr}{\fa}}.
\end{equation*}
Our goal now is to give a definition for a type of decomposition so that this
is the only such decomposition for $\FE(P \fulland Q)$. We also ensure that
$\FE(P \fullor Q)$ does not have a decomposition of that type, so that we can
distinguish $\FE(P \fulland Q)$ from $\FE(P \fullor Q)$.  Similarly, we define
another type of decomposition so that $\FE(P \fullor Q)$ can only be decomposed
as
\begin{equation*}
\FE(P)\ssub{\tr}{\triangle_1}{\fa}{\triangle_2}
\ssub{\triangle_1}{\FE(Q)\sub{\fa}{\tr}}{\triangle_2}{\FE(Q)}
\end{equation*}
and that $\FE(P \fulland Q)$ does not have a decomposition of that type.

\begin{definition}
The pair $(Y, Z) \in \Ttwo \times \T$ is a \textbf{candidate conjunction
decomposition (ccd)} of $X \in \T$, if 
\begin{itemize}
 \setlength\itemsep{0mm}
\item $X = Y\ssub{\triangle_1}{Z}{\triangle_2}{Z\sub{\tr}{\fa}}$,
\item $Y$ contains both $\triangle_1$ and $\triangle_2$,
\item $Y$ contains neither $\tr$ nor $\fa$, and
\item $Z$ contains both $\tr$ and $\fa$.
\end{itemize}
Similarly, $(Y, Z)$ is a \textbf{candidate disjunction decomposition (cdd)} of
$X$, if
\begin{itemize}
 \setlength\itemsep{0mm}
\item $X = Y\ssub{\triangle_1}{Z\sub{\fa}{\tr}}{\triangle_2}{Z}$,
\item $Y$ contains both $\triangle_1$ and $\triangle_2$,
\item $Y$ contains neither $\tr$ nor $\fa$, and
\item $Z$ contains both $\tr$ and $\fa$.
\end{itemize}
\end{definition}

The ccd and cdd are not necessarily the decompositions we are looking for,
because, for example, $\FE((P \fulland Q) \fulland R)$ has a ccd
$(\FE(P)\ssub{\tr}{\triangle_1}{\fa}{\triangle_2}, \FE(Q \fulland R))$, whereas the
decomposition we need is $(\FE(P \fulland
Q)\ssub{\tr}{\triangle_1}{\fa}{\triangle_2}, \FE(R))$. Therefore we refine these
definitions to obtain the decompositions we seek.

\begin{definition}
The pair $(Y, Z) \in \Ttwo \times \T$ is a \textbf{conjunction decomposition
(cd)} of $X \in \T$, if it is a ccd of $X$ and there is no other ccd $(Y', Z')$
of $X$ where the depth of $Z'$ is smaller than that of $Z$.  Similarly, $(Y,
Z)$ is a \textbf{disjunction decomposition (dd)} of $X$, if it is a cdd of $X$
and there is no other cdd $(Y', Z')$ of $X$ where the depth of $Z'$ is smaller
than that of $Z$.
\end{definition}

\begin{theorem}
\label{thm:cddd}
For any $*$-term $P \fulland Q$, i.e., with $P \in P^*$ and $Q \in P^d$, $\FE(P
\fulland Q)$ has the (unique) cd
\begin{equation*}
(\FE(P)\ssub{\tr}{\triangle_1}{\fa}{\triangle_2}, \FE(Q))
\end{equation*}
and no dd. For any $*$-term $P \fullor Q$, i.e., with $P \in P^*$ and $Q \in
P^c$, $\FE(P \fullor Q)$ has no cd and its (unique) dd is
\begin{equation*}
(\FE(P)\ssub{\tr}{\triangle_1}{\fa}{\triangle_2}, \FE(Q)).
\end{equation*}
\end{theorem}
\begin{proof}
We first treat the case for $P \fulland Q$ and start with cd. Note that $\FE(P
\fulland Q)$ has a ccd $(\FE(P)\ssub{\tr}{\triangle_1}{\fa}{\triangle_2}, \FE(Q))$
by definition of $\FE$ (for the first condition) and by Lemma~\ref{lem:pert}
(for the fourth condition). It is immediate that it satisfies the second and
third conditions. It also follows that for any ccd $(Y, Z)$ either $Z$ contains
or is contained in $\FE(Q)$, for suppose otherwise, then $Y$ will contain an
occurrence of $\tr$ or of $\fa$, namely those we know by
Lemma~\ref{lem:pert} that $\FE(Q)$ has. Therefore it suffices to show that
there is no ccd $(Y, Z)$ where $Z$ is strictly contained in $\FE(Q)$. Suppose
for contradiction that $(Y, Z)$ is such a ccd. If $Z$ is strictly contained in
$\FE(Q)$ we can decompose $\FE(Q)$ as $\FE(Q) = V\sub{\triangle}{Z}$ for some $V \in
\Tone$ that contains but is not equal to $\triangle$. By Lemma \ref{lem:nondectf}
this implies that $V$ contains $\tr$ or $\fa$. But then so does $Y$,
because
\begin{equation*}
Y = \FE(P)\ssub{\tr}{V\sub{\triangle}{\triangle_1}}{\fa}{V\sub{\triangle}{\triangle_2}},
\end{equation*}
and so $(Y, Z)$ is not a ccd for $\FE(P \fulland Q)$. Therefore
$(\FE(P)\ssub{\tr}{\triangle_1}{\fa}{\triangle_2}, \FE(Q))$ is the \emph{unique} cd
for $\FE(P \fulland Q)$.

Now for the dd. It suffices to show that there is no cdd for $\FE(P \fulland
Q)$. Suppose for contradiction that $(Y, Z)$ is a cdd for $\FE(P \fulland Q)$.
We note that $Z$ cannot be contained in $\FE(Q)$, for then by Lemma
\ref{lem:nondectf}, $Y$ would contain $\tr$ or $\fa$. So $Z$ (strictly)
contains $\FE(Q)$. But then because
\begin{equation*}
Y\ssub{\triangle_1}{Z\sub{\fa}{\tr}}{\triangle_2}{Z} = \FE(P \fulland Q),
\end{equation*}
we would have by Lemma \ref{lem:pert} that $\FE(P \fulland Q)$ does not contain
an occurrence of $\FE(Q)\sub{\tr}{\fa}$. But the cd of $\FE(P \fulland Q)$
tells us that it does, contradiction! Therefore there is no cdd, and hence no
dd, for $\FE(P \fulland Q)$. The case for $\FE(P \fullor Q)$ proceeds
analogously.
\end{proof}

At this point we have the tools necessary to invert $\FE$ on $*$-terms, at
least down to the level of $\ell$-terms. We note that we can easily detect if a
tree in the image of $\FE$ is in the image of $P^\ell$, because all leaves to
the left of the root are one truth value, while all the leaves to the right are
the other. To invert $\FE$ on $\tr$-$*$-terms we still need to be able to
reconstruct $\FE(P^\tr)$ and $\FE(Q^*)$ from $\FE(P^\tr \fulland Q^*)$. To
this end we define a $\tr$-$*$-decomposition.

\begin{definition}
The pair $(Y, Z) \in \Tone \times \T$ is a \textbf{$\tr$-$*$-decomposition
(tsd)} of $X \in \T$, if $X = Y\sub{\triangle}{Z}$, $Y$ does not contain $\tr$ or
$\fa$ and there is no decomposition $(V, W) \in \Tone \times \T$ of $Z$ such
that
\begin{itemize}
 \setlength\itemsep{0mm}
\item $Z = V\sub{\triangle}{W}$,
\item $V$ contains $\triangle$,
\item $V \neq \triangle$, and
\item $V$ contains neither $\tr$ nor $\fa$.
\end{itemize}
\end{definition}

\begin{theorem}
\label{thm:tsd}
For any $\tr$-term $P$ and $*$-term $Q$ the (unique) tsd of $\FE(P \fulland
Q)$ is 
\begin{equation*}
(\FE(P)\sub{\tr}{\triangle}, \FE(Q)).
\end{equation*}
\end{theorem}
\begin{proof}
First we observe that $(\FE(P)\sub{\tr}{\triangle}, \FE(Q))$ is a tsd because by
definition of $\FE$ on $\fulland$ we have $\FE(P)\sub{\tr}{\FE(Q)} = \FE(P
\fulland Q)$ and $\FE(Q)$ is non-decomposable by Lemma \ref{lem:nondectf}.

Suppose for contradiction that there is another tsd $(Y, Z)$ of $\FE(P
\fulland Q)$. Now $Z$ must contain or be contained in $\FE(Q)$ for otherwise
$Y$ would contain $\tr$ or $\fa$, i.e., the ones we know $\FE(Q)$ has by
Lemma \ref{lem:pert}.

If $Z$ is strictly contained in $\FE(Q)$, then $\FE(Q) = V\sub{\triangle}{Z}$ for
some $V \in \Tone$ with $V \neq \triangle$ and $V$ not containing $\tr$ or
$\fa$ (because then $Y$ would too).  But this violates Lemma
\ref{lem:nondectf}, which states that no such decomposition exists. If $Z$
strictly contains $\FE(Q)$, then $Z$ contains at least one atom from $P$.  But
the left branch of any atom in $\FE(P)$ is equal to its right branch and hence
$Z$ is decomposable. Therefore $(\FE(P)\sub{\tr}{\triangle}, \FE(Q))$ is the
\emph{unique} tsd of $\FE(P \fulland Q)$.
\end{proof}

\paragraph{Completeness.}
With the last two theorems we can prove completeness for
$\FFEL$. We define three auxiliary functions to aid in our definition of the
inverse of $\FE$ on $\FNF$. Let $\cd : \T \to \Ttwo \times \T$ be the function
that returns the conjunction decomposition of its argument, $\dd$ of the same
type its disjunction decomposition and $\tsd: \T \to \Tone \times \T$ its
$\tr$-$*$-decomposition. Naturally, these functions are undefined when their
argument does not have a decomposition of the specified type. Each of these
functions returns a pair and we will use $\cd_1$ ($\dd_1$, $\tsd_1$) to denote
the first element of this pair and $\cd_2$ ($\dd_2$, $\tsd_2$) to denote the
second element.

We define $\inv: \T \to \FT$ using the functions $\inv^\tr: \T \to \FT$ for
inverting trees in the image of $\tr$-terms and $\inv^\fa$, $\inv^\ell$
and $\inv^*$ of the same type for inverting trees in the image of
$\fa$-terms, $\ell$-terms and $*$-terms, respectively. These functions are
defined as follows.
\begin{align}
\inv^\tr(X) &=
  \begin{cases}
    \tr
      &\textrm{if $X = \tr$} \\
    a \fullor \inv^\tr(Y)
      &\textrm{if $X = Y \tlef a \trig Z$}
  \end{cases} \\
\intertext{We note that we might as well have used the right branch from the
root in the recursive case. 
We chose the left branch here to more closely 
mirror the definition of the corresponding function for Free short-circuit logic as
defined in~\cite{Stau,PS18}.
}
\inv^\fa(X) &=
  \begin{cases}
    \fa
      &\textrm{if $X = \fa$} \\
    a \fulland \inv^\fa(Z)
      &\textrm{if $X = Y \tlef a \trig Z$}
  \end{cases} \\
\intertext{Similarly, we could have taken the left branch in this case.}
\inv^\ell(X) &=
  \begin{cases}
    a \fulland \inv^\tr(Y) 
      &\textrm{if $X = Y \tlef a \trig Z$ for some $a \in A$} \\
      &\textrm{and $Y$ only has $\tr$-leaves} \\
    \neg a \fulland \inv^\tr(Z)
      &\textrm{if $X = Y \tlef a \trig Z$ for some $a \in A$} \\
      &\textrm{and $Z$ only has $\tr$-leaves}
  \end{cases} \\
\inv^*(X) &=
  \begin{cases}
    \inv^*(\cd_1(X)\ssub{\triangle_1}{\tr}{\triangle_2}{\fa}) \fulland
      \inv^*(\cd_2(X))
      &\textrm{if $X$ has a cd} \\
    \inv^*(\dd_1(X)\ssub{\triangle_1}{\tr}{\triangle_2}{\fa}) \fullor
      \inv^*(\dd_2(X))
      &\textrm{if $X$ has a dd} \\
    \inv^\ell(X)
      &\textrm{otherwise}
  \end{cases} \\
\intertext{We can immediately see how Theorem \ref{thm:cddd} will be used in
the correctness proof of $\inv^*$.}
\inv(X) &=
  \begin{cases}
    \inv^\tr(X)
      &\textrm{if $X$ has only $\tr$-leaves} \\
    \inv^\fa(X)
      &\textrm{if $X$ has only $\fa$-leaves} \\
    \inv^\tr(\tsd_1(X)\sub{\triangle}{\tr}) \fulland \inv^*(\tsd_2(X))
      &\textrm{otherwise}
  \end{cases}
\end{align}
Similarly, we can see how Theorem \ref{thm:tsd} is used in the correctness
proof of $\inv$. It should come as no surprise that $\inv$ is indeed correct
and inverts $\FE$ on $\FNF$.

\begin{theorem}
\label{thm:felinv}
For all $P \in \FNF$, $\inv(\FE(P))=P$, i.e. $\inv(\fe(P))$ is syntactically equal to $P$ for $P\in\FNF$.
\end{theorem}

The proof for this theorem can be found in Appendix~\ref{app:A2}. For the
sake of completeness, we separately state the completeness result below.

\begin{theorem}[Completeness]
\label{thm:felcpl}
For all $P, Q \in \FT$, if $\FFEL \vDash P = Q$ then $\EqFFEL \vdash
P = Q$.
\end{theorem}

\begin{proof}
It suffices to show that for $P,Q \in \FNF$, $\FE(P) = \FE(Q)$ implies $P
= Q$, i.e., $P$ and $Q$ are syntactically equal.
To see this suppose that $P'$ and $Q'$ are two $\FEL$-terms and
$\FE(P') = \FE(Q')$. By Theorem~\ref{thm:nf}, $P'$ is derivably equal to an $\FNF$-term
$P$, i.e., $\EqFFEL \vdash P' = P$, and $Q'$ is derivably equal to an
$\FNF$-term $Q$. Lemma \ref{la:Fsound} then
gives us $\FE(P') = \FE(P)$ and $\FE(Q') = \FE(Q)$, and thus $\fe(P)=\fe(Q)$. Hence 
by Theorem \ref{thm:felinv}, 
$P= Q$, so in particular $\EqFFEL \vdash P = Q$. Transitivity then gives
us $\EqFFEL \vdash P' = Q'$ as desired.
\end{proof}

\section{Free FEL with undefinedness: \FFELu}
\label{sec:3}
In this section we define \FFELu, the extension of \FFEL\ with the truth value \undefi,
for which we use the constant \und. 
Evaluation trees with undefinedness were introduced in~\cite{BPS21}, where they serve as 
a semantics for short-circuit logic with undefinedness: the evaluation/interpretation of each 
atom may be undefined, i.e.\ not yield a classical truth value (\true\ or \false).
Well-known three-valued extensions of propositional logic 
are Kleene's `strong' three-valued logic~\cite{Kle38},
in which evaluation is executed in parallel so that $\fa\wedge x=x\wedge\fa=\fa$, 
Bochvar's `strict' logic~\cite{Boc38} with a constant $\mathsf N$ for `nonsense' or `meaningless', in which an expression has the value $\mathsf N$ as soon as it has a component 
with that value, and McCarthy's `sequential' logic~\cite{McC63}, in which evaluation proceeds sequentially 
from left to right so that $\fa\wedge\und=\fa$ and $\und\wedge\fa=\und$. 
We refer to Bergstra, Bethke, and Rodenburg~\cite[Sect.2]{BBR95} for a brief discussion of these logics.
Here we provide equational axioms for the equality of evaluation trees for full left-sequential evaluation
with undefinedness and prove a completeness result.

\begin{definition}
\label{def:treesU}
The set \NTu\ of \textbf{$\bm\und$-evaluation trees} over $A$ with leaves in 
$\{\tr, \fa,\und\}$ is defined inductively by
\[\tr\in\NTu,\quad\fa\in\NTu, \quad \und\in\NTu, \quad(X\unlhd \underline a\unrhd Y)\in\NTu 
~\text{ for any }X,Y \in \NTu \text{  and } a\in A.\]
The operator $\_\unlhd \underline a\unrhd\_$ is called 
\textbf{$\bm\und$-tree composition over $a$}.
In the evaluation tree $X \unlhd \underline a \unrhd Y$, 
the root is represented by $\underline a$,
the left branch by $X$, the right branch by $Y$, and the underlining of the root
represents a middle branch to the leaf \und. 
\end{definition}

Next to the formal notation for evaluation
trees we again introduce a more pictorial representation. For example,
the tree
\[(\fa\unlhd \underline a\unrhd\fa)\unlhd \underline a\unrhd(\tr\unlhd \underline a\unrhd\fa)\]
can be represented as follows, where $\unlhd$ yields a left branch, and $\unrhd$ a right branch:
\begin{equation}
\label{plaatje2}
\tag{Picture 2}
\hspace{-12mm}
\begin{tikzpicture}[%
      level distance=7.5mm,
      level 1/.style={sibling distance=15mm},
      level 2/.style={sibling distance=7.5mm},
      baseline=(current bounding box.center)]
      \node (root) {$a$}
        child {node (lef1) {$a$}
          child {node (lef2) {$\fa$}} 
          child {node (mid2) {$\und$}} 
          child {node (rig2) {$\fa$}}
        }
        child {node (mid1) {$\und$}
        }
        child {node (rig1) {$a$}
          child {node (lef2) {$\tr$}} 
          child {node (mid2) {$\und$}} 
          child {node (rig2) {$\fa$}}
        };
      \end{tikzpicture}
\end{equation}

We extend the set \SP\ to \SPu\ of (left-sequential) propositional expressions over $A$
with \und\ by the following grammar ($a\in A$):
\[
P ::= \tr\mid\fa\mid \und\mid a\mid\neg P\mid P\fulland P\mid P\fullor P
\]
and refer to its signature by
\(\SigFELu=\{\fulland,\fullor,\neg,\tr,\fa,\und,a\mid a\in A\}.\)

We interpret propositional expressions in \SPu\ as evaluation trees
by extending the function \fe\ (Definition~\ref{def:fe}).

\begin{definition}
\label{def:feu}
The unary \textbf{full evaluation function} $\feu : \SPu \to\NTu$ 
is defined as follows, where $a\in A$:
\begin{align*}
\feu(\tr) &= \tr,~~\feu(\fa) = \fa,
&\feu(\neg P)&=\feu(P)[\tr\mapsto \fa,\fa\mapsto \tr],\\
\feu(\und) &= \und,
&\feu(P \fulland Q)&= \feu(P)[\tr\mapsto \feu(Q),\fa\mapsto\feu(Q)[\tr\mapsto\fa]],\\
\feu(a)&=\tr\unlhd \underline a\unrhd \fa,
&\feu(P \fullor Q)&= \feu(P)[\tr\mapsto\feu(Q)[\fa\mapsto\tr],\fa\mapsto \feu(Q)].
\end{align*}
\end{definition}

Two examples: the evaluation trees $\feu(a\fulland \und)=\und\unlhd\underline a\unrhd \und$ and 
$\feu((a\fullor\tr)\fulland b)=(\tr\unlhd \underline b\unrhd \fa)
\\
\unlhd \underline a\unrhd (\tr\unlhd \underline b\unrhd \fa)$ 
 can be depicted as follows:
\[
\begin{array}{ll}
\begin{array}{l}
\begin{tikzpicture}[%
level distance=7.5mm,
level 1/.style={sibling distance=15mm},
level 2/.style={sibling distance=7.5mm},
level 3/.style={sibling distance=3.75mm}
]
\node (root) {$a$}
  child {node (lef1) {$\und$}
  }
  child {node (mid1) {$\und$}
  }
  child {node (rig1) {$\und$}
  };
\end{tikzpicture}
\\[8mm]
\end{array}
\qquad
&
\begin{array}{l}
\begin{tikzpicture}[%
level distance=7.5mm,
level 1/.style={sibling distance=15mm},
level 2/.style={sibling distance=7.5mm},
level 3/.style={sibling distance=3.75mm}
]
\node (root) {$a$}
  child {node (lef1) {$b$}
    child {node (lef2) {$\tr$}
    }
    child {node (mid2) {$\und$}
    }
    child {node (rig2) {$\fa$}
    }
  }
  child {node (mid1) {$\und$}
  }
  child {node (rig1) {$b$}
    child {node (riglef2) {$\tr$}
    }
    child {node (rigmid2) {$\und$}
    }
    child {node (rigrig2) {$\fa$}
    }
  };
\end{tikzpicture}
\end{array}
\end{array}
\]

\begin{definition}
\label{def:congruenceu}
The binary relation $=_{\feu}$ 
on \SPu\ defined by $P=_{\feu} Q\iff\feu(P)=\feu(Q)$
is called \textbf{free full $\bm\und$-valuation congruence}.
\end{definition}

\begin{lemma}
\label{la:Fcongruenceu}
The relation $=_{\feu}$ is  a congruence.
\end{lemma}

\begin{proof}
Cf.~Lemma~\ref{la:Fcongruence}.
\end{proof}

\begin{definition}
\label{def:FFELu}
A \textbf{Fully Evaluated Left-Sequential Logic with undefinedness \emph{(\FELu)}} is a logic that 
satisfies the consequences of \feu-equality. 
\textbf{Free Fully Evaluated Left-Sequential Logic with undefinedness \emph{(\FFELu)}} is the fully evaluated 
left-sequential logic with undefinedness that satisfies no more consequences than those of 
\feu-equality, i.e., for all $P, Q \in\SPu$,
\[
\FELu\models P=Q \Longleftarrow P=_{\feu} Q \quad\text{and}\quad \FFELu\models P=Q \iff P=_{\feu} Q.
\]
\end{definition}

In order to axiomatise $\FFELu$, we extend \FFELe\ (Table~\ref{tab:FFELe}) as follows: 
\[
\FFELeu=\FFELe \cup\{\neg\und=\und,~\und\fulland x= \und\}.
\]

\begin{lemma}[Soundness]
\label{la:Fsoundu}
For all $P, Q\in\SPu$, $\FFELeu\vdash P =Q ~\Longrightarrow~ \FFELu\models P = Q$.
\end{lemma}

\begin{proof}
By Lemma~\ref{la:Fcongruenceu}, the relation $=_{\feu}$ is a congruence on \SPu, 
so it suffices to show that all closed instances of the \FFELeu-axioms satisfy $=_{\feu}$,
which follows easily (cf.\ the proof of Lemma~\ref{la:Fsound}).
The validity of the two new axioms of $\FFELeu$ is also easily verified.
\end{proof}

Defining $\und^{dl}=\und$, it also immediately follows that \FFELeu\ satisfies the duality principle.
Furthermore, defining $x^{dl} = x$
for each variable $x$, the duality principle extends to equations, 
that is, for all terms $s, t$ over $\SigFELu$,
\[
\FFELeu\vdash s = t \iff \FFELeu\vdash s^{dl} = t^{dl}.
\]

We start with some useful consequences of \FFELeu.

\begin{lemma}
\label{la:aux0}
The following equations are consequences of \FFELeu:
\textup{
\begin{enumerate}
\setlength\itemsep{5pt}
\item
$x\fullor (y\fulland\und)=(x\fullor y)\fulland\und$, 
\item 
$x\fullor (y\fulland\und)=x\fulland (y\fulland\und)$,
\item 
$\neg x\fulland (y\fulland\und)=x\fulland (y\fulland\und)$.
\end{enumerate}
}
\end{lemma}

\begin{proof}
Note that by duality, $\und\fullor x=\und$. We first derive two auxiliary results:
\begin{align}
\tag{Aux1}
\label{A.1}
x\fullor\und
&=
x\fullor(\und\fulland\fa)\stackrel{\eqref{FFEL10}}=x\fulland(\und\fullor\tr)
=x\fulland\und,
\\
\tag{Aux2}
\label{A.2}
\neg x \fulland\und
&=
\neg x\fulland(\und\fulland\fa)
=\neg x\fulland(\fa\fulland\und)
=(\neg x\fulland\fa)\fulland\und
=(x\fulland\fa)\fulland\und
=x\fulland\und.
\end{align}
Consequence 1: $ x\fullor (y\fulland\und)
\stackrel{\eqref{A.1}}=x\fullor (y\fullor\und)
=(x\fullor y)\fullor\und
\stackrel{\eqref{A.1}}=(x\fullor y)\fulland\und
$.
\\[2mm]
Consequence 2:
$x\fullor (y\fulland\und)=x\fullor ((y\fulland\und)\fulland\fa)
\stackrel{\eqref{FFEL10}}=x\fulland ((y\fulland\und)\fullor\tr)
\stackrel{\eqref{A.1}}=x\fulland ((y\fullor\und)\fullor\tr)
=x\fulland (y\fullor(\und\fullor\tr))
=x\fulland (y\fullor\und)
\stackrel{\eqref{A.1}}=x\fulland (y\fulland\und)
$.
\\[2mm]
Consequence 3: by consequences $1$ and $2$, $(x\fullor y)\fulland\und
=x\fulland (y\fulland\und)
$, 
hence 
$\neg x\fulland(y\fulland\und)
=\neg x\fulland(\neg y\fulland\und)
=(\neg x\fulland \neg y)\fulland\und
\stackrel{\eqref{A.2}}=(x\fullor y)\fulland\und
=x\fulland (y\fulland\und)
$.
\end{proof}

We  introduce ``\und-normal forms" and prove a few lemmas from which our next 
completeness result follows easily.

\begin{definition}
\label{def:Unf}
For $\sigma\in A^*$ and $a\in A$, 
$\bm{\und_\sigma}$ is defined by 
$\und_\epsilon=\und$ and
$\und_{a\rho}=a\fulland\und_\rho.$
\end{definition}

\begin{lemma}
\label{la:auxA} For all $\sigma\in A^*$,
$\und_\sigma\fulland x=\und_\sigma$ and 
$\und_\sigma\fullor x=\und_\sigma$ are consequences of \FFELeu.
\end{lemma}

\begin{proof}
The first consequence follows easily by induction on the length of $\sigma$ and 
associativity \eqref{FFEL4},
as well as $(\und_\sigma)^{dl}=\und_\sigma$. By the latter,
$\und_\sigma\fullor x=\und_\sigma$.
\end{proof}

\begin{lemma}
\label{la:hulpje}
For all $\sigma\in A^*$, 
the following equations are consequences of $\FFELeu:$
\textup{
\begin{enumerate}
\item
$x\fullor (y\fulland\und_\sigma)=(x\fullor y)\fulland\und_\sigma$, 
\item
$x\fullor (y\fulland\und_\sigma)=x\fulland (y\fulland\und_\sigma)$,
\item
$\neg x\fulland (y\fulland\und_\sigma)=x\fulland (y\fulland\und_\sigma)$.
\end{enumerate}
}
\end{lemma}

\begin{proof}
By simultaneous induction on the length of $\sigma$.
The base case ($\sigma=\epsilon$) is Lemma~\ref{la:aux0}.1-3. 
For $\sigma=a\rho~(a\in A)$, derive
\\[1mm]
$x\fullor (y\fulland\und_{a\rho})
=x\fullor (y\fulland (a\fulland\und_{\rho}))
\stackrel{\text{IH2}}=x\fullor (y\fullor (a\fulland\und_{\rho}))
=(x\fullor y)\fullor (a\fulland\und_{\rho})
=(x\fullor y)\fulland\und_{a\rho}
$,
\\[1mm]
$x\fullor (y\fulland\und_{a\rho})
=x\fullor ((y\fulland a)\fulland\und_{\rho})
\stackrel{\text{IH2}}=x\fulland ((y\fulland a)\fulland\und_{\rho})
=x\fulland (y\fulland \und_{a\rho})$,
\\[1mm]
$\neg x \fulland(y\fulland\und_{a\rho})
=\neg x \fulland((y\fulland a)\fulland \und_{\rho})
\stackrel{\text{IH3}}=x \fulland((y\fulland a)\fulland \und_{\rho})
=x \fulland(y\fulland \und_{a\rho})
$.
\end{proof}

To prove the forthcoming completeness result, it suffices to restrict 
to Negation Normal Forms (NNFs), which are defined as follows ($a\in A$):
\begin{align*}
&P ::= \tr\mid\fa\mid\und\mid a\mid \neg a\mid P\fulland P\mid P\fullor P.
\end{align*}
By structural induction it immediately follows that for each  $P\in\SPu$, there is a unique NNF
$Q$ such that $\FFELeu\vdash P=Q$. 

\begin{lemma}
\label{la:aux1}
For each $P\in\SPu$ and $\sigma\in A^*$, there exists $\rho\in A^*$ s.t. 
$\FFELeu\vdash P\fulland\und_\sigma=\und_\rho$.
\end{lemma}

\begin{proof}
By structural induction, restricting to NNFs. The base cases $\fa\fulland\und_\sigma=\und_\sigma$
and $\neg a\fulland\und_\sigma=a\fulland \und_\sigma=\und_{a\sigma}$
follow by Lemma~\ref{la:hulpje}.3, and the other base cases are trivial.
For the induction there are two cases:

Case $P=Q\fulland R$. Derive $ (Q\fulland R)\fulland\und_\sigma=Q\fulland(R\fulland\und_\sigma)
\stackrel{\text{\text{IH}}}=Q\fulland\und_{\rho'}\stackrel{\text{IH}}=\und_\rho$ 
for some $\rho',\rho\in A^*$.

Case $P=Q\fullor R$. By Lemma~\ref{la:hulpje}.1 and induction, 
$ (Q\fullor R)\fulland\und_\sigma=Q\fullor (R\fulland\und_\sigma)=Q\fullor\und_{\rho'}$ 
for some $\rho'\in A^*$, and  
by Lemma~\ref{la:hulpje}.2 and induction, $ Q\fullor\und_{\rho'}=Q\fulland\und_{\rho'}=\und_\rho$
for some $\rho\in A^*$.
\end{proof}

\begin{lemma}
\label{la:aux2}
For each $P\in\SPu$ that contains \und, there exists 
$\sigma\in A^*$ s.t. $\FFELeu\vdash P=\und_\sigma$.
\end{lemma}

\begin{proof}
By structural induction, restricting to NNFs. The only base case is $P=\und=\und_\epsilon$.
For the induction there are two cases:

Case $P=Q\fulland R$. Apply a case distinction: 
if \und\ occurs in $Q$, then by induction, $ Q=\und_\sigma$
and by Lemma~\ref{la:auxA}, $ P=\und_\sigma\fulland R=\und_\sigma$;
if \und\ does not occur in $Q$, then \und\ occurs in $R$ and by induction, $ R=\und_\sigma$,
so $ P=Q\fulland\und_\sigma$. By Lemma~\ref{la:aux1}, $ Q\fulland\und_\sigma=\und_\rho$
for some $\rho\in A^*$.

Case $P=Q\fullor R$. Apply a case distinction: 
if \und\ occurs in $Q$, then by induction, $ Q=\und_\sigma$, and 
by Lemma~\ref{la:auxA}, $ P=\und_\sigma\fullor R=\und_\sigma$; 
if \und\ does not occur in $Q$, then \und\ occurs in $R$ and by induction, $ R=\und_\sigma$, 
so,  $ P=Q\fullor\und_\sigma$.
By Lemma~\ref{la:hulpje}.2 and Lemma~\ref{la:aux1}, 
$ Q\fullor\und_\sigma=Q\fulland\und_\sigma=\und_\rho$
for some $\rho\in A^*$.
\end{proof}

\begin{theorem}[Completeness]
The logic \FFELu\ is axiomatised by \FFELeu.
\end{theorem}

\begin{proof} 
By Lemma~\ref{la:Fsoundu}, \FFELeu\ is sound.
For completeness, assume $P_1=_{\feu} P_2$. Then, either $P_1$ and $P_2$ do not contain \und\
and by Theorem~\ref{thm:felcpl} 
we are done, or both $P_1$ and $P_2$ contain \und.
By Lemma~\ref{la:aux2}, there are $\und_{\sigma_i}\in A^*$ such that $\FFELeu\vdash P_i=\und_{\sigma_i}$. 
By assumption and soundness, $\sigma_1=\sigma_2$.
Hence $\FFELeu\vdash P_1=\und_{\sigma_1}=P_2$.
\end{proof}

\section{Memorising FEL (\MFEL)} 
\label{sec:4}
In this section we define and axiomatise “Memorising FEL” (\MFEL). As noted in the Introduction, 
the logic \MFEL\ is based on the evaluation strategy in which the evaluation of each atom 
is memorised: 
no complete path in a memorising evaluation tree contains multiple occurrences of the same atom.
Memorising evaluation trees appear to be fundamental to stronger logics  
such as \CLFELtwo\ (Section~\ref{sec:6}) and several short-circuit logics defined in~\cite{BPS21,BP23}.
We present equational axioms for the equality of memorising evaluation trees and prove
a completeness result.

\begin{definition}
\label{def:treesM}
The evaluation trees $\tr,\fa\in\NT$ are \textbf{memorising evaluation trees}.
\\
The evaluation tree $(X\unlhd a\unrhd Y)\in\NT$ is a \textbf{memorising evaluation tree} over $A$ 
if both $X$ and $Y$ are memorising evaluation trees that do not contain the label $a$.
\end{definition}

We interpret propositional expressions in \SP\ as memorising evaluation trees
by a function \mfe.

\begin{definition}
\label{def:mfe}
The unary \textbf{memorising full evaluation} function $\mfe : \SP \to\NT$ 
is defined by 
\[
\mfe(P)=\memt(\fe(P)),
\]
where the auxiliary function $\memt:\NT\to\NT$ is defined as
follows, for $a\in A$:
\[
\memt(\tr) = \tr,~\memt(\fa)=\fa,~\memt(X\unlhd a\unrhd Y)=\memt(\Le_a(X))\unlhd a\unrhd \memt(\Ri_a(Y)),
\] 
and the auxiliary functions $\Le_a,\Ri_a:\NT\to\NT$ are defined as follows, for $b\in A$:
\begin{align*}
&\Le_a(\tr) = \Ri_a(\tr)=\tr, ~~\Le_a(\fa)=\Ri_a(\fa)=\fa,
\\[1mm]
&\Le_a(X\unlhd b\unrhd Y)=
\begin{cases}
\Le_a(X) 
& \text{if $b=a$},\\
\Le_a(X)\unlhd b\unrhd \Le_a(Y) 
& \text{otherwise},
\end{cases}
\\[1mm]
&\Ri_a(X\unlhd b\unrhd Y)=
\begin{cases}
\Ri_a(Y) 
& \text{if $b=a$},\\
\Ri_a(X)\unlhd b\unrhd \Ri_a(Y) 
& \text{otherwise}.
\end{cases}
\end{align*}
\end{definition}

Let $A^s$ be the set of strings over $A$ with the property 
that each $\sigma\in A^s$ contains no multiple occurrences of the same atom.
Hence, for all $P\in\SP$, each path in $\mfe(P)$ from the root 
to a leaf has a sequence of labels of the form $\sigma\in A^s$.
It is clear that not all memorising evaluation trees can be expressed in \MFEL, a simple example
is $(\tr\unlhd b\unrhd\fa)\unlhd a\unrhd(\tr\unlhd c\unrhd\fa)$.
Three examples of memorising evaluation trees in $\mfe(\SP)$:
\begin{align*}
&\mfe(a\fulland a)=\memt(\fe(a\fulland a))=\memt(\fe(a)[\tr\mapsto\fe(a),\fa\mapsto\fe(a)])
\\
&\hspace{17.2mm}
=\memt((\tr\unlhd a\unrhd\fa)\unlhd a\unrhd(\tr\unlhd a\unrhd\fa))
=\memt(\Le_a(\tr))\unlhd a\unrhd\memt(\Ri_a(\fa))
=\fe(a), 
\\[2mm]
&\mfe(a\fullor \neg a)=\tr\unlhd a\unrhd\tr=\fe(a\fullor\tr),
\\[2mm]
&\mfe((a\fulland b)\fullor(\neg a\fulland\neg b))=(\tr\unlhd b\unrhd\fa)\unlhd a\unrhd
(\fa\unlhd b\unrhd\tr).
\end{align*}
The last example suggests that all memorising evaluation trees 
with complete traces $abB_i$ with $i=1,..,4$ and $B_i\in\{\tr,\fa\}$ can be expressed in \MFEL\ and
below we show that each such evaluation tree can be expressed
as $\mfe((a\fulland t_1)\fullor(\neg a\fulland t_2))$ with 
$t_j\in\{b,\neg b, b\fulland\fa,b\fullor \tr\}$. We note that this last evaluation tree cannot be expressed
in \FFEL\ and return to expressiveness issues in Section~\ref{sec:8}. 

\begin{definition}
\label{def:Mcongruence}
The binary relation $=_{\mfe}$ 
on \SP\
is called \textbf{memorising full valuation congruence} 
and is defined by $P=_{\mfe} Q\iff\mfe(P)=\mfe(Q)$.
\end{definition}

The following lemma is an adaptation of~\cite[La.3.5]{BPS21}, a 
detailed proof is given in Appendix~\ref{app:A4}.

\begin{lemma}
\label{la:Mcongruence} 
The relation $=_{\mfe}$ is a congruence.
\end{lemma}

\begin{definition}
\label{def:MFEL}
\textbf{Memorising Fully Evaluated Left-Sequential Logic \emph{(\MFEL)}} is the fully evaluated 
left-sequential logic that satisfies no more consequences than those of \mfe-equality, i.e., 
for all $P, Q \in\SP$,
\[
\MFEL\models P=Q \iff P=_\mfe Q.
\]
\end{definition}

In Table~\ref{tab:MFELe} we provide a set \MFELe\ of axioms for \MFEL\ that is an extension
of \FFELe\ with the axiom~\eqref{M1}. In the remainder of this section
we show that \MFELe\ axiomatises \MFEL.

\begin{lemma}[Soundness]
\label{la:Msound}
For all $P, Q\in\SP$, $\MFELe\vdash P =Q ~\Longrightarrow~ \MFEL\models P = Q$.
\end{lemma}

\begin{proof}
By Lemma~\ref{la:Mcongruence}, the relation $=_\mfe$ is a congruence on \SP, so it suffices to show
that all closed instances of the \MFELe-axioms satisfy $=_\mfe$.
By Lemma~\ref{la:Fsound} (soundness of \FFEL), we only have to prove this for axiom~\eqref{M1}, 
thus for all $P,Q,R \in\SP$, 
\[\memt(\fe((P \fullor Q)\fulland R)) = \memt(\fe((\neg P\fulland(Q\fulland R))\fullor 
(P\fulland R))). 
\]
We prove this in detail in Appendix~\ref{app:A4}.  
\end{proof}

\begin{table}[t]
{
\centering
\rule{1\textwidth}{.4pt}
\begin{align}
\nonumber
&\text{Import: }\quad\FFELe \text{ (Table~\ref{tab:FFELe})}
\\[2mm]
\label{M1}
\tag{M1}
&(x\fullor y) \fulland z
=(\neg x\fulland(y\fulland z))\fullor (x\fulland z)
\end{align}
\hrule
}
\caption{\MFELe, axioms for \MFEL}
\label{tab:MFELe}
\end{table}

\begin{table}
{
\centering
\rule{1\textwidth}{.4pt}
\begin{align}
\label{C1}
\tag{C1}
x\fulland (y\fulland x)&=x\fulland y
\\
\label{C2}
\tag{C2}
(x\fulland y)\fullor x&=x\fulland (y\fullor x)
\\
\label{C3}
\tag{C3}
(x\fulland y) \fullor (\neg x\fulland z)
&=(\neg x \fullor y)\fulland (x\fullor z)
\\
\label{C4}
\tag{C4}
x\fulland(y\fullor z)
&=(x \fulland y)\fullor (x\fulland z)
\end{align}
\hrule
}
\caption{Consequences of \MFELe\
}
\label{tab:MFELeCons}
\end{table}

\begin{lemma}
\label{la:xx}
The equations in Table~\ref{tab:MFELeCons} are consequences of \MFELe.
\end{lemma}

\begin{proof}
First,
$x\fulland\fa=(x\fulland\fa)\fullor\fa\stackrel{\eqref{M1}'}=
(\neg x\fullor(\fa\fullor\fa))\fulland(x\fullor\fa)
=\neg x\fulland x$, where $\eqref{M1}'$ is the dual of
\eqref{M1}.
\\
\eqref{C1}: 
$x\fulland y= x\fulland(y\fullor\fa)
\stackrel{\eqref{C4}}=(x\fulland y)\fullor(x\fulland\fa)=(x\fulland y)\fullor(\neg x\fulland\fa)
\stackrel{\eqref{C3}}=(\neg x\fullor y)\fulland(x\fullor\fa)=(\neg x\fullor y)\fulland x\stackrel{\eqref{M1}}
=(x\fulland (y\fulland x))\fullor (\neg x\fulland x)
=(x\fulland (y\fulland x))\fullor (x\fulland \fa)
=x\fulland ((y\fulland x)\fullor  \fa)=x\fulland (y\fulland x)$.
\\
\eqref{C2}: $(x\fulland y)\fullor x\stackrel{\eqref{M1}'}= 
(\neg x\fullor (y\fullor x))\fulland(x\fullor x)\stackrel{\eqref{C3}}
=(x\fulland (y\fullor x))\fullor(x\fulland\fa)$
$\stackrel{\eqref{C4}}=x\fulland((y\fullor x)\fullor\fa)=x\fulland(y\fullor x)$. 

Each of \eqref{C3} and \eqref{C4} follows with \emph{Prover9} 
with (the default) options \texttt{lpo} and \texttt{unfold} in 1s,
but we could not find any short, handwritten proofs.
\end{proof}

In order to prove completenss of \MFELe, we use normal forms.
\begin{definition}
\label{def:Mnf} 
Let $\sigma\in A^s$ and $a\in A$. Then $P$ is a \textbf{$\bm\sigma$-normal form} if

$\sigma=\epsilon$ and $P\in\{\tr,\fa\}$, or if

$\sigma=a\rho$ and $P=(a\fulland P_1)\fullor (\neg a \fulland P_2)$ with both 
$P_1$ and $P_2$ $\rho$-normal forms.
\end{definition}

For $\sigma=a\rho\in A^s$, each $\sigma$-normal form 
$(a\fulland P_1)\fullor (\neg a \fulland P_2)$ yields the perfect $a\rho$-tree
with root $a$, left child $P_1$ and right child $P_2$.
Note that there are $2^{(2^{|\sigma|})}$ $\sigma$-normal forms.

We will often denote $\sigma$-normal forms by $P_\sigma, Q_\sigma, R_\sigma$, etc. 
In many coming proofs we apply induction on $P_\sigma$, that is, we induct on the length of 
$\sigma$, often using the template $\sigma=a\rho$ with $a\in A$ for the induction step(s).
We call this approach ``by induction on $\sigma$" for short.

\begin{definition} 
\label{def:MnfSpecial}
Let $\sigma\in A^s$.
The $\sigma$-normal forms $\bm{\tr_\sigma}$ and $\bm{\fa_\sigma}$ are defined by
\[\text{$\tr_\epsilon=\tr$, ~$\tr_{a\rho}=(a\fulland \tr_\rho)\fullor(\neg a\fulland\tr_\rho)$, 
~and~ 
$\fa_\epsilon=\fa$, ~$\fa_{a\rho}=(a\fulland\fa_\rho)\fullor(\neg a\fulland\fa_\rho)$.}
\]
We will also use terms $\bm{\widetilde \fa_\sigma}$ defined by $\widetilde \fa_\epsilon=\fa$ and 
$\widetilde\fa_{a\rho}=a\fulland\widetilde \fa_\rho$.
\end{definition}

\begin{lemma}
\label{la:nul}
For all $\sigma\in A^s$, $\MFELe\vdash \neg\fa_\sigma=\tr_\sigma,
~\widetilde\fa_\sigma=\fa_\sigma$.
If $\sigma=a\rho$, then  
\[\MFELe\vdash \widetilde\fa_\sigma=\neg a\fulland\widetilde\fa_\rho.
\]
\end{lemma}

\begin{proof}
The first consequence follows by induction on $\sigma$. If $\sigma=\epsilon$, this is trivial, 
and if $\sigma=a\rho$, then 
\[\neg\fa_{a\rho}=(\neg a\fullor\neg\fa_\rho)\fulland(a\fullor\neg\fa_\rho)
\stackrel{\eqref{C3},\text{IH}}=(a\fulland\tr_\rho)\fullor(\neg a\fulland\tr_\rho)=\tr_{a\rho}.
\]
The last consequence follows from 
associativity and axiom~\eqref{FFEL7}:
$\widetilde\fa_\rho= \widetilde\fa_\rho\fulland\fa=\fa\fulland\widetilde\fa_\rho$,
hence 
\[\widetilde\fa_{a\rho}
=(a\fulland(\fa\fulland\widetilde\fa_{\rho}))
=((a\fulland\fa)\fulland\widetilde\fa_{\rho}))
\stackrel{\eqref{FFEL8}}=((\neg a\fulland\fa)\fulland\widetilde\fa_{\rho}))
=(\neg a\fulland\widetilde\fa_{\rho}).
\]
Finally, $\widetilde\fa_\sigma=\fa_\sigma$ 
follows by induction on $\sigma$. 
If $\sigma=\epsilon$ this is immediate, and if $\sigma=a\rho$, 
\begin{align*}
\widetilde\fa_{a\rho}
&=(a\fulland \widetilde\fa_{\rho})\fullor(\neg a\fulland \widetilde\fa_{\rho})
&&\text{by idempotence and the last consequence}
\\
&=(a\fulland\fa_{\rho})\fullor(\neg a\fulland \fa_{\rho})
&&\text{by IH}
\\
&=\fa_{a\rho}.
\end{align*}
\end{proof}

\begin{lemma}
\label{la:hulp1}
Let $\sigma\in A^s$, then
$P_\sigma\fullor \fa_\sigma=P_\sigma$ and
$P_\sigma\fulland\tr_\sigma=P_\sigma$
are consequences of \MFELe.
\end{lemma}

\begin{proof}
By induction on $\sigma$
The two cases for $\sigma=\epsilon$ are immediate. If $\sigma=a\rho$, 
\begin{align*}
P_{a\rho}\fullor\fa_{a\rho}
&=((a\fulland Q_\rho)\fullor(\neg a\fulland R_\rho))\fullor(\neg a\fulland \widetilde\fa_\rho)
&&\text{by Lemma~\ref{la:nul}}\\
&=(a\fulland Q_\rho)\fullor((\neg a\fulland R_\rho)\fullor(\neg a\fulland \fa_\rho))
&&\text{by associativity and Lemma~\ref{la:nul}}\\
&=(a\fulland Q_\rho)\fullor(\neg a\fulland (R_\rho\fullor\fa_\rho))
&&\text{by \eqref{C4}}\\
&=(a\fulland Q_\rho)\fullor(\neg a\fulland R_\rho)
&&\text{by IH}\\
&=P_{a\rho}.
\end{align*}

Next, it follows from~\eqref{C3}
(in Table~\ref{tab:MFELeCons}) by induction on $\nu\in A^s$ that $\neg P_\nu$ is provably equal to a 
$\nu$-normal form. By Lemma~\ref{la:nul},
\(P_{a\rho}\fulland\tr_{a\rho}=\neg(\neg P_{a\rho}\fullor\fa_{a\rho})=P_{a\rho}.
\)
\end{proof}

We define an auxiliary opertor $\auxf(x,y,z)$ that preserves the 
property that in the evaluation of closed terms the left-right order is always respected
(modulo memorising occurrences of atoms).

\begin{definition}
\label{def:h}
The ternary operator $\bm{\auxf(x,y,z)}$ on terms over $\SigFEL$ is defined by
\[\auxf(x,y,z)=(x\fulland y)\fullor (\neg x\fulland z).
\]
\end{definition}

Hence, $\auxf(a,P_\sigma,Q_\sigma)$ is a $a\sigma$-normal form (adopting the notational convention that
$P_\sigma$ and $Q_\sigma$ are $\sigma$-normal forms). 
In particular, $\tr_{a\rho}=\auxf(a,\tr_\rho,\tr_\rho)$ and $\fa_{a\rho}=\auxf(a,\fa_\rho,\fa_\rho)$.

\begin{lemma}
\label{la:crux}
The following equations are consequences of \MFELe:
\begin{itemize}
\item[$1.$]
$\auxf(x,y,z) \fulland w =\auxf(x,(y\fullor (z\fulland\fa))\fulland w,~z\fulland w)$,
\item[$2.$]
$\auxf(x,y,z) \fullor w =\auxf(x,(y\fulland (z\fullor \tr))\fullor w,~z\fullor w)$.
\end{itemize}
\end{lemma}

\begin{proof}
Consequence 1 follows from consequence 2 and $\neg \auxf(x,y,z)=\auxf(x,\neg y,\neg z)$
(which follows easily from~\eqref{C3}).
Consequence 2 follows with \emph{Prover9} with options \texttt{lpo} and \texttt{unfold} in 2s.
\end{proof}

\begin{lemma}
\label{la:hulp2}
Let $\sigma\in A^s$, then
$P_\sigma\fulland\fa= \fa_\sigma$ and
$P_\sigma\fullor\tr=\tr_\sigma$
are consequences of \MFELe.
\end{lemma}

\begin{proof}
By induction on $\sigma$. The case $\sigma=\epsilon$ is immediate, and if $\sigma=a\rho$ then
\begin{align*}
P_{a\rho}\fulland\fa
&=\auxf(a,Q_\rho,R_\rho)\fulland\fa\\
&=\auxf(a, (Q_\rho\fullor(R_\rho\fulland\fa))\fulland\fa, ~R_\rho\fulland\fa)
&&\text{by Lemma~\ref{la:crux}.1}\\
&=\auxf(a, (Q_\rho\fullor\fa_\rho)\fulland\fa, ~\fa_\rho)
&&\text{by IH}\\
&=\auxf(a, Q_\rho\fulland\fa, ~\fa_\rho)
&&\text{by Lemma~\ref{la:hulp1}}\\
&=\auxf(a,\fa_{\rho},\fa_\rho)
&&\text{by IH}\\
&=\fa_{a\rho}.
\end{align*}

The case for $P_{a\rho}\fullor\tr$ follows in a similar way
(with help of Lemma~\ref{la:crux}.2).
\end{proof}

In order to compose $\sigma$-normal forms, we introduce some notation.

\begin{definition}[$str(P)$ and $\sigma\gg\rho$]
\label{def:mstring}
For $P\in\SP$, the string $\bm{str(P)}\in A^s$ is defined by $str(\tr)=str(\fa)=\epsilon$, 
$str(a)=a$ $(a\in A)$, $str(\neg P)=str(P)$, and
$str(P\fulland Q)=str(P\fullor Q)=str(P)\gg str(Q)$, where 
$
\_\gg\_:A^s\times A^s\to A^s
$
is the operation that filters out the atoms of the left argument in the right argument:  
\[
\text{
$\sigma\gg\epsilon=\sigma$ and $\sigma\gg a\rho=
\begin{cases}
\sigma\gg\rho
&\text{if $a$ occurs in $\sigma$},\\
\sigma a\gg\rho
&\text{otherwise}.
\end{cases}$}
\]
\end{definition}

Observe that $\sigma\gg a\rho=(\sigma\gg a)\gg\rho$ and $\epsilon\gg\rho=\rho$.

\begin{lemma}
\label{la:hulp}
In \MFELe, $\sigma$-normal forms in \SP\ are provably closed under $\fulland$ and $\fullor$ composition.
\end{lemma}

\begin{proof}
We first consider $P_\sigma\fulland R_\nu$ and prove this case by induction
on $\sigma$.

\noindent
For $\sigma=\epsilon$, $\tr\fulland R_\nu=R_\nu$ 
is immediate, and $\fa\fulland R_\nu=\fa_\nu$ 
follows from~\eqref{FFEL7} and Lemma~\ref{la:hulp2}.

\noindent
For $\sigma=a\rho$ it suffices to prove that 
\begin{equation*}
\auxf(a,P_\rho,Q_\rho)\fulland R_\nu= \auxf(a,~P_\rho\fulland R_\nu,~Q_\rho\fulland R_\nu)
\end{equation*}
because by induction, both $P_\rho\fulland R_\nu$ and $Q_\rho\fulland R_\nu$ have a 
provably equal $(\rho\gg\nu)$-normal form:
\begin{align*}
\auxf(a,P_\rho,Q_\rho)\fulland R_\nu
&=\auxf(a,(P_\rho\fullor(Q_\rho\fulland \fa))\fulland R_\nu,~Q_\rho\fulland R_\nu)
&&\text{by Lemma~\ref{la:crux}.1}\\
&=\auxf(a,(P_\rho\fullor\fa_\rho)\fulland R_\nu,~Q_\rho\fulland R_\nu).
&&\text{by Lemma~\ref{la:hulp2}}\\
&=\auxf(a,P_\rho\fulland R_\nu,~Q_\rho\fulland R_\nu).
&&\text{by Lemma~\ref{la:hulp1}}
\end{align*}

The case for $P_\sigma\fullor R_\nu$ follows in a similar way (with help of Lemma~\ref{la:crux}.2).
\end{proof}

\begin{lemma}
\label{la:MFELe}
For each $P\in\SP$ there is a unique $\sigma$-normal form
$Q$ such that $\MFELe\vdash P=Q$.
\end{lemma}

\begin{proof}
By structural induction on $P$, restricting to NNFs. For $P\in\{\tr,\fa,a,\neg a\mid a\in A\}$
this is trivial: $a=\auxf(a,\tr,\fa)$ and $\neg a = \auxf(a,\fa,\tr)$.
For the cases $P=P_1\fulland P_2$ and $P=P_1\fullor P_2$ this follows from Lemma~\ref{la:hulp}.

Uniqueness follows from the facts that $str(P_1\fulland P_2)$ is fixed and that 
syntactically different normal forms have different 
evaluation trees: for $F(x,y,z)=\auxf(y,x,z)$, their $F$-representation mimics the tree structure.
\end{proof}

\begin{theorem}[Completeness]
\label{thm:MFELe}
The logic \MFEL\ is axiomatised by \MFELe.
\end{theorem}

\begin{proof}
By Lemma~\ref{la:Msound}, \MFELe\ is sound. For completeness, assume $P_1=_{\mfe} P_2$.
By Lemma~\ref{la:MFELe} there are unique $\sigma_i$-normal forms $Q_{\sigma_i}$ such that $\MFELe\vdash 
P_i=_\mfe Q_{\sigma_i}$. By assumption and soundness, $Q_{\sigma_1}=Q_{\sigma_2}$. 
Hence, $\MFELe\vdash P_1=Q_{\sigma_1}=P_2$.
\end{proof}

\section{MFEL with undefinedness: \MFELu}
\label{sec:5}
In this section we define \MFELu, the extension of \MFEL\ with \und.
First, we formally define memorising \und-evaluation trees.
Given the detailed accounts in Sections~\ref{sec:3} and~\ref{sec:4}, this extension
is rather simple and straightforward. We define memorising \und-evaluation trees,
give equational axioms for the equality of these trees and prove
a completeness result.

\begin{definition}
\label{def:treesMu}
The evaluation trees $\tr,\fa,\und\in\NTu$ are \textbf{memorising $\bm\und$-evaluation trees}.
\\
The evaluation tree $(X\unlhd \underline a\unrhd Y)\in\NTu$ is a 
\textbf{memorising $\bm\und$-evaluation tree} over $A$ 
if both $X$ and $Y$ are memorising \und-evaluation trees that do not contain the label $\underline a$.
\end{definition}

We interpret propositional expressions in \SPu\ as memorising \und-evaluation trees
by extending the function \mfe\ (Definition~\ref{def:mfe}).

\begin{definition}
\label{def:mfeu}
The unary \textbf{memorising full evaluation} function $\mfeu : \SPu \to\NTu$ 
is defined by 
\[\mfeu(P)=\memtu(\feu(P)),
\]
where the auxiliary function $\memtu:\NTu\to\NTu$ is defined as
follows, for $a\in A$:
\[
\memtu(B) = B\text{ for } B\in\{\tr,\fa,\und\},~\memtu(X\unlhd \underline a\unrhd Y)
=\memtu(\Leu_a(X))\unlhd \underline a\unrhd \memtu(\Riu_a(Y)),
\] 
and the auxiliary functions $\Leu_a,\Riu_a:\NTu\to\NTu$ are defined as follows, for $b\in A$:
\begin{align*}
&\Leu_a(B) = \Riu_a(B)=B\text{ for } B\in\{\tr,\fa,\und\},
\\[1mm]
&\Leu_a(X\unlhd \underline b\unrhd Y)=
\begin{cases}
\Leu_a(X) 
& \text{if $b=a$},\\
\Leu_a(X)\unlhd \underline b\unrhd \Leu_a(Y) 
& \text{otherwise},
\end{cases}
\\[1mm]
&\Riu_a(X\unlhd \underline b\unrhd Y)=
\begin{cases}
\Riu_a(Y) 
& \text{if $b=a$},\\
\Riu_a(X)\unlhd \underline b\unrhd \Riu_a(Y) 
& \text{otherwise}.
\end{cases}
\end{align*}
\end{definition}

\begin{definition}
\label{def:Mcongruenceu} 
The binary relation $=_{\mfeu}$ 
on \SPu\
is called \textbf{memorising full $\bm\und$-valuation congruence} 
and is defined by $P=_{\mfeu} Q\iff\mfeu(P)=\mfeu(Q)$.
\end{definition}

\begin{lemma}
\label{la:Mcongruenceu} 
The relation $=_{\mfeu}$ is a congruence.
\end{lemma}

\begin{proof}
The proof of Lemma~\ref{la:Mcongruence} (in Appendix~\ref{app:A4}) requires one extra base case 
$X=\und$ for all sub-proofs, which follows trivially.
\end{proof}

\begin{definition}
\label{def:MFELu}
\textbf{Memorising Fully Evaluated Left-Sequential Logic with undefinedness \emph{(\MFEL$^{\bm\und}$)}} 
is the fully evaluated 
left-sequential logic with undefinedness that satisfies no more consequences than those of \mfeu-equality, i.e., 
for all $P, Q \in\SPu$,
\[
\MFELu\models P=Q \iff P=_{\mfeu} Q.
\]
\end{definition}

We extend \MFELe\ to \MFELeu\ with the two axioms
$\neg\und=\und$ ~and~ $\und\fulland x= \und$, so all
\FFELeu-results of Section~\ref{sec:3} hold in \MFELeu.

\begin{lemma}[Soundness]
\label{la:Msoundu}
For all $P, Q\in\SPu$, $\MFELeu\vdash P =Q ~\Longrightarrow~ \MFELu\models P = Q$.
\end{lemma}

\begin{proof}
By Lemma~\ref{la:Mcongruenceu}, the relation $=_{\mfeu}$ is a congruence on \SPu, 
so it suffices to show that all closed instances of the \MFELeu-axioms satisfy $=_{\mfeu}$,
which follows  easily (cf.\ the proof of Lemma~\ref{la:Msound} and, for the \und-axioms,
the proof of Lemma~\ref{la:Fsoundu}).
\end{proof}

In order to prove completeness of \MFELeu,  we use the $\sigma$-normal forms 
$\und_\sigma$ from Definition~\ref{def:Unf} for $\sigma\in A^s$,
thus $\und_\epsilon=\und$ and $\und_{a\rho}=a\fulland\und_\rho$.
By induction on $\sigma$ and  Lemma~\ref{la:hulpje}.2-3, 
\begin{equation}
\label{eq:eerst}
\tag{$\ddagger$}
\text{For all $\sigma\in A^s$,
$\MFELeu\vdash\neg\und_\sigma=\und_\sigma$.
}\end{equation}

We further use the ternary operator $\auxf(x,y,z)=(x\fulland y)\fullor (\neg x\fulland z)$  
(Definition~\ref{def:h}) on terms over $\SigFELu$.

\begin{lemma}
\label{la:Unf}
For each $a\in A$ and $\sigma=a\rho\in A^s$,
$\MFELeu\vdash \und_{a\rho}=\auxf(a,\und_\rho,\und_\rho)$.
\end{lemma}

\begin{proof}
Derive $a\fulland\und_\rho=(a\fulland\und_\rho)\fullor(a\fulland\und_\rho)
\stackrel{\text{La.\ref{la:hulpje}.3}}=(a\fulland\und_\rho)\fullor(\neg a\fulland\und_\rho)=
\auxf(a,\und_\rho,\und_\rho)$.
\end{proof}

In the following, we extend previous results from Section~\ref{sec:4} to $\SPu$.

\begin{lemma}[extension of La.\ref{la:hulp1}]
\label{la:hulp1u}
Let $\sigma\in A^s$, then $P_\sigma\fullor \und_\sigma=\und_\sigma$ and $P_\sigma\fulland
\und_\sigma=\und_\sigma$
are consequences of \MFELeu.
\end{lemma}
 
\begin{proof}
If $P_\sigma=\und_\sigma$, we are done by Lemma~\ref{la:auxA}. If $P_\sigma\in\SP$, the statement follows
by induction on $\sigma$:
if $\sigma= \epsilon$, this is trivial. If $\sigma=a\rho$,
\begin{align*}
P_{a\rho}\fullor\und_{a\rho}
&=((a\fulland Q_\rho)\fullor(\neg a\fulland R_\rho))\fullor(a\fulland \und_\rho)
\\
&=(a\fulland Q_\rho)\fullor((\neg a\fulland R_\rho)\fullor(\neg a\fulland \und_\rho))
&&\text{by associativity and Lemma~\ref{la:hulpje}.3}\\
&=(a\fulland Q_\rho)\fullor(\neg a\fulland (R_\rho\fullor\und_\rho))
&&\text{by \eqref{C4}}\\
&=(a\fulland Q_\rho)\fullor(a\fulland \und_\rho)
&&\text{by IH and Lemma~\ref{la:hulpje}.3}\\
&=a\fulland (Q_\rho\fullor \und_\rho)
&&\text{by \eqref{C4}}\\
&=\und_{a\rho}.
&&\text{by IH}
\end{align*}

In the proof of Lemma~\ref{la:hulp1} it was already argued that if $P_\sigma\in\SP$ then
$\neg P_\sigma$ is provably equal to a $\sigma$-normal form. By consequence~\eqref{eq:eerst},
$\neg \und_\sigma=\und_\sigma$. Hence, 
\(P_{\sigma}\fulland\und_{\sigma}=\neg(\neg P_{\sigma}\fullor\und_{\sigma})=\und_{\sigma}.
\)
\end{proof}

\begin{lemma}[extension of La.\ref{la:hulp2}]
\label{la:hulp2u}
Let $\sigma\in A^s$, then
$P_\sigma\fulland\und= \und_\sigma$ and
$P_\sigma\fullor\und=\und_\sigma$
are consequences of \MFELeu.
\end{lemma}

\begin{proof}
If $P_\sigma=\und_\sigma$, this follows by Lemma~\ref{la:auxA}. 
If $P_\sigma\in \SP$, this follows by induction on $\sigma$. 
We first consider the case $P_\sigma\fullor\und=\und_\sigma$.
If $\sigma=\epsilon$
the statement is trivial. If $\sigma=a\rho$ then 
\begin{align*}
P_{a\rho}\fullor\und
&=\auxf(a,Q_\rho,R_\rho)\fullor\und\\
&=\auxf(a, (Q_\rho\fulland(R_\rho\fullor\tr))\fullor\und, ~R_\rho\fullor\und)
&&\text{by {Lemma~\ref{la:crux}.2}}\\
&=\auxf(a, (Q_\rho\fulland\tr_\rho)\fullor\und, ~\und_\rho)
&&\text{by Lemma~\ref{la:hulp2} and IH}\\
&=\auxf(a, Q_\rho\fullor\und, ~\und_\rho)
&&\text{by Lemma~\ref{la:hulp1}}\\
&=\auxf(a,\und_{\rho},\und_\rho)
&&\text{by IH}\\
&=\und_{a\rho}.
&&\text{by Lemma~\ref{la:Unf}}
\end{align*}

The case for $P_{a\rho}\fulland\und$ follows in a similar way
(with help of {Lemma~\ref{la:crux}.1}).
\end{proof}

\begin{lemma}[extension of La.\ref{la:hulp}]
\label{la:hulpu}
In \MFELeu, $\sigma$-normal forms in \SPu\ are provably closed under 
${\fulland}$ and $\fullor$ composition, and each of these 
compositions determines a unique $\nu$-normal form.
\end{lemma}

\begin{proof}
We first consider $P_\sigma\fulland R_\nu$ and distinguish three cases.

Case
$P_\sigma=\und_\sigma$. By Lemma~\ref{la:auxA},
$\und_\sigma\fulland R_\nu =\und_\sigma$. 

Case
$P_\sigma,R_\nu\in\SP$. By Lemma~\ref{la:hulp}, there is a $(\sigma\gg \nu)$-normal form
s.t. $P_\sigma\fulland R_\nu=Q_{\sigma\gg \nu}$.

Case
$P_\sigma\in\SP$ and $R_\nu=\und_\nu$. We prove this case by induction
on $\sigma$: 
\\[1mm]
If $\sigma=\epsilon$, $\tr\fulland \und_\nu=\und_\nu$ is immediate, and 
$\fa\fulland \und_\nu=\und_\nu$ follows from~\eqref{FFEL7} and Lemma~\ref{la:auxA}.
\\[1mm]
If $\sigma=a\rho\in A^s$, it suffices to prove that 
\begin{equation*}
 \auxf(a,P_\rho,Q_\rho)\fulland \und_\nu= \auxf(a,~P_\rho\fulland \und_\nu,~Q_\rho\fulland \und_\nu)
\end{equation*}
because by induction, both $P_\rho\fulland \und_\nu$ and $Q_\rho\fulland \und_\nu$ have a 
provably equal $(\rho\gg\nu)$-normal form: 
\begin{align*}
\auxf(a,P_\rho,Q_\rho)\fulland \und_\nu
&=\auxf(a,(P_\rho\fullor(Q_\rho\fulland \fa))\fulland \und_\nu,~Q_\rho\fulland \und_\nu)
&&\text{by Lemma~\ref{la:crux}.1}\\
&=\auxf(a,(P_\rho\fullor\fa_\rho)\fulland \und_\nu,~Q_\rho\fulland \und_\nu).
&&\text{by Lemma~\ref{la:hulp2}}\\
&=\auxf(a,P_\rho\fulland \und_\nu,~Q_\rho\fulland \und_\nu).
&&\text{by Lemma~\ref{la:hulp1}}
\end{align*}

The proof for $P_\sigma\fullor R_\nu$ follows in a similar way (with help of Lemma~\ref{la:crux}.2).
\end{proof}

\begin{lemma}
\label{la:MFELeu}
For each $P\in\SPu$ there is a unique $\sigma$-normal form
$Q$ s.t. $\MFELeu\vdash P=Q$.
\end{lemma}

\begin{proof}
By structural induction on $P$, restricting to NNFs. For $P\in\{\tr,\fa,\und,a,\neg a\mid a\in A\}$
this is trivial: $a=\auxf(a,\tr,\fa)$ and $\neg a = \auxf(a,\fa,\tr)$.
For the cases $P=P_1\fulland P_2$ and $P=P_1\fullor P_2$ this follows from Lemma~\ref{la:hulpu}.

Uniqueness follows from the facts that for $\sigma$-normal form $P_1$ and $\rho$-normal form $P_2$, 
the normal forms of $P_1\fulland P_2$ and $P_1\fullor P_2$ are unique by Lemma~\ref{la:hulpu}, 
and that 
syntactically different normal forms have different 
evaluation trees: for $F(x,y,z)=\auxf(y,x,z)$, their $F$-representation mimics the tree structure.
\end{proof}

\begin{theorem}[Completeness]
\label{thm:MFELeu}
The logic \MFELu\ is axiomatised by \MFELeu.
\end{theorem}

\begin{proof}
By Lemma~\ref{la:Msoundu}, \MFELeu\ is sound. For completeness, assume $P_1=_{\mfeu} P_2$.
By Lemma~\ref{la:MFELeu} there are unique $\sigma_i$-normal forms $Q_{\sigma_i}$ 
such that $\MFELeu\vdash P_i=_\mfe Q_{\sigma_i}$. By assumption and soundness, 
$Q_{\sigma_1}=Q_{\sigma_2}$. Hence, $\MFELeu\vdash P_1=Q_{\sigma_1}=P_2$.
\end{proof}

\section{Two-Valued Conditional \FEL\ (\CLFELtwo) and \CLFELu}
\label{sec:6}

Two-valued Conditional \FEL\ (\CLFELtwo) and three-valued Conditional \FEL\
(\CLFELu) are both sublogics of Conditional logic (defined in 1990 by Guzmán and Squier),
as was shown in~\cite{BP23}.\footnote{%
  In fact it was shown in~\cite{BP23} that $x\fulland y=(x\leftand y) \leftor 
  (y\leftand x)$, and hence the axiom  
  $(x\leftand y) \leftor (y\leftand x)=(y\leftand x)\leftor (x\leftand y)$, which is a 
  characteristic axiom of Conditional logic, implies commutativity
  of $\fulland$. We return to this point in Section~\ref{sec:8}.

  }
\CLFELtwo\ is equivalent to a sequential version of propositional logic that refutes the absorption laws.
The characteristic property of \CLFELtwo\ and \CLFELu\ is that sequential propositions are identified
according to the commutativity of ${\fulland}$ (and ${\fullor}$).
For both the case without \und\ and with \und, we define a semantics based on ordered memorising
trees, give equational axioms for their equality, and prove a completeness result.
We note that \CLFELu\ seems most close to ``three-valued left-sequential propositional logic"
as one can get, the cost being the loss 
of absorption.\footnote{%
   Note that $x\fulland\fa=\fa$ is a consequence of absorption: 
   $x\fulland\fa=\fa\fulland x=\fa\fulland(\fa\fullor x)=\fa$. 
   This equation plays an important role in Static FEL (see Section~\ref{sec:7})}.
   
For $P\in\SPu$, its \emph{alphabet ${\alpha(P)}\subset A$} is defined by 
$\alpha(\tr)=\alpha(\fa)=\alpha(\und)=\emptyset$, 
$\alpha(a)=\{a\}$, $\alpha(\neg P)=\alpha(P)$ and $\alpha(P\fulland Q)=
\alpha(P\fullor Q)=\alpha(P)\cup\alpha(Q)$. 

We further assume that the atoms in $A$ are ordered $a_1<a_2<...<a_n<a_{n+1}...$, notation $(A,<)$. 

\begin{definition}
\label{def:XX}
$A^s_o$ is  the set of strings in $A^*$ whose elements are ordered according to $(A,<)$.
\end{definition}

So, $A^s_o\subset A^s$.
We reserve the symbols 
$\beta,\gamma$ for strings $b_{1}b_{2}...b_{k}b_{k+1}$ in $A^s_o$ that satisfy $b_i < b_{i+1}$
(not necessarily neighbours in $(A,<)$), 
and for such a string $\beta$ we write $\alp{\beta}$ for its 
set of atoms.

Below we define evaluation functions for \CLFELtwo\ and \CLFELu\ that use disjunctions
$\widetilde\fa_\beta\fullor P$ in order to identify $P$ with a term that respects
$A^s_o$, where the terms $\widetilde\fa_\beta$ are defined  in Definition~$\ref{def:MnfSpecial}$.

\begin{definition}
\label{def:clfe}
The \textbf{conditional evaluation
function} $\clfe:\SP\to\NT$ is defined by 
\[
\clfe(P)=\mfe(\widetilde\fa_\beta\fullor P), \text{ where $\beta\in A^s_o$ satisfies
$\alp{\beta}=\alpha(P)$}.
\] 
The \textbf{conditional evaluation
function} $\clfeu:\SPu\to\NTu$ is defined by 
\[
\clfeu(P)=\begin{cases}
\mfeu(\widetilde\fa_\beta\fullor P)&\text{if $P\in\SP$ and $\beta\in A^s_o$ satisfies
$\alp{\beta}=\alpha(P)$,}\\
\und&\text{otherwise}.
\end{cases}
\]
~
\end{definition}

These evaluation functions are well-defined because $\beta$ is uniquely determined 
by $\alpha(P)$ and $(A,<)$.
For each $P\in\SP$, the evaluation tree $\clfe(P)$ 
can be called a ``$\beta$-tree", its complete traces agree with $\beta$, and 
its leaves determine the evaluation result of $P$.

\begin{example}
Four typical examples with $\beta=ab$, thus 
$\wfa_\beta=a\fulland (b\fulland\fa)$:
\begin{align*}
\clfe(b\fulland a)
&= \mfe(\wfa_\beta \fullor (b\fulland a))
=\mfe(a\fulland b),
\\
\clfe(b\fullor a)
&= \mfe(\wfa_\beta \fullor (b\fullor a))
=\mfe(a\fullor b),
\\
\clfe((b\fullor a)\fulland b)
&= \mfe(\wfa_\beta \fullor ((b\fullor a)\fulland b))
=\mfe((a\fullor \tr)\fulland b),
\\
\clfe(b\fulland(a\fullor b))
&= \mfe(\wfa_\beta \fullor (b\fulland(a\fullor b)))
=\mfe((a\fullor \tr)\fulland b).
\end{align*}
\end{example}

\begin{definition}
The relation $=_{\clfe}\subset \SP\times\SP$, \textbf{conditional full valuation congruence}, is 
defined by $P=_{\clfe}Q$ if $\clfe(P)=\clfe(Q)$. The relation
$=_{\clfeu}\subset \SPu\times\SPu$, \textbf{conditional full $\bm\und$-valuation congruence}, 
is defined by $P=_{\clfe}Q$ if $\clfeu(P)=\clfeu(Q)$.
\end{definition}

\begin{lemma}
\label{la:CLcongruence} 
The relations $=_{\clfe}$ and $=_{\clfeu}$ are both congruences. 
\end{lemma}

\begin{proof}
Because $=_{\mfe}$ is a congruence (Lemma~\ref{la:Mcongruence}), so is $=_{\clfe}$.
This follows from the fact that $=_\mfe$ 
preserves the alphabet of its arguments, that is, if $P=_{\mfe}Q$ then $\alpha(P)=\alpha(Q)$.
E.g., assume for $P,Q,R\in\SP$ that $P=_{\clfe}Q$, then 
$\clfe(P\fulland R)=\mfe(\widetilde\fa_\beta\fullor(P\fulland R))
\stackrel{\text{La.\ref{la:Mcongruence}}}=\mfe(\widetilde\fa_\beta\fullor(Q\fulland R))
=\clfe(Q\fulland R)$, where the last equality holds because $\beta$ satisfies 
$\alp{\beta}=\alpha(P\fulland R)=\alpha(Q\fulland R)$.

In a similar way it follows by Lemma~\ref{la:Mcongruenceu} that $=_{\clfeu}$ is a congruence. 
\end{proof}

\begin{definition}
\label{def:CLFEL}
\textbf{Conditional Fully Evaluated Left-Sequential Logic \emph{($\bm\CLFELtwo$)}} is the fully evaluated 
left-sequential logic that satisfies no more consequences than those of \clfe-equality, i.e., 
for all $P, Q \in\SP$,
\[
\CLFELtwo\models P=Q \iff P=_{\clfe} Q.
\]
\textbf{Conditional Fully Evaluated Left-Sequential Logic with undefinedness \emph{($\bm\CLFELu$)}} 
is the fully evaluated 
left-sequential logic with undefinedness that satisfies no more consequences than those of \clfeu-equality, i.e., 
for all $P, Q \in\SPu$,
\[
\CLFELu\models P=Q \iff P=_{\clfeu} Q.
\]
\end{definition}

We define the following sets 
of axioms for \CLFELtwo\ and \CLFELu:
\begin{align*}
\CLFELe&=\MFELe\cup\{x \fulland y=y\fulland x\},
\\
\CLFELeu&=\MFELeu\cup\{x \fulland y=y\fulland x\}.
\end{align*}

\begin{lemma}[Soundness]
\label{la:CLsound}
$1.$ For all $P,Q\in\SP$, $\CLFELe\vdash P=Q~\Longrightarrow \CLFELtwo\models P=Q$,
and 
\\[1mm]
$2.$ For all $P,Q\in\SPu$, $\CLFELeu\vdash P=Q~\Longrightarrow \CLFELu\models P=Q$.
\end{lemma} 

\begin{proof}
$1.$ By Lemma~\ref{la:CLcongruence}, the relation $=_\clfe$ is a congruence on \SP, so it suffices to show
that all closed instances of the \CLFELe-axioms satisfy $=_\clfe$.
By Lemma~\ref{la:Msound}, we only have to prove this for the axiom $x \fulland y=y\fulland x$.
Let $P,Q \in\SP$ and let $\beta\in A^s_o$ be such that $\alp{\beta}=\alpha(P\fulland Q)$. 
In Appendix~\ref{app:A6} we prove
that $\mfe(\widetilde\fa_\beta\fullor (P\fulland Q))=\mfe(\widetilde\fa_\beta\fullor (Q\fulland P))$, 
hence $\clfe(P \fulland Q) = \clfe(Q\fulland P)$.
\\[2mm]
$2.$
Soundness of \CLFELeu\ follows in a similar way (see also Appendix~\ref{app:A6}).
\end{proof}

We first prove completeness of \CLFELe, and then discuss the completeness of \CLFELeu.
\begin{lemma}
\label{la:crux3}
$\CLFELe\vdash \auxf(x, ~\auxf(y,z,u), ~\auxf(y,v,w)) = \auxf(y, ~\auxf(x,z,v), ~\auxf(x,u,w))$.
\end{lemma}

\begin{proof}
By commutativity and associativity of $\fulland$ and $\fullor$, and \eqref{C4} (distributivity):
\begin{align*}
\auxf(x,~&\auxf(y,z,u),~\auxf(y,v,w))\\
\quad
&= (x \fulland ((y \fulland z) \fullor(\neg y \fulland u))) \fullor
(\neg x \fulland ((y \fulland v) \fullor(\neg y \fulland w)))\\
&= ((x \fulland (y \fulland z)) \fullor(x \fulland (\neg y \fulland u))) \fullor
((\neg x \fulland (y \fulland v)) \fullor(\neg x \fulland (\neg y \fulland w)))\\
&= ((y \fulland (x \fulland z)) \fullor(\neg y \fulland (x \fulland u))) \fullor
((y \fulland (\neg x \fulland v)) \fullor(\neg y \fulland (\neg x \fulland w)))\\
&= (y \fulland ((x \fulland z) \fullor(\neg x \fulland v))) \fullor
(\neg y \fulland ((x \fulland u) \fullor(\neg x \fulland w)))\\
&= \auxf(y,~\auxf(x,z,v),~\auxf(x,u,w)).
\end{align*}
\end{proof}

\begin{lemma}
\label{la:crux4}
Let $\sigma\in A^s$ and $\sigma'$ be a permutation of $\sigma$. 
Then for each $\sigma$-normal $P_\sigma$ there is a $\sigma'$-normal form $Q$ such that 
$\CLFELe\vdash P_\sigma=Q_{\sigma'}$.
\end{lemma}

\begin{proof} By induction on $\sigma$.
If $\sigma=\epsilon$ or $\sigma=a\in A$, this is trivial. 

Assume $\sigma=a_0a_1...a_k$ and $\sigma'= b_0b_1...b_k$. If $a_0=b_0$,  
we are done by induction.
Otherwise, $a_0=b_\ell$ with $\ell\ne 0$ and by induction, $P_\sigma 
=\auxf(a_0,R_{\rho'} ,S_{\rho'} )$ with $\rho'=b_0\rho''$ for some $\rho''\in A^s$.
By Lemma~\ref{la:crux3}, $\auxf(a_0,R_{\rho'} ,S_{\rho'} )=\auxf(b_0,V_{a_0\rho''} ,W_{a_0\rho''} )$.
Now, either we are done, or  
there is a permuation $p()$ of $a_0\rho''$ such that $b_0 p(a_0\rho'')=\sigma'$.
By induction, $\auxf(b_0,V_{a_0\rho''} ,W_{a_0\rho''} )=\auxf(b_0,V_{p(a_0\rho'')}' ,W_{p(a_0\rho'')}' )
=Q_{\sigma'}$.
\end{proof}

\begin{definition}
A \textbf{$\bm\beta$-normal form} is a $\sigma$-normal form (Definition~\ref{def:Mnf})
with $\sigma\in A^s_o$. We denote $\beta$-normal forms by $P_\beta$, $Q_\beta$, etc.
\end{definition}

\begin{lemma}
\label{la:cl}
Let $\beta\in A^s_o$ and $\sigma$ a permutation of $\beta$. 
Then in \CLFELtwo, each $\sigma$-normal form $P_\sigma$ is provably equal to a $\beta$-normal form.
\end{lemma}

\begin{proof} 
By Lemma~\ref{la:crux4}.
\end{proof}

\begin{theorem}[Completeness]
\label{thm:CLFELe}
The logic \CLFELtwo\ is axiomatised by \CLFELe.
\end{theorem}

\begin{proof}
By Lemma~\ref{la:CLsound}, \CLFELe\ is sound. For completeness, assume $P_1=_{\clfe} P_2$, thus,
for some $\beta\in A^s_o$, $\alpha(P_i)=\alp{\beta}$.
By assumption and Lemma~\ref{la:cl} there is a $\beta$-normal form $Q_{\beta}$ such that 
$\CLFELe\vdash P_1=Q_{\beta}=P_2$.
\end{proof}

The completeness of \CLFELeu\ can be proved in the same way.
In comparison with \CLFELtwo, there is  only one additional $\beta$-normal form \und.

\begin{theorem}[Completeness]
\label{thm:CLFELeu}
The logic \CLFELu\ is axiomatised by \CLFELeu.
\end{theorem}

It follows that \CLFELu\ is equivalent to  Bochvar's well-known 
three-valued logic~\cite{Boc38}.
In Comment~\ref{com:Boch} (on independent axiomatisations)
we discuss this equivalence, or more precisely, 
the equivalence of \CLFELeu\ and the equational axiomatisation of Bochvar's logic 
discussed in~\cite{BBR95}.

\section{Static \FEL\ (\SFEL), short-circuit logic, and independence}
\label{sec:7}

In this section we define Static \FEL\ (\SFEL), the strongest two-valued \FEL\ we consider and
which is equivalent to a sequential of propositional logic.
\SFEL\ is axiomatised by adding $x\fulland \fa=\fa$ to the axioms of (two-valued) \CLFELtwo. 
It is immediately clear that \SFEL\ cannot be extended with \und: $\fa=\und\fulland\fa=\und$. 
In \SFEL, the difference between \emph{full} and \emph{short-circuit} evaluation 
has disappeared and therefore we can reuse results on \emph{static short-circuit logic}, which 
we briefly introduce.
Finally, we review all \FEL-axiomatisations and give 
simple equivalent alternatives with help of \emph{Prover9}, all of which are independent. 

We assume $(A,<)$ and $A^s_o$ (Definition~\ref{def:XX})
and consider terms whose alphabet is constrained by some $\beta\in A^s_o$ and
for which we define an evaluation funtion $\sfe_\beta$.

\begin{definition}
For any $\beta\in A^s_o$,  $\SPbeta=\{P\in\SP\mid \alpha(P)\subset\alp{\beta}\}$.
For any $P\in\SPbeta$, 
\[
\sfe_\beta(P)=\mfe(\wfa_\beta\fullor P).
\]
The relation $=_{\sfe,\beta}\subset \SPbeta\times\SPbeta$ is called \textbf{static full valuation 
$\bm\beta$}-congruence
and is defined by $P=_{\sfe,\beta}Q$ if $\sfe_\beta(P)=\sfe_\beta(Q)$.
\end{definition}

The crucial difference between $\sfe_\beta$ and the evaluation function $\clfe$ 
(Definition~\ref{def:clfe}) concerns their domain: for $\sfe_\beta(P)$ it is only
required that $\alpha(P)$ is a subset of $\alp\beta$.

\begin{example}
If $\beta=ab$, we find that
\begin{align*}
&(\tr\unlhd b\unrhd \tr)\unlhd a\unrhd(\tr\unlhd b\unrhd \tr)
=
\sfe_{ab}(\tr)
=\sfe_{ab}(\neg b \fullor b)
=\sfe_{ab}(b\fullor \tr),
\\[2mm]
&(\tr\unlhd b\unrhd \tr)\unlhd a\unrhd(\fa\unlhd b\unrhd \fa)
\;=\sfe_{ab}(a)
=
\sfe_{ab}((b\fulland\fa) \fullor a)
=
\sfe_{ab}(\auxf(b,a,a)),
\\[2mm]
&(\tr\unlhd b\unrhd \fa)\unlhd a\unrhd(\tr\unlhd b\unrhd \fa)
\;=
\sfe_{ab}(b)
=
\sfe_{ab}(b\fullor(a\fulland\fa)))
=
\sfe_{ab}(\auxf(a,b,b)).
\end{align*}
\end{example}

\begin{definition}
\label{def:SFEL}
\textbf{Static Fully Evaluated Left-Sequential Logic \emph{\SFEL}} is the fully evaluated 
left-sequential logic that satisfies for any $\beta\in A^s_o$
no more consequences than those of $\sfe_\beta$-equality, i.e., 
for all $P, Q \in\SPbeta$,
\[
\SFEL\models P=Q \iff P=_{\sfe_\beta} Q.
\]
\end{definition}

We define the following set \SFELe\ of axioms for \SFEL:
\[
\SFELe=\CLFELe\cup\{x\fulland \fa=\fa\}.
\]

In order to prove completeness of \SFELe, we reuse a result from \emph{Short-Circuit Logic} (SCL),
a family of logics based on the connectives ${\leftand}$ and ${\leftor}$ that prescribe 
short-circuit evaluation:
once the evaluation result of a propositional expression is determined, evaluation stops. Typically,
$a\leftand b$ has as its semantics the evaluation tree in which atom $b$ is only evaluated if $a$ evaluates to $\true\,$:
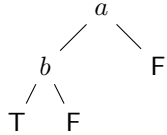
\begin{center}
\begin{tikzpicture}[%
      level distance=7.5mm,
      level 1/.style={sibling distance=15mm},
      level 2/.style={sibling distance=7.5mm},
      baseline=(current bounding box.center)]
      \node (a) {$a$}
        child {node (b1) {$b$}
          child {node (d1) {$\tr$}} 
          child {node (d2) {$\fa$}}
        }
        child {node (b2) {$\fa$}
        };
      \end{tikzpicture}
\end{center}
The \SCL-family comprises the counterparts of the \FEL s 
presented in this paper, that is, \FSCL, \MSCL, \CLSCLtwo\ and \SSCL, and in all these logics 
the connective $\fulland$ is definable, as was noted in~\cite{Stau,BPS18,BP23}:
\[
x\fulland y=(x\leftor(y\leftand\fa))\leftand y.
\]
We return to these facts in Section~\ref{sec:8}.

\begin{table}
\centering
\hrule
\begin{align}
\fa&=\neg\tr
\label{Mem1}
\tag{Mem1}
\\
x\leftor y&=\neg(\neg x\leftand\neg y)
\label{Mem2}
\tag{Mem2}
\\[0mm]
\label{Mem3}
\tag{Mem3}
\tr\leftand x&=x
\\ 
\label{Mem4}
\tag{Mem4}
x\leftand(x\leftor y)&=x
\\[0mm]
\label{Mem5}
\tag{Mem5}
(x\leftor y)\leftand z&=(\neg x\leftand (y\leftand z))\leftor(x\leftand z)
\\[0mm]
\label{S6}
\tag{Comm}
x\leftand y&=y\leftand x
\end{align}
\hrule
\caption{$\SSCLe$, a set of axioms for \SSCL\ 
(axioms~\eqref{Mem1}--\eqref{Mem5} are used in Appendix~\ref{app:A4})}
\label{tab:SSCL}
\end{table}

\begin{theorem}[Completeness]
\label{thm:SFELe}
The logic \SFEL\ is axiomatised by \SFELe.
\end{theorem}

\begin{proof}
For $\sigma\in A^s$, let ${\cal S}_\sigma$ be the set of closed SCL-terms with atoms 
in $\sigma$.
In \cite{BPS18}, it is argued on p.21 that the static evaluation trees of all $P\in{\cal S}_\sigma$
are perfect binary trees, where each level characterises the evaluation
of a single atom in $\sigma$, and it easily follows that $\sse_\sigma(P)=\sfe_\sigma(P')$ if $P'$ 
is obtained from $P$ by 
replacing the short-circuit connectives by their fully evaluated counterparts.

According to~\cite{BPS18}, \SSCLe\ is the set of axioms in Table~\ref{tab:SSCL} that axiomatises
\SSCL. 
Three consequences of \SSCLe\ are $\fa\leftand x=\fa$, $x\leftand y=y\leftand x$, and 
$x\fulland y=x\leftand y$. We recall from~\cite{BPS18} a derivation of the last one:
$x\fulland y=(x\leftor(y\leftand\fa))\leftand y=(x\leftor(\fa\leftand y))\leftand y
=(x\leftor\fa)\leftand y= x\leftand y$.

In~\cite[Thm.6.11]{BPS18} it is proven that for all $P,Q\in {\cal S}_\sigma$, 
$\SSCLe\vdash P=Q \iff \sse_\sigma(P)=\sse_\sigma(Q)$.
Identifying $\fulland$ and $\leftand$, it suffices to show that \SSCLe\ and \SFELe\ are equivalent, 
and this follows quickly with \emph{Prover9}. Hence, for all $P,Q\in\SPbeta$,
$\SFELe\vdash P=Q\iff P=_{\sfe_\beta}Q$.
\end{proof}

Finally, we provide simple equivalent alternatives for all \FEL-axiomatizations
with help \emph{Prover9}, where we use the superscript \und\ in their names to indicate the extension 
with the two axioms $\neg\und=\und$ and $\und\fulland x=\und$.
Moreover, each of these alternative axiomatisations is independent, which 
quickly follows with help of \emph{Mace4}.

\begin{description}
\item[\text{\FFELe}] and \textbf{\FFELeu.}
~Define $\FFELe_2=\FFELe\setminus\{\eqref{FFEL1}\}$, where \eqref{FFEL1} is the axiom $\fa=\neg\tr$. 
With \emph{Prover9} it quickly follows that  
$\FFELe_2\vdash \fa=\neg\tr$, and according to 
\emph{Mace4}, $\FFELeu_2$ 
is independent. Hence, $\FFELe_2$ is also independent.

\item[\textbf{\MFELe}] and \textbf{\MFELeu.} 
~Let MF be the set of axioms in Table~\ref{tab:MF}.

\begin{table}[b]
{
\centering
\hrule
~\\[-5mm]
\begin{align}
\label{MF1}
\tag{MF1}
x\fullor y
&=\neg(\neg x\fulland\neg y)
&&\text{\blauw{[this is \eqref{FFEL2}]}
}
\\ 
\label{MF2}
\tag{MF2}
\neg\neg x&=x
&&
\text{\blauw{[this is \eqref{FFEL3}]}}
\\
\label{MF3}
\tag{MF3}
\tr\fulland x&=x
&&\text{\blauw{[this is \eqref{FFEL5}]}\hspace{-18mm}}
\\
\label{MF4}
\tag{MF4}
(x\fullor y) \fulland z
&=(\neg x\fulland(y\fulland z)) \fullor (x\fulland z)
&&\text{\blauw{[this is \eqref{M1}]}\hspace{-18mm}}
\\
\label{MF5}
\tag{MF5}
(x\fulland y)\fullor x&=x\fullor (y\fulland\fa)
\\
\label{MF6}
\tag{MF6}
x\fulland (y\fullor z)&=(x\fulland y)\fullor (x\fulland z)
&&\text{\blauw{[this is \eqref{C4}]}\hspace{-18mm}}
\end{align}
\hrule
}
\caption{MF, an alternative set of axioms for \MFEL}
\label{tab:MF}
\end{table}
With
\emph{Prover9} it follows that $\text{MF}\vdash \MFELe$ and $\MFELe\vdash\eqref{MF5}$, 
hence \MFELe\ and \text{MF} are equivalent, and so are \MFELeu\ and $\text{MF}^\und$. 
Furthermore, with help of \emph{Mace4} it quickly follows that $\text{MF}^\und$ is independent,
and so is \text{MF}.

\item[\textbf{\CLFELe}] and \textbf{\CLFELeu.}  
~Define
$\text{CF}=\{x\fulland y=y\fulland x\}\cup\text{MF}\setminus\{\eqref{MF6}\}$.
According to  \emph{Prover9}, $\text{CF}\vdash\CLFELe$, and according to
\emph{Mace4}, $\text{CF}^\und$ is independent, and so is \text{CF}.

\item[\SFELe.] 
~Define 
$\text{SF}=\{x\fulland \fa=\fa\}\cup\text{MF}\setminus\{\eqref{MF6}\}$.
According to \emph{Prover9}, $\text{SF}\vdash\SFELe$, and according to \emph{Mace4}, 
$\text{SF}$ is independent.
\end{description}

Thus, in addition to its independence, the alternative axiomatisation \text{MF} is elegant, 
due to its simplicity and conciseness, and this is also the case for the stronger axiomatisations.
We end this section with a comment on \CLFELeu\ and $\text{CF}^\und$.

\begin{comment}
\label{com:Boch}
Another axiomatisation of \CLFELu\ is briefly discussed 
in~\cite[Sect.6]{BP23} and stems from~\cite{BBR95}, in which an equational axiomatisation
of Bochvar's logic~\cite{Boc38} is introduced and proven complete that we repeat here 
(using full left-sequential connectives instead of $\wedge$ and $\vee$):
~\\[-4mm]\indent
\begin{minipage}[t]{0.46\linewidth}\centering
\begin{Lalign}
\tag{S1}
\neg\tr&=\fa
\\
\tag{S2}
\neg\und&=\und
\\
\tag{S3}
\neg\neg x&=x
\\
\tag{S4}
\neg(x\fulland y)&=\neg x\fullor \neg y
\\
\tag{S6}
(x\fulland y)\fulland z
&=x\fulland (y\fulland z)
\end{Lalign}
\end{minipage}
\begin{minipage}[t]{0.62\linewidth}\centering
\begin{Lalign}
\tag{S7}
\tr\fulland x
&=x
\\
\tag{S8}
x\fullor (\neg x\fulland y)
&=x\fullor  y
\\
\tag{S9}
x\fulland y
&= y\fulland x
\\
\tag{S10}
x\fulland (y\fullor z)
&= (x\fulland y)\fullor(x\fulland z)
\\
\tag{S11}
\und\fulland x&=\und
\end{Lalign}
\end{minipage}
\\[3mm]
The `missing' axiom (S5) defines the connective 
$\fullimp$ (full left-sequential implication, $x\fullimp y=\neg x\fullor y$) and is not relevant here.
According to \emph{Prover9}, this axiomatisation is equivalent with \CLFELeu.
However, without axiom (S9), this set is not an axiomatisation of \MFELu, and therefore
we prefer $\text{CF}^\und$ and \text{CF}.
Finally, we note that according to \emph{Mace4},
axiom (S6) is derivable from the remaining axioms, which 
are independent. Moreover, the same is true for axiom (S10).  
\end{comment}

\section{Discussion and conclusions}
\label{sec:8}
As mentioned earlier, this paper is a continuation of~\cite{Stau}, 
in which \FFEL\ was introduced  
and which contained the first completeness proof for the short-circuit logic \FSCL\ (Free SCL).
We introduced six more \emph{fully evaluated sequential logics} (\FEL s), two of which 
are not new: \CLFELu, which is equivalent to the logic of Bochvar~\cite{Boc38}, and \SFEL, which
is equivalent to \SSCL\ and to propositional logic, see~\cite{BPS13}. 

We could have defined the fully evaluated connectives $\fulland$ and $\fullor$, and thus the FELs, with help
of the ternary connective $x\lef y\rig z$, the so-called ``conditional" that expresses 
``\texttt{if $y$ then $x$ else $z$}" and with which
Hoare characterised the propositional calculus in 
1985~\cite{Hoa85} (also using both constants \tr\ and \fa).\footnote{%
  However, in 1948, Church introduced \emph{conditioned disjunction} $[p,q,r]$ in~\cite{Chu}
  as a primitive connective for the propositional calculus, which expresses the same 
  as Hoare’s conditional $p\lef q\rig r$, and proved the same result.
  }
 With this connective, short-circuit conjunction $\leftand$ and disjunction $\leftor$ are simply defined 
by 
\[x \leftand y = y\lef x\rig\fa \quad \text{and} \quad x \leftor y = \tr\lef x\rig y,\]
 and negation by 
$\neg x = \fa\lef x\rig \tr$. 
The conditional connective naturally characterises short-circuit evaluation and has a notation
that supports equational reasoning, which led us to define \emph{short-circuit logics} (\SCL s) 
based on this connective~\cite{BPS13,BP23}. 
SCLs can be characterised by their efficiency, 
in the sense that atoms are not evaluated if their 
evaluation is not needed to determine the evaluation of a term as a whole. 
From that perspective FELs might seem rather inefficient, but this is not necessarily so. We first 
elaborate on SCLs and then return to this point.
\\[1.6mm]
\noindent\textbf{Short-circuit evaluation and SCLs.}
The short-circuit logic \FSCL\ is defined by the equational theory of terms 
over the signature $\Sigma_\CP=\{\leftand,\neg,\tr,\fa\}$ that is generated by the two
equations $x \leftand y = y\lef x\rig\fa$ and $\neg x = \fa\lef x\rig \tr$,
and the set of four basic axioms that we call \CP: 
\[
x\lef\tr\rig y=x, \quad x\lef\fa\rig y=y, \quad \tr\lef x\rig\fa=x, \quad x\lef(y\lef z \rig u)\rig v=
(x\lef y\rig v)\lef z\rig (x\lef u\rig v).
\]
Write $\CP(\leftand,\neg)$ for the set of these six axioms.\footnote{%
  Short-circuit disjunction $x\leftor y$ is defined by 
  $\neg(\neg x\leftand \neg y)$.}
As an example, the associativity of  ${\leftand}$ is
quickly derived from $\CP(\leftand,\neg)$, and thus holds in \FSCL:
\begin{align*}
(x\leftand y)\leftand z
&= z\lef(y\lef x\rig\fa)\rig \fa
&&\text{by definition}
\\
&=(z\lef y\rig\fa)\lef x\rig(z\lef \fa\rig\fa)
&&\text{by $x\lef(y\lef z \rig u)\rig v=(x\lef y\rig v)\lef z\rig (x\lef u\rig v)$}
\\
&=(z\lef y\rig\fa)\lef x\rig\fa
&&\text{by $x\lef\fa\rig y=y$}
\\
&=x\leftand(y\leftand z).
&&\text{by definition}
\end{align*}
Each of the  \FEL s discussed in this paper is related to one of the \SCL s
defined in~\cite{BPS13,BP23}. 
\emph{Memorising} SCL (\MSCL) 
is defined as the restriction to $\Sigma_\CP$
of the equational theory of $\CP(\leftand,\neg)$ extended with the ``memorising" axiom
\begin{equation}
\label{eq:mem}
x\lef y\rig (z\lef u \rig(v\lef y\rig w))=x\lef y\rig (z\lef u \rig w)
\qquad\text{($y$ is ``memorised").}
\end{equation}
As an example, $x\leftand x=x$ is easily derivable from these seven axioms,
and so are the axioms~\eqref{Mem1}--\eqref{Mem5} in Table~\ref{tab:SSCL}. 
More generally, the validity of any equation in a particular \SCL\ 
can be checked in the associated \CP-system: either a proof can be
found (often with help of \emph{Prover9}), or \emph{Mace4} finds a counter model. 
We finally note that in~\cite{BPS21}, \SCL s were extended with a constant \und\ for the 
third truth value \undefi\  and the defining axiom $x\lef\und\rig y=\und$. 
A complete axiomatisation for \MSCLu\ was given, but an axiomatisation of \FSCLu\ is not yet known,
not even when restricted to closed terms, see \cite[Conject.8.1]{BPS21}.
\\[1.6mm]
\noindent\textbf{\FEL s and \SCL s.}
As was noted in~\cite{Stau}, 
the connective ${\fulland}$ can be defined 
using short-circuit connectives and the constant \fa\ (due to the \FSCL-law $\fa\leftand x=\fa$):
\begin{equation}
\label{eq:rhs}
x\fulland y=(x\leftor (y\leftand\fa))\leftand y.
\end{equation}
It is easy to prove that 
$\CP(\leftand,\neg)+\eqref{eq:rhs}\vdash x\fulland y=y\lef x\rig(\fa\lef y\rig\fa)$
and the latter equation is perhaps more intuitive than~\eqref{eq:rhs}.
The addition of $x\fulland y=y\lef x\rig(\fa\lef y\rig\fa)$ and $\neg x=\fa\lef x\rig\tr$  to \CP, 
say $\CP(\fulland,\neg)$, 
implies that $\CP(\fulland,\neg)\vdash\FFELe$.
Similarly, the addition of axiom~\eqref{eq:mem} to $\CP(\fulland,\neg)$
implies that $x\fulland x=x$ and~\eqref{M1} in Table~\ref{tab:MFELe}
are derivable. 
Like in the case of \SCL s, the validity of any equation
in a particular \FEL\ can be easily checked in the associated \CP-system. 
The addition of \und\ 
is in the case of \FEL s simple.
For \FFELu\ there are normal forms $\und_\sigma$ with $\sigma\in A^*$, and for \MFELu,
normal forms $\und_\sigma$ with $\sigma\in A^s$. In the case of \CLFELu, these all reduce to \und.
\\[1.6mm]
\noindent\textbf{Side effects.}
A quote from~\cite{Stau}: We find that some programming languages offer full 
left-sequential connectives, which motivated the initial investigation of FEL. 
We claim that FEL has a greater value than merely to act as means of writing certain SCL-terms 
using fewer symbols. 
The usefulness of a full evaluation strategy lies in the 
increased 
predictability of the state of the environment after a (sub)term has been
evaluated. In particular, we know that the side effects of all the atoms in the term have occurred. 
To determine the state of the environment after the evaluation of a FEL-term in a given environment, 
we need only compute how each atom in the term transforms the environment. 
It is not necessary to compute the evaluation of any of the atoms. 
With SCL-terms in general we must know to what the first atom evaluates in order to determine which atom 
is next to transform the environment. Thus to compute the state of the environment after the 
evaluation of an SCL-term we must compute the evaluation result of each atom that transforms the environment 
and we must compute the transformation of the environment for each atom that affects it. 
\\[1.6mm]
\noindent\textbf{Expressiveness.} 
We call a perfect evaluation tree \emph{uniform} if each complete trace 
contains the same sequence of atoms. 
By definition, each \FEL\ defines a certain equality on the 
domain of uniform evaluation trees.
In \FFEL\ and \FFELu, some uniform evaluation trees can be expressed, but for
example not
\[
(\tr\unlhd b\unrhd\fa)\unlhd a\unrhd(\fa\unlhd b\unrhd\tr),
\]
which is expressible in \MFEL\ by $\auxf(a,b,\neg b)=(a\fulland b)\fullor(\neg a \fulland \neg b)$,
and it easily follows that \emph{each} uniform tree can be expressed in \MFEL. 
Of course, in \SFEL, which is equivalent to \SSCL,
with $x\fulland\fa=x\leftand \fa=\fa$, each uniform
evaluation tree can be compared to a truth table for propositional logic and
it is possible to define a semantics for \SFEL\ based on an equivalence $\equiv$ 
on evaluation trees that identifies each (sub)tree with only \fa-leaves with \fa, and each (sub)tree
with only \tr-leaves with \tr. 
Then, any $P\in\SPbeta$ has a representing evaluation tree without such subtrees, for example,
\begin{align*}
&\sfe_{abc}((a\fulland b)\fullor(\neg a \fulland c))
\\
&=
((\tr\unlhd c\unrhd\tr)\unlhd b\unrhd(\fa\unlhd c\unrhd\fa))\unlhd a 
\unrhd((\tr\unlhd c\unrhd\fa)\unlhd b\unrhd(\tr\unlhd c\unrhd\fa))
\\
&=
((\tr\unlhd c\unrhd\tr)\unlhd b\unrhd(\fa\unlhd c\unrhd\fa))\unlhd a 
\unrhd((\tr\unlhd b\unrhd\tr)\unlhd c\unrhd(\fa\unlhd b\unrhd\fa))
&&\text{(by Lemma~\ref{la:crux3})}
\\
&=
(\tr\unlhd b\unrhd\fa)\unlhd a \unrhd(\tr\unlhd c\unrhd\fa).
&&\text{(by $\SFEL/{\equiv}$)} 
\end{align*}
In \MSCL, the latter tree is the evaluation tree of $(a\leftand b)\leftor(\neg a \leftand c)$, 
see~\cite{BPS21} for a detailed explanation.
\newpage
As a final point, we note that there are left-sequential, binary connectives for which the distinction 
between short-circuit 
evaluation and full evaluation does not exist, such as the \emph{left-sequential biconditional} and the
\emph{left-sequential exclusive or}, defined by Cornets de Groot in~\cite{CdG20} by 
$x\liffC y=y\lef x\rig(\fa\lef y\rig\tr)$ 
and $x\lxorC y=(\fa\lef y\rig\tr)\lef x\rig y$, respectively.
It immediately follows that $x\lxorC y=x\liffC\neg y$ and $\neg(x\liffC y)=x\liffC \neg y$.
In~\cite{PaPo22}, the notations ${\liffo}$ and $\lxoro$ are used, and below we use
the symbols ${\liff}$ and $\lxor$, expressing only left-sequentiality. 
Regarding expressiveness, we note that in \FFEL\ extended with $\liff$, 
the evaluation tree of $a\liff b$ is $(\tr\unlhd b\unrhd\fa)\unlhd a\unrhd(\fa\unlhd b\unrhd\tr)$.
However, it easily follows that not all uniform 
evaluation trees can be expressed in \FFEL\ with ${\liff}$.\footnote{%
  Suppose $A=\{a\}$. Then there are $5\cdot 4^4\cdot 3^3$ terms composed of four elements of
  $\{a,\neg a, a\fullor\tr,a\fulland\fa\}$ and the three connectives in $\{\fulland,\fullor,\liff\}$ 
  that have evaluation trees of depth 4. It is not hard to prove that these terms represent all 
  closed terms with evaluation trees of depth 4.
  However, there are $2^{16}$ such trees. }
Of course, in \MFEL\ and all stronger \FEL s, the left-sequential biconditional $x\liff y$ 
is definable by $\auxf(x,y,\neg y)$. 

\noindent\textbf{Related work.}
In the recent paper~\cite{BP23}, \emph{Conditional logic} (Guzmán and 
Squier~\cite{GS90}) is viewed as a short-circuit logic, and its two-valued variant 
\CLSCLtwo\ is located in between MSCL and 
SSCL.\footnote{%
  The name \CLSCL\ was chosen because CSCL is already used for ``Contractive
  SCL" in~\cite{BPS13}, the short-circuit logic that (only) contracts adjacent atoms: $a\leftand a=a$, but
  $(a\leftand b)\leftand(a\leftand b)$ and $a\leftand b$ have different evaluation trees.
  }
In~\cite{BP23} it was also established that Conditional logic defines a fully evaluated 
sublogic that is equivalent to the strict three-valued logic of Bochvar~\cite{Boc38};  
this is the logic we now call \CLFELu.
It is remarkable that also for the two-valued case,
the typical axiom $x\fulland y = y \fulland x$ is a consequence of the 
distinguishing C$\ell$SCL-axiom
\[(x\leftand y)\leftor(y\leftand x)= (y\leftand x)\leftor(x\leftand y).\]
This particular consequence led us to \CLFELtwo,
a \FEL\ in between \MFEL\ and \SFEL\ that can
be motivated as a sublogic of C$\ell$SCL. 
Thus, we locate \CLFEL\ in a hierarchy of two-valued FELs, and we derive its existence 
and motivation both from the three-valued logic \CLSCLu, 
and thus from Bochvar's logic,
and from the consideration to adopt a full evaluation strategy.

With regard to the family of \FEL s, it remains a challenge to find related work.
This is also the case for the (somewhat larger) family of SCLs, apart from~\cite{GS90} and
our own papers on SCLs.

\small
\appendix
\section{Proofs - Section~\ref{sec:2}}
\label{app:A2}

\noindent\textbf{Correctness of $\bm\nf$.}
\label{sec:lslnf}
To prove that $\nf: \FT \to \FNF$ is indeed a normalization function we need to
prove that for all $\FEL$-terms $P$, $\nf(P)$ terminates, $\nf(P) \in \FNF$ and
$\EqFFEL \vdash \nf(P) = P$. To arrive at this result, we prove several
intermediate results about the functions $\nf^n$ and $\nf^c$, roughly in the
order in which their definitions were presented in Section \ref{sec:2}. For
the sake of brevity we will not explicitly prove that these functions
terminate. To see that each function terminates consider that a termination
proof would closely mimic the proof structure of the lemmas dealing with the
grammatical categories of the images of these functions.

\begin{lemma}
\label{lem:ptpf}
For any $P^\fa$ and $P^\tr$, $\EqFFEL \vdash P^\fa = P^\fa \fulland
\fa$ and $\EqFFEL \vdash P^\tr = P^\tr \fullor \tr$.
\end{lemma}

\begin{proof}
We prove both claims simultaneously by induction. In the base case we have
$\fa = \tr \fulland \fa$ by \eqref{FFEL5}, which is equal to $\fa
\fulland \fa$ by \eqref{FFEL8} and \eqref{FFEL1}. The base case for the
second claim follows from that for the first claim by duality.

For the induction we have $a \fulland P^\fa = a \fulland (P^\fa \fulland
\fa)$ by the induction hypothesis and the result follows from
\eqref{FFEL4}. For the second claim we again appeal to duality.
\end{proof}

\begin{lemma}
\label{lem:feqs2}
The following equations can be derived by equational logic and $\EqFFEL$.
\textup{
\begin{enumerate}
\setlength\itemsep{5pt}
\item $x \fulland (y \fulland (z \fulland \fa)) = (x \fullor y) \fulland (z
  \fulland \fa)$
  \label{eq:a3}
\item $\neg x \fulland (y \fullor \tr) = \neg (x \fulland (y \fullor \tr))$
  \label{eq:a4}
\end{enumerate}
}
\end{lemma}

\begin{proof}
\begin{align*}
x \fulland (y \fulland (z \fulland \fa))
&= x \fulland ((\neg y \fulland z) \fulland \fa)
&&\textrm{by \eqref{FFEL4} and Lemma \ref{la:feqs}.\ref{eq:a1}} \\
&= (\neg x \fulland \neg y) \fulland (z \fulland \fa)
&&\textrm{by \eqref{FFEL4} and Lemma \ref{la:feqs}.\ref{eq:a1}} \\
&= \neg(\neg x \fulland \neg y) \fulland (z \fulland \fa)
&&\textrm{by Lemma \ref{la:feqs}.\ref{eq:a1}} \\
&= (x \fullor y) \fulland (z \fulland \fa),
&&\textrm{by \eqref{FFEL2}}
\end{align*}
\begin{align*}\neg x \fulland (y \fullor \tr)
&= \neg x \fullor (y \fulland \fa)
&&\textrm{by \eqref{FFEL10}} \\
&= \neg (x \fulland \neg(y \fulland \fa))
&&\textrm{by \eqref{FFEL2} and \eqref{FFEL3}} \\
&= \neg (x \fulland \neg(\neg y \fulland \neg \tr))
&&\textrm{by \eqref{FFEL8} and \eqref{FFEL1}} \\
&= \neg (x \fulland (y \fullor \tr)).
&&\textrm{by \eqref{FFEL2}} 
\end{align*}
\end{proof}

\begin{lemma}
\label{lem:nfn}
For all $P \in \FNF$, if $P$ is a
$\tr$-term then $\nf^n(P)$ is an $\fa$-term, if it is an $\fa$-term
then $\nf^n(P)$ is a $\tr$-term, if it is a $\tr$-$*$-term then so is
$\nf^n(P)$, and
\begin{equation*}
\EqFFEL \vdash \nf^n(P) = \neg P.
\end{equation*}
\end{lemma}
\begin{proof}
We start with proving the claims for $\tr$-terms, by induction on $P^\tr$.
In the base case $\nf^n(\tr) = \fa$. It is immediate that $\nf^n(\tr)$
is an $\fa$-term. The claim that $\EqFFEL \vdash \nf^n(\tr) = \neg \tr$
is immediate by \eqref{FFEL1}. For the inductive case we have that $\nf^n(a
\fullor P^\tr) = a \fulland \nf^n(P^\tr)$, where we may assume that
$\nf^n(P^\tr)$ is an $\fa$-term and that $\EqFFEL \vdash \nf^n(P^\tr) =
\neg P^\tr$. The grammatical claim now follows immediately from the induction
hypothesis. Furthermore, noting that by the induction hypothesis we may assume
that $\nf^n(P^\tr)$ is an $\fa$-term, we have:
\begin{align*}
\nf^n(a \fullor P^\tr)
&= a \fulland \nf^n(P^\tr)
&&\textrm{by definition} \\
&= a \fulland (\nf^n(P^\tr) \fulland \fa)
&&\textrm{by Lemma \ref{lem:ptpf}} \\
&= \neg a \fulland (\nf^n(P^\tr) \fulland \fa)
&&\textrm{by Lemma \ref{la:feqs}.\ref{eq:a1}} \\
&= \neg a \fulland \nf^n(P^\tr)
&&\textrm{by Lemma \ref{lem:ptpf}} \\
&= \neg a \fulland \neg P^\tr
&&\textrm{by induction hypothesis} \\
&= \neg (a \fullor P^\tr).
&&\textrm{by \eqref{FFEL3} and \eqref{FFEL2}}
\end{align*}

For $\fa$-terms we prove our claims by induction on $P^\fa$. In the base
case $\nf^n(\fa) = \tr$. It is immediate that $\nf^n(\fa)$ is a
$\tr$-term. The claim that $\EqFFEL \vdash \nf^n(\fa) = \neg \fa$ is
immediate by the dual of \eqref{FFEL1}. For the inductive case we have that
$\nf^n(a \fulland P^\fa) = a \fullor \nf^n(P^\fa)$, where we may assume
that $\nf^n(P^\fa)$ is a $\tr$-term and $\EqFFEL \vdash \nf^n(P^\fa) =
\neg P^\fa$. It follows immediately from the induction hypothesis that
$\nf^n(a \fulland P^\fa)$ is a $\tr$-term.  Furthermore, noting that by
the induction hypothesis we may assume that $\nf^n(P^\fa)$ is a
$\tr$-term, we prove the remaining claim as follows:
\begin{align*}
\nf^n(a \fulland P^\fa)
&= a \fullor \nf^n(P^\fa)
&&\textrm{by definition} \\
&= a \fullor (\nf^n(P^\fa) \fullor \tr)
&&\textrm{by Lemma \ref{lem:ptpf}} \\
&= \neg a \fullor (\nf^n(P^\fa) \fullor \tr)
&&\textrm{by the dual of Lemma \ref{la:feqs}.\ref{eq:a1}} \\
&= \neg a \fullor \nf^n(P^\fa)
&&\textrm{by Lemma \ref{lem:ptpf}} \\
&= \neg a \fullor \neg P^\fa
&&\textrm{by induction hypothesis} \\
&= \neg (a \fulland P^\fa).
&&\textrm{by \eqref{FFEL3} and \eqref{FFEL2}}
\end{align*}

To prove the lemma for $\tr$-$*$-terms we first verify that the auxiliary
function $\nf^n_1$ returns a $*$-term and that for any $*$-term $P$, $\EqFFEL
\vdash \nf^n_1(P) = \neg P$. We show this by induction on the number of
$\ell$-terms in $P$. For the base cases, i.e., for $\ell$-terms, it is
immediate that $\nf^n_1(P)$ is a $*$-term. If $P$ is an $\ell$-term with a
positive determinative atom we have:
\begin{align*}
\nf^n_1(a \fulland P^\tr)
&= \neg a \fulland P^\tr
&&\textrm{by definition} \\
&= \neg a \fulland (P^\tr \fullor \tr)
&&\textrm{by Lemma \ref{lem:ptpf}} \\
&= \neg (a \fulland (P^\tr \fullor \tr))
&&\textrm{by Lemma \ref{lem:feqs2}.\ref{eq:a4}} \\
&= \neg (a \fulland P^\tr).
&&\textrm{by Lemma \ref{lem:ptpf}}
\end{align*}
If $P$ is an $\ell$-term with a negative determinative atom the proof proceeds
the same, substituting $\neg a$ for $a$ and applying \eqref{FFEL3} where
needed. For the inductive step we assume that the result holds for $*$-terms
with fewer $\ell$-terms than $P^* \fulland Q^d$ and $P^* \fullor Q^c$. We
note that each application of $\nf^n_1$ changes the main connective (not
occurring inside an $\ell$-term) and hence the result is a $*$-term. Derivable
equality is, given the induction hypothesis, an instance of (the dual of)
\eqref{FFEL2}.

With this result we can now see that $\nf^n(P^\tr \fulland Q^*)$ is indeed a
$\tr$-$*$-term. Furthermore we find that:
\begin{align*}
\nf^n(P^\tr \fulland Q^*) 
&= P^\tr \fulland \nf^n_1(Q^*)
&&\textrm{by definition} \\
&= P^\tr \fulland \neg Q^*
&&\textrm{as shown above} \\
&= (P^\tr \fullor \tr) \fulland \neg Q^*
&&\textrm{by Lemma \ref{lem:ptpf}} \\
&= \neg(P^\tr \fullor \tr) \fullor \neg Q^*
&&\textrm{by Lemma \ref{la:feqs}.\ref{eq:a5}} \\
&= \neg P^\tr \fullor \neg Q^*
&&\textrm{by Lemma \ref{lem:ptpf}} \\
&= \neg (P^\tr \fulland Q^*).
&&\textrm{by \eqref{FFEL2} and \eqref{FFEL3}}
\end{align*}
Hence for all $P \in \FNF$, $\EqFFEL \vdash \nf^n(P) = \neg P$.
\end{proof}

\begin{lemma}
\label{lem:nfc1}
For any $\tr$-term $P$ and $Q \in \FNF$, $\nf^c(P, Q)$ has the same
grammatical category as $Q$ and 
\begin{equation*}
\EqFFEL \vdash \nf^c(P, Q) = P \fulland Q.
\end{equation*}
\end{lemma}
\begin{proof}
By induction on the complexity of the first argument. In the base case we see
that $\nf^c(\tr, P) = P$ and hence has the same grammatical category as $P$.
Derivable equality follows from \eqref{FFEL5}.

For the induction step we make a case distinction on the grammatical category
of the second argument. If the second argument is a $\tr$-term we have that
$\nf^c(a \fullor P^\tr, Q^\tr) = a \fullor \nf^c(P^\tr, Q^\tr)$, where
we assume that $\nf^c(P^\tr, Q^\tr)$ is a $\tr$-term and $\EqFFEL
\vdash \nf^c(P^\tr, Q^\tr) = P^\tr \fulland Q^\tr$. The grammatical
claim follows immediately from the induction hypothesis. The claim about
derivable equality is proved as follows:
\begin{align*}
\nf^c(a \fullor P^\tr, Q^\tr)
&= a \fullor \nf^c(P^\tr, Q^\tr)
&&\textrm{by definition} \\
&= a \fullor (P^\tr \fulland Q^\tr)
&&\textrm{by induction hypothesis} \\
&= a \fullor (P^\tr \fulland (Q^\tr \fullor \tr))
&&\textrm{by Lemma \ref{lem:ptpf}} \\
&= (a \fullor P^\tr) \fulland (Q^\tr \fullor \tr)
&&\textrm{by Lemma \ref{la:feqs}.\ref{eq:a2}} \\
&= (a \fullor P^\tr) \fulland Q^\tr.
&&\textrm{by Lemma \ref{lem:ptpf}}
\end{align*}

If the second argument is an $\fa$-term we assume that $\nf^c(P^\tr,
Q^\fa)$ is an $\fa$-term and that $\EqFFEL \vdash \nf^c(P^\tr,
Q^\fa) = P^\tr \fulland Q^\fa$. The grammatical claim follows
immediately from the induction hypothesis. Derivable equality is proved as
follows:
\begin{align*}
\nf^c(a \fullor P^\tr, Q^\fa)
&= a \fulland \nf^c(P^\tr, Q^\fa)
&&\textrm{by definition} \\
&= a \fulland (P^\tr \fulland Q^\fa)
&&\textrm{by induction hypothesis} \\
&= a \fulland (P^\tr \fulland (Q^\fa \fulland \fa))
&&\textrm{by Lemma \ref{lem:ptpf}} \\
&= (a \fullor P^\tr) \fulland (Q^\fa \fulland \fa)
&&\textrm{by Lemma \ref{lem:feqs2}.\ref{eq:a3}} \\
&= (a \fullor P^\tr) \fulland Q^\fa.
&&\textrm{by Lemma \ref{lem:ptpf}} \\
\end{align*}

Finally, if the second argument is a $\tr$-$*$-term then $\nf^c(a \fullor
P^\tr, Q^\tr \fulland R^*) = \nf^c(a \fullor P^\tr, Q^\tr) \fulland
R^*$. The fact that this is a $\tr$-$*$-term follows from the fact that
$\nf^c(a \fullor P^\tr, Q^\tr)$ is a $\tr$-term as was shown above.
Derivable equality follows from the case where the second argument is a
$\tr$-term and \eqref{FFEL4}.
\end{proof}

\begin{lemma}
\label{lem:nfc4}
For any $\tr$-$*$-term $P$ and $\fa$-term $Q$, $\nf^c(P, Q)$ is an
$\fa$-term and
\begin{equation*}
\EqFFEL \vdash \nf^c(P, Q) = P \fulland Q.
\end{equation*}
\end{lemma}

\begin{proof}
By \eqref{FFEL4} and Lemma \ref{lem:nfc1} it suffices to show that
$\nf^c_2(P^*, Q^\fa)$ is an $\fa$-term and that $\EqFFEL \vdash
\nf^c_2(P^*, Q^\fa) = P^* \fulland Q^\fa$. We prove this by induction on
the number of $\ell$-terms in $P^*$. In the base cases, i.e., $\ell$-terms, the
grammatical claims follow from Lemma \ref{lem:nfc1}. The claim about derivable
equality in the case of $\ell$-terms with positive determinative atoms follows
from Lemma~\ref{lem:nfc1} and \eqref{FFEL4}. For $\ell$-terms with negative
determinative atoms it follows from Lemma \ref{lem:nfc1}, Lemma \ref{lem:ptpf},
\eqref{FFEL7}, \eqref{FFEL4} and \eqref{FFEL8}.

For the induction step we assume the claims hold for any $*$-terms with fewer
$\ell$-terms than $P^* \fulland Q^d$ and $P^* \fullor Q^c$. In the case of
conjunctions we have $\nf^c_2(P^* \fulland Q^d, R^\fa) = \nf^c_2(P^*,
\nf^c_2(Q^d, R^\fa))$ and the grammatical claim follows from the induction
hypothesis (applied twice). Derivable equality follows from the induction
hypothesis and \eqref{FFEL4}.

For disjunctions we have $\nf^c_2(P^* \fullor Q^c, R^\fa) = \nf^c_2(P^*,
\nf^c_2(Q^c, R^\fa))$ and the grammatical claim follows from the induction
hypothesis (applied twice). The claim about derivable equality is proved as
follows:
\begin{align*}
\nf^c_2(P^* \fullor Q^c, R^\fa)
&= \nf^c_2(P^*, \nf^c_2(Q^c, R^\fa))
&&\textrm{by definition} \\
&= P^* \fulland (Q^c \fulland R^\fa)
&&\textrm{by induction hypothesis} \\
&= P^* \fulland (Q^c \fulland (R^\fa \fulland \fa))
&&\textrm{by Lemma \ref{lem:ptpf}} \\
&= (P^* \fullor Q^c) \fulland (R^\fa \fulland \fa)
&&\textrm{by Lemma \ref{lem:feqs2}.\ref{eq:a3}} \\
&= (P^* \fullor Q^c) \fulland R^\fa.
&&\textrm{by Lemma \ref{lem:ptpf}}
\end{align*}
\end{proof}

\begin{lemma}
\label{lem:nfc2}
For any $\fa$-term $P$ and $Q \in \FNF$, $\nf^c(P, Q)$ is an
$\fa$-term and
\begin{equation*}
\EqFFEL \vdash \nf^c(P, Q) = P \fulland Q.
\end{equation*}
\end{lemma}
\begin{proof}
We make a case distinction on the grammatical category of the second argument.
If the second argument is a $\tr$-term we proceed by induction on the first
argument. In the base case we have $\nf^c(\fa, P^\tr) = \nf^n(P^\tr)$
and the result is by Lemma \ref{lem:nfn}, Lemma \ref{lem:ptpf}, \eqref{FFEL7}
and \eqref{FFEL8}. In the inductive case we have $\nf^c(a \fulland P^\fa,
Q^\tr) = a \fulland \nf^c(P^\fa, Q^\tr)$, where we assume that
$\nf^c(P^\fa, Q^\tr)$ is an $\fa$-term and $\EqFFEL \vdash
\nf^c(P^\fa, Q^\tr) = P^\fa \fulland Q^\tr$. The result now follows
from the induction hypothesis and \eqref{FFEL4}.

If the second argument is an $\fa$-term the proof is almost the same, except
that we need not invoke Lemma \ref{lem:nfn} or \eqref{FFEL8} in the base case.

Finally, if the second argument is a $\tr$-$*$-term we again proceed by
induction on the first argument. In the base case we have $\nf^c(\fa,
P^\tr \fulland Q^*) = \nf^c(P^\tr \fulland Q^*, \fa)$. The grammatical
claim now follows from Lemma \ref{lem:nfc4} and derivable equality follows
from Lemma \ref{lem:nfc4} and \eqref{FFEL7}. For for the inductive case the
results follow from the induction hypothesis and \eqref{FFEL4}.
\end{proof}

\begin{lemma}
\label{lem:nfc3}
For any $\tr$-$*$-term $P$ and $\tr$-term $Q$, $\nf^c(P, Q)$ has the same
grammatical category as $P$ and
\begin{equation*}
\EqFFEL \vdash \nf^c(P, Q) = P \fulland Q.
\end{equation*}
\end{lemma}

\begin{proof}
By \eqref{FFEL4} it suffices to prove the claims for $\nf^c_1$, i.e., that
$\nf^c_1(P^*, Q^\tr)$ has the same grammatical category as $P^*$ and that
$\EqFFEL \vdash \nf^c_1(P^*, Q^\tr) = P^* \fulland Q^\tr$. We prove this by
induction on the number of $\ell$-terms in $P^*$. In the base case we deal with
$\ell$-terms and the results follow from Lemma \ref{lem:nfc1} and
\eqref{FFEL4}.

For the inductive cases we assume that the results hold for any $*$-term with
fewer $\ell$-terms than $P^* \fulland Q^d$ and $P^* \fullor Q^c$.  In the case
of conjunctions the results follow from the induction hypothesis and
\eqref{FFEL4}. In the case of disjunctions the grammatical claim follows from
the induction hypothesis. For derivable equality we have:
\begin{align*}
\nf^c_1(P^* \fullor Q^c, R^\tr)
&= P^* \fullor \nf^c_1(Q^c, R^\tr)
&&\textrm{by definition} \\
&= P^* \fullor (Q^c \fulland R^\tr)
&&\textrm{by induction hypothesis} \\
&= P^* \fullor (Q^c \fulland (R^\tr \fullor \tr))
&&\textrm{by Lemma \ref{lem:ptpf}} \\
&= (P^* \fullor Q^c) \fulland (R^\tr \fullor \tr)
&&\textrm{by Lemma \ref{la:feqs}.\ref{eq:a2}} \\
&= (P^* \fullor Q^c) \fulland R^\tr.
&&\textrm{by Lemma \ref{lem:ptpf}} 
\end{align*}
\end{proof}

\begin{lemma}
\label{lem:nfc5}
For any $P, Q \in \FNF$, $\nf^c(P, Q)$ is in $\FNF$ and
\begin{equation*}
\EqFFEL \vdash \nf^c(P, Q) = P \fulland Q.
\end{equation*}
\end{lemma}

\begin{proof}
By the four preceding lemmas it suffices to show that $\nf^c(P^\tr \fulland
Q^*, R^\tr \fulland S^*)$ is in $\FNF$ and that $\EqFFEL \vdash \nf^c(P^\tr
\fulland Q^*, R^\tr \fulland S^*) = (P^\tr \fulland Q^*) \fulland (R^\tr
\fulland S^*)$. By \eqref{FFEL4}, in turn, it suffices to prove that
$\nf^c_3(P^*, Q^\tr \fulland R^*)$ is a $*$-term and that $\EqFFEL \vdash
\nf^c_3(P^*, Q^\tr \fulland R^*) = P^* \fulland (Q^\tr \fulland R^*)$. We
prove this by induction on the number of $\ell$-terms in $R^*$. In the base
case we have that $\nf^c_3(P^*, Q^\tr \fulland R^\ell) = \nf^c_1(P^*,
Q^\tr) \fulland R^\ell$. The results follow from Lemma \ref{lem:nfc3} and
\eqref{FFEL4}.

For the inductive cases we assume that the results hold for all $*$-terms with
fewer $\ell$-terms than $R^* \fulland S^d$ and $R^* \fullor S^c$. For
conjunctions the result follows from the induction hypothesis and
\eqref{FFEL4} and for disjunctions it follows from Lemma \ref{lem:nfc3} and
\eqref{FFEL4}.
\end{proof}

\noindent
\textbf{Theorem~\ref{thm:nf}.}
\textit{For any $P \in \FT$, $\nf(P)$ terminates, $\nf(P) \in \FNF$ and $\EqFFEL \vdash
\nf(P) = P$.}

\begin{proof}
By induction on the complexity of $P$. If $P$ is an atom, the result is by
\eqref{FFEL5} and \eqref{FFEL6}. If $P$ is $\tr$ or $\fa$ the result is
by identity. For the induction we assume that the result holds for all
$\FEL$-terms of lesser complexity than $P \fulland Q$ and $P \fullor Q$. The
result now follows from the induction hypothesis, Lemma \ref{lem:nfn}, Lemma
\ref{lem:nfc5} and \eqref{FFEL2}.
\end{proof}

\noindent\textbf{Correctness of $\bm\inv$.}~
\label{sec:felinv}

\noindent
\textbf{Theorem \ref{thm:felinv}}.
\textit{For all $P \in \FNF$, $\inv(\FE(P))=P$, i.e. $\inv(\fe(P))$ is syntactically equal 
to $P$ for $P\in\FNF$.
}

\begin{proof}
We first prove that for all $\tr$-terms $P$, $\inv^\tr(\FE(P)) = P$,
by induction on $P$. In the base case $P = \tr$ and we have
$\inv^\tr(\FE(P)) = \inv^\tr(\tr) = \tr = P$. For the
inductive case we have $P = a \fullor Q^\tr$ and
\begin{align*}
\inv^\tr(\FE(P)) &= \inv^\tr(\FE(Q^\tr) \tlef a \trig
  \FE(Q^\tr))
&&\textrm{by definition of $\FE$} \\
&= a \fullor \inv^\tr(\FE(Q^\tr))
&&\textrm{by definition of $\inv^\tr$} \\
&= a \fullor Q^\tr
&&\textrm{by induction hypothesis} \\
&= P.
\end{align*}

Similarly, we see that for all $\fa$-terms $P$, $\inv^\fa(\FE(P)) =
P$, by induction on $P$. In the base case $P = \fa$ and we have
$\inv^\fa(\FE(P)) = \inv^\fa(\fa) = \fa = P$. For
the inductive case we have $P = a \fulland Q^\fa$ and
\begin{align*}
\inv^\fa(\FE(P)) &= \inv^\fa(\FE(Q^\fa) \tlef a \trig
  \FE(Q^\fa))
&&\textrm{by definition of $\FE$} \\
&= a \fulland \inv^\fa(\FE(Q^\fa))
&&\textrm{by definition of $\inv^\fa$} \\
&= a \fulland Q^\fa
&&\textrm{by induction hypothesis} \\
&= P.
\end{align*}

Now we check that for all $\ell$-terms $P$, $\inv^\ell(\FE(P)) = P$.
We observe that either $P = a \fulland Q^\tr$ or $P = \neg a
\fulland Q^\tr$. In the first case we have
\begin{align*}
\inv^\ell(\FE(P)) &= \inv^\ell(\FE(Q^\tr) \tlef a
  \trig \FE(Q^\tr)\sub{\tr}{\fa})
&&\textrm{by definition of $\FE$} \\
&= a \fulland \inv^\tr(\FE(Q^\tr))
&&\textrm{by definition of $\inv^\ell$} \\
&= a \fulland Q^\tr
&&\textrm{as shown above} \\
&= P.
\end{align*}
In the second case we have that
\begin{align*}
\inv^\ell(\FE(P)) &= \inv^\ell(\FE(Q^\tr)\sub{\tr}{\fa} \tlef a
  \trig \FE(Q^\tr))
&&\textrm{by definition of $\FE$} \\
&= \neg a \fulland \inv^\tr(\FE(Q^\tr))
&&\textrm{by definition of $\inv_\ell$} \\
&= \neg a \fulland Q^\tr
&&\textrm{as shown above} \\
&= P.
\end{align*}

We now prove that for all $*$-terms $P$, $\inv^*(\FE(P)) = P$, by
induction on $P$ modulo the complexity of $\ell$-terms. In the base case we are
dealing with $\ell$-terms. Because an $\ell$-term has neither a cd nor a dd we
have $\inv^*(\FE(P)) = \inv^\ell(\FE(P)) = P$, where the first
equality is by definition of $\inv^*$ and the second was shown above. For the
induction we have either $P = Q \fulland R$ or $P = Q \fullor R$. In
the first case note that by Theorem \ref{thm:cddd}, $\FE(P)$ has a cd and no
dd. So we have 
\begin{align*}
\inv^*(\FE(P)) &= \inv^*(\cd_1(\FE(P))\ssub{\triangle_1}{\tr}{\triangle_2}{\fa})
  \fulland \inv^*(\cd_2(\FE(P)))
&&\textrm{by definition of $\inv^*$} \\
&= \inv^*(\FE(Q)) \fulland \inv^*(\FE(R))
&&\textrm{by Theorem \ref{thm:cddd}} \\
&= Q \fulland R
&&\textrm{by induction hypothesis} \\
&= P.
\end{align*}
In the second case, again by Theorem \ref{thm:cddd}, $\FE(P)$ has a dd and no
cd. So we have that
\begin{align*}
\inv^*(\FE(P)) &= \inv^*(\dd_1(\FE(P))\ssub{\triangle_1}{\tr}{\triangle_2}{\fa})
  \fullor \inv^*(\dd_2(\FE(P)))
&&\textrm{by definition of $\inv^*$} \\
&= \inv^*(\FE(Q)) \fullor \inv^*(\FE(R))
&&\textrm{by Theorem \ref{thm:cddd}} \\
&= Q \fullor R
&&\textrm{by induction hypothesis} \\
&= P.
\end{align*}

Finally, we prove the theorem's statement by making a case distinction on the
grammatical category of $P$. If $P$ is a $\tr$-term, then $\FE(P)$ has only
$\tr$-leaves and hence $\inv(\FE(P)) = \inv^\tr(\FE(P)) = P$,
where the first equality is by definition of $\inv$ and the second was shown
above. If $P$ is an $\fa$-term, then $\FE(P)$ has only $\fa$-leaves and
hence $\inv(\FE(P)) = \inv^\fa(\FE(P)) = P$, where the first
equality is by definition of $\inv$ and the second was shown above. If $P$ is a
$\tr$-$*$-term, then it has both $\tr$ and $\fa$-leaves and hence,
letting $P = Q \fulland R$,
\begin{align*}
\inv(\FE(P)) &= \inv^\tr(\tsd_1(\FE(P))\sub{\triangle}{\tr}) \fulland
  \inv^*(\tsd_2(\FE(P)))
&&\textrm{by definition of $\inv$} \\
&= \inv^\tr(\FE(Q)) \fulland \inv^*(\FE(R))
&&\textrm{by Theorem \ref{thm:tsd}} \\
&= Q \fulland R
&&\textrm{as shown above} \\
&= P,
\end{align*}
which completes the proof.
\end{proof}

\section{Proofs - Section~\ref{sec:4}}
\label{app:A4}

In order to show that the relation $=_{\mfe}$ is indeed a congruence, we first
repeat a part of~\cite[La.3.3]{BPS21}, 
except that we generalise its clause 5 to clause 2 below.\footnote{%
  The proof of clause 5 of~\cite[La.3.3]{BPS21} (p.266) contains two errors in Case $c=a$,
  correct is:
  ``and by induction $\Le_a(X_1[\tr\mapsto Y]) = \Le_a(X_1)[\tr\mapsto \Le_a(Y)] 
  = \Le_a(X)[\tr\mapsto \Le_a(Y)]$" (cf.~Case $b=a$ in the proof 
  below and take $Z=\fa$).
  }
  
\begin{lemma}
\label{la:seven}
For all $a\in A$, $f\in\{\Le,\Ri\}$ and  $X,Y,Z\in\NT$,
\begin{description}
\item[$1.$]
\label{Aux4} 
$\memt(f_a(X[\tr\mapsto \fa,\fa\mapsto \tr]))=\memt(f_a(X))[\tr\mapsto \fa,\fa\mapsto \tr]$,
\item[$2.$]
\label{Aux5}
$f_a(X[\tr\mapsto Y,\fa\mapsto Z])=f_a(X)[\tr\mapsto f_a(Y),\fa\mapsto f_a(Z)]$.
\end{description}
\end{lemma}

\begin{proof}
Clause 1 is Lemma 3.3.4 in~\cite{BPS21}.
We prove clause 2 by induction on the structure of $X$.
The base cases are trivial. 

If $X=X_1\unlhd b\unrhd X_2$, distinguish the cases $b=a$ and $b\ne a$:
\\[1mm]
Case $b=a$, subcase $f=\Le$:
$\Le_a(X[\tr\mapsto Y,\fa\mapsto Z])=\Le_a(X_1[\tr\mapsto Y,\fa\mapsto Z])$ and by
induction, $\Le_a(X_1[\tr\mapsto Y,\fa\mapsto Z])=\Le_a(X_1)[\tr\mapsto \Le_a(Y),\fa\mapsto \Le_a(Z)]$.
By $\Le_a(X)=\Le_a(X_1)$, we are done.
The other subcase $f=\Ri$ follows in a similar way.
\\[1mm]
Case $b\ne a$, subcase $f=\Le$:
\begin{align*}
\Le_a(X[\tr\mapsto Y,\fa\mapsto Z])
&=\Le_a(X_1[\tr\mapsto Y,\fa\mapsto Z])\unlhd b\unrhd \Le_a(X_2[\tr\mapsto Y,\fa\mapsto Z])
\\
&\stackrel{\text{IH}}=
\Le_a(X_1)[\tr\mapsto \Le_a(Y),\fa\mapsto \Le_a(Z)]\unlhd b\unrhd \Le_a(X_2)
[\tr\mapsto \Le_a(Y),\fa\mapsto \Le_a(Z)]
\\
&=\Le_a(X)[\tr\mapsto \Le_a(Y),\fa\mapsto \Le_a(Z)]. 
\end{align*}
The other subcase $f=\Ri$ follows in a similar way.
\end{proof}

\noindent
\textbf{Lemma \ref{la:Mcongruence}.}
\textit{The relation $=_{\mfe}$ is a congruence.}

\begin{proof}
Define $\treeneg X:\NT\to\NT$ and 
$X\treefulland Y,~X\treefullor Y:\NT\times\NT\to\NT$ by
\begin{align*}
\treeneg X&=X[\tr\mapsto \fa,\fa\mapsto\tr],
\\
~X\treefulland Y
&=X[\tr\mapsto Y,\fa\mapsto Y[\tr\mapsto\fa]],
\\
X\treefullor Y
&=X[\tr\mapsto Y[\fa\mapsto\tr],\fa\mapsto Y].
\end{align*}
Hence, $\treeneg(\fe(P)) = \fe(\neg P)$, $\fe(P)\treefulland(Q) = \fe(P \fulland Q)$, and 
$\fe(P)\treefullor\fe(Q) = \fe(P \fullor Q)$.
It suffices to show that for $X,Y\in\NT$, if $\memt(X)=\memt(X')$ and $\memt(Y)=\memt(Y')$, then 
$\memt(\treeneg X)=\memt(\treeneg X')$,
$\memt(X\treefulland Y)=\memt(X'\treefulland Y')$, and
$\memt(X\treefullor Y)=\memt(X'\treefullor Y')$.

The case for $\treeneg X$ follows by case distinction on the form of $X$.
The base cases $X\in\{\tr,\fa\}$ are trivial, and if $X=X_1\unlhd a\unrhd X_2$, then 
$\memt(X)=\memt(X')$ implies that 
$X'=X_1'\unlhd a\unrhd X_2'$ for some $X_1',X_2'\in\NT$, and
\begin{equation}
\label{enig}
\tag{Aux3}
\memt(\Le_a(X_1))=\memt(\Le_a(X_1'))
\quad\text{and}\quad
\memt(\Ri_a(X_2))=\memt(\Ri_a(X_2')).
\end{equation}

Write $\negtext$ for the leaf replacement $[\tr\mapsto \fa,\fa\mapsto \tr]$
and derive 
\begin{align*}
\memt(\treeneg X)
&=\memt((X_1\unlhd a\unrhd X_2)\negtext)
\\
&=\memt(\Le_a(X_1\negtext))\unlhd a\unrhd \memt(\Ri_a(X_2\negtext))
\\
&=\memt(\Le_a(X_1'\negtext))\unlhd a\unrhd \memt(\Ri_a(X_2'\negtext))
&&\text{by La.\ref{la:seven}.1 and \eqref{enig}}
\\
&=\memt(\treeneg X').
\end{align*}

The case for $\memt(X\treefulland Y)=\memt(X'\fulland Y')$:
for readability we split the proof obligation into two parts.

(A) $\memt(X\treefulland Y)=\memt(X'\treefulland Y)$.
This follows by induction on the depth of $X$.
The base cases $X\in\{\tr,\fa\}$ are simple: note that
if $X=\tr$ and $\memt(X)=\memt(X')$, then $X'=\tr$ and we are done, and similarly if $X=\fa$.

If $X=X_1\unlhd a\unrhd X_2$ and $\memt(X)=\memt(X')$, then by \eqref{enig}, 
$X'=X_1'\unlhd a\unrhd X_2'$ for some $X_1',X_2'\in\NT$ with 
$\memt(\Le_a(X_1))=\memt(\Le_a(X_1'))$ 
and $\memt(\Ri_a(X_2))=\memt(\Ri_a(X_2'))$.
Write 
\begin{align*}
\fullandtext
&\text{ for $[\tr\mapsto Y,\fa\mapsto Y[\tr\mapsto\fa]]$},
\\
\Lfullandtext
&\text{ for $[\tr\mapsto \Le_a(Y),\fa\mapsto\Le_a(Y)[\tr\mapsto\fa]]$},
\\
\Rfullandtext
&\text{ for $[\tr\mapsto \Ri_a(Y),\fa\mapsto\Ri_a(Y)[\tr\mapsto\fa]]$}.
\end{align*}
Derive
\begin{align*}
\memt(X\treefulland Y)
&=\memt((X_1\unlhd a\unrhd X_2)\treefulland Y)
\\
&=\memt(X_1\fullandtext\unlhd a\unrhd X_2\fullandtext)
\\
&=\memt(\Le_a(X_1\fullandtext))\unlhd a\unrhd \memt(\Ri_a(X_2\fullandtext))
\\
&=\memt(\Le_a(X_1)\Lfullandtext)\unlhd a\unrhd \memt(\Ri_a(X_2)\Rfullandtext)
&&\text{by La.\ref{la:seven}.2}\\
&=\memt(\Le_a(X_1)\treefulland \Le_a(Y))\unlhd a\unrhd \memt(\Ri_a(X_2)\treefulland \Ri_a(Y))
&&\text{}\\
&=\memt(\Le_a(X_1')\treefulland \Le_a(Y))\unlhd a\unrhd \memt(\Ri_a(X_2')\treefulland \Ri_a(Y))
&&\text{by IH and \eqref{enig}}\\
&=...=\memt(X'\treefulland Y).
\end{align*}

(B) $\memt(X\treefulland Y)=\memt(X\treefulland Y')$.
This follows by induction on the depth of $X$.
The base cases $X\in\{\tr,\fa\}$ are trivial.
If $X=X_1\unlhd a\unrhd X_2$ derive
\begin{align*}
\memt(X\treefulland Y)
&=\memt((X_1\unlhd a\unrhd X_2)\treefulland Y)
\\
&=\memt(X_1\fullandtext\unlhd a\unrhd X_2\fullandtext)
\\
&=\memt(\Le_a(X_1\fullandtext))\unlhd a\unrhd \memt(\Ri_a(X_2\fullandtext))
\\
&=\memt(\Le_a(X_1)\Lfullandtext)\unlhd a\unrhd \memt(\Ri_a(X_2)\Rfullandtext)
&&\text{by La.\ref{la:seven}.2}\\
&=\memt(\Le_a(X_1)\treefulland \Le_a(Y))\unlhd a\unrhd \memt(\Ri_a(X_2)\treefulland \Ri_a(Y))
&&\text{}\\
&=\memt(\Le_a(X_1)\treefulland \Le_a(Y'))\unlhd a\unrhd \memt(\Ri_a(X_2)\treefulland \Ri_a(Y'))
&&\text{by IH}\\
&=...=\memt(X\treefulland Y').
\end{align*}

The proof for the case $\memt(X\treefullor Y)=\memt(X'\treefullor Y')$
is similar.
\end{proof}

For the continuation of the proof of Lemma~\ref{la:Msound}, we use results about
short-circuit connectives and the memorising short-circuit evaluation function 
$\mse(P)= 
\memt(\se(P))$ defined
in~\cite{BPS21} (available online); for the  function $\memt$ see Definition~\ref{def:mfe}
and below we recall the definition of the short-circuit evaluation
function \se. In particular, we use the 
definability of the connectives ${\fulland}$ and ${\fullor}$ in short-circuit logic
(cf.\ Section~\ref{sec:7}). 

Given a set of atoms $A$, the set $\SPsc$ of closed terms with short-circuit connectives
is  defined by the following grammar ($a\in A$):
\[P ::= \tr\mid\fa\mid a\mid\neg P\mid P\leftand P\mid P\leftor P.\]
For $a\in A$, the evaluation function $\se:\SPsc\to\NT$ and the translation function $\tra:\SP\to\SPsc$ 
are defined by
\begin{align*}
&\se(\tr)=\tr,~
\se(\fa)=\fa,~~
&&\tra(\tr)=\tr,~~
\tra(\fa)=\fa,
\\
&\se(a)=\tr\unlhd a\unrhd \fa,
&&\tra(a)=a,~
\\
&\se(\neg P)=\se(P)[\tr\mapsto\fa, \fa\mapsto\tr],~
&&\tra(\neg P)=\neg(\tra(P)),~
\\
&\se(P\leftand Q)=\se(P)[\tr\mapsto \se(Q)],
&&\tra(P\fulland Q)=(\tra(P)\leftor(\tra(Q)\leftand\fa))\leftand\tra(Q),\\
&\se(P\leftor Q)=\se(P)[\fa\mapsto \se(Q)],
&&\tra(P\fullor Q)=(\tra(P)\leftand(\tra(Q)\leftor\tr))\leftor\tra(Q).
\end{align*}

\begin{lemma}
\label{la:aux}
For all $P\in\SP$, $\fe(P)=\se(\tra(P))$.
\end{lemma}

\begin{proof}
By structural induction on $P$. The base cases are immediate, and so is $\fe(\neg P)=\se(\tra(\neg P))$.

For the case $P=Q\fulland R$, derive 
\begin{align*}
\se(\tra(P))
&=\se((\tra(Q)\leftor(\tra(R)\leftand\fa))\leftand\tra(R))
\\
&=\se(\tra(Q)\leftor(\tra(R)\leftand\fa))[\tr\mapsto\se(\tra(R))]
\\
&=(\se(\tra(Q))[\fa\mapsto\se(\tra(R)\leftand\fa)])[\tr\mapsto\se(\tra(R))].
\\
&=(\se(\tra(Q))[\fa\mapsto\se(\tra(R))[\tr\mapsto\fa])[\tr\mapsto\se(\tra(R))]
\\
&=\se(\tra(Q))[\tr\mapsto\se(\tra(R)),~\fa\mapsto\se(\tra(R))[\tr\mapsto\fa][\tr\mapsto\se(\tra(R))]]
\\
&=\se(\tra(Q))[\tr\mapsto\se(\tra(R)),~\fa\mapsto\se(\tra(R))[\tr\mapsto\fa]]
&&(\ast) 
\\
&=\fe(Q)[\tr\mapsto\fe(R),\fa\mapsto\fe(R)[\tr\mapsto\fa]]
&&\text{by IH}
\\
&=\fe(P).
\end{align*}
\indent$(\ast)$: \text{By redundancy, $\se(\tra(R))[\tr\mapsto\fa]$ has no \tr-leaves.}
\\[2mm]\indent
The case  $P=Q\fullor R$ follows in a similar way.
\end{proof}

\begin{proof}[\textbf{Continuation} of the proof of Lemma~$\ref{la:Msound}$]
For a variable $v\in\{x,y,z\}$, define $\tra(v)=v$. It follows with \emph{Prover9} that 
\begin{equation}
\label{eq:p9}
\MSCLe\vdash \tra((x\fullor y)\fulland z)=\tra((\neg x\fulland (y\fulland z))\fullor(x\fulland z)),
\end{equation}
where \MSCLe\ consists of the axioms~\eqref{Mem1}--\eqref{Mem5} in Table~\ref{tab:SSCL},
see Comment~\ref{com:x} below.
From~\cite[Thm.4.1]{BPS21} it follows that for all $P,Q\in\SPsc$, 
$\MFELe\vdash P=Q\Longrightarrow \memt(\se(P))=\memt(\se(Q))$.
Hence, for all $P,Q,R \in\SP$,
\[\memt(\se(\tra((P\fullor Q)\fulland R)))=
\memt(\se(\tra((\neg P\fulland (Q\fulland R))\fullor(P\fulland R)))),
\]
so, what has to be proved, i.e.,
\[
\memt(\fe((P \fullor Q)\fulland R)) = \memt(\fe((\neg P\fulland(Q\fulland R))\fullor (P\fulland R))),
\]
 follows from Lemma~\ref{la:aux}.
\end{proof}

\begin{comment}
\label{com:x}
A proof with \emph{Prover9}'s options \texttt{lpo} and \texttt{unfold}, and
with as assumptions the five axioms of \MSCLe\ was obtained in 0.02s.
The syntax of the proven 
goal for \emph{Prover9}, that is, the \tra()-translation of~\eqref{eq:p9} where 
\verb+^,v,+\texttt{\textquotesingle,1,0} stand for ${\leftand}, \leftor,\neg, \tr,\fa$, is:
\begin{verbatim}
(((x ^ (y v 1)) v y) v (z ^ 0)) ^ z = 
  (((x' v (((y v (z ^ 0)) ^ z) ^ 0)) ^ ((y v (z ^ 0)) ^ z)) ^ (((x v (z ^ 0)) ^ z) v 1)) 
v ((x v (z ^ 0)) ^ z).
\end{verbatim}
\end{comment}

\section{Proofs - Section~\ref{sec:6}}
\label{app:A6}

We start with two auxiliary lemmas.

\begin{lemma}
\label{la:B}
$\MFELe\vdash(x\fulland \fa) \fullor (x\fulland y) = (x\fulland \fa) \fullor (y\fulland x)$.
\end{lemma}

\begin{proof}
With \emph{Prover9} with options \texttt{lpo} and \texttt{unfold}:  0.2s.
\end{proof}

\begin{lemma} 
\label{la:C}
For all $\beta\in A^s_o$ and $P\in\SP$ such that $\alpha(P)\subset\alp{\beta}$,
$\MFELe\vdash\wfa_\beta\fullor (P\fulland\fa)=\wfa_\beta$.
\end{lemma}

\begin{proof}
By induction on the length of $\beta$. The base case $\beta=\epsilon $ is trivial.

For $\beta=a\gamma~(a\in A)$, thus $\wfa_\beta=a\fulland\wfa_\gamma$, we 
apply structural induction on $P$, and we
note that 
from Lemma~\ref{la:feqs}.3 and $\wfa_\gamma=\wfa_\gamma\fulland\fa$ it follows that 
\begin{equation}
\label{star}
\tag{C.$\dagger$}
\wfa_\beta=(a\fulland\fa)\fullor\wfa_\gamma.
\end{equation}
The base cases: 
1) the case
$P\in\{\tr,\fa\}$
is trivial; 
2) if $P=a$, then by \eqref{star} and the dual of $\eqref{C1}$, $\wfa_\beta\fullor (a\fulland\fa)
=((a\fulland\fa)\fullor
\wfa_\gamma)\fullor (a\fulland\fa)=\wfa_\beta$;
3) if $P\in\alp{\gamma}$, then by \eqref{star}, 
$\wfa_\beta\fullor(P\fulland\fa)=((a\fulland\fa)\fullor\wfa_\gamma)\fullor(P\fulland\fa)
=(a\fulland\fa)\fullor(\wfa_\gamma\fullor(P\fulland\fa))
\stackrel{\text{IH}}=(a\fulland\fa)\fullor\wfa_\gamma=\wfa_\beta$.

If $P=\neg P_1$, then $\wfa_\beta\fullor (\neg P_1\fulland\fa)
\stackrel{\eqref{FFEL8}}=\wfa_\beta\fullor (P_1\fulland\fa)\stackrel{\text{IH}}=\wfa_\beta$.

If $P= P_1\fulland P_2$, then $\wfa_\beta\fullor ((P_1\fulland P_2)\fulland\fa)
=\wfa_\beta\fullor ((P_1\fulland\fa)\fulland(P_2\fulland\fa))
\stackrel{\eqref{C4}'}=\\(\wfa_\beta\fullor (P_1\fulland\fa))\fulland(\wfa_\beta\fullor (P_2\fulland\fa))
\stackrel{\text{IH}}=\wfa_\beta\fulland\wfa_\beta=\wfa_\beta$, where $\eqref{C4}'$ is the dual of
\eqref{C4}.

If $P= P_1\fullor P_2$, then $\wfa_\beta\fullor ((P_1\fullor P_2)\fulland\fa)
\stackrel{\eqref{M1}}=\wfa_\beta\fullor ((\neg P_1\fulland(P_2\fulland\fa))\fullor(P_1\fulland\fa))
\stackrel{\text{La.\ref{la:feqs}.1}}=
\\
\wfa_\beta\fullor ((P_1\fulland(P_2\fulland \fa)\fullor(P_1\fulland\fa))
=(\wfa_\beta\fullor ((P_1\fulland P_2)\fulland\fa))\fullor(\wfa_\beta\fullor (P_1\fulland\fa))
\stackrel{(pc), \text{IH}}=\wfa_\beta\fullor\wfa_\beta=\wfa_\beta$, 
\\[1mm]
where (pc) refers to the previous case $P=P_1\fulland P_2$.
\end{proof}

\begin{proof}[\textbf{Continuation} of the proof of Lemma~$\ref{la:CLsound}$]~

\noindent
1. Let $P,Q\in\SP$. We have to prove that 
\[\mfe(\widetilde\fa_\beta\fullor (P\fulland Q))=\mfe(\wfa_\beta\fullor (Q\fulland P)),
\]
where $\beta\in A^s_o$ satisfies $\alp{\beta}=\alpha(P\fulland Q)$.
By soundness of \MFELe\ (Lemma~\ref{la:Msound}), it suffices to prove 
that for any $\beta\in A^s_o$ that satisfies $\alpha(P)\subset\alp{\beta}$,
\[
\MFELe\vdash \wfa_\beta\fullor (P\fulland Q)=\wfa_\beta\fullor (Q\fulland P).
\]
Derive
\begin{align*}
\MFELe\vdash \wfa_\beta\fullor (P\fulland Q)
&=(\wfa_\beta\fullor(P\fulland\fa))\fullor (P\fulland Q)
&&\text{by Lemma~\ref{la:C}}
\\
&=\wfa_\beta\fullor((P\fulland\fa)\fullor (P\fulland Q))
\\
&=\wfa_\beta\fullor((P\fulland\fa)\fullor (Q\fulland P))
&&\text{by Lemma~\ref{la:B}}
\\
&=\wfa_\beta\fullor(Q\fulland P).
&&\text{(as above)}
\end{align*}

\noindent
2. The soundness of \CLFELeu\ follows as in 1, the additional case is that 
for any $P\in\SPu\setminus\SP$ and $Q\in\SPu$,
\(
\clfe(P\fulland Q)=\clfe(Q\fulland P)=\und.
\)
\end{proof}


\addcontentsline{toc}{section}{References}
\small

\begin{thebibliography}{99}

\bibitem{BBR95}
Bergstra, J.A., Bethke, I., and Rodenburg, P.H.  (1995).
A propositional logic with 4 values: true, false, divergent and meaningless,
\emph{Journal of Applied Non-Classical Logics}, 5(2):199-217.
\\
\url{https://doi.org/10.1080/11663081.1995.10510855}

\bibitem{BP15}
Bergstra, J.A. and Ponse, A. (2015). 
Evaluation trees for proposition algebra. 
\\
\newblock 
\url{https://doi.org/10.48550/arXiv.1504.08321}
[cs.LO].

\bibitem{BP23}
Bergstra, J.A. and Ponse, A. (2023). 
Conditional logic as a short-circuit logic. 
\\
\url{https://doi.org/10.48550/arXiv.2304.14821} [cs.LO].

\bibitem{BPS13}
Bergstra, J.A., Ponse, A., and Staudt, D.J.C. (2010 / 2013). 
Short-circuit logic. 
(First version appeared in 2010.)
\newblock \url{https://doi.org/10.48550/arXiv.1010.3674} [cs.LO,math.LO]. 

\bibitem{BPS18}
Bergstra, J.A., Ponse, A., and Staudt, D.J.C. (2018). 
Propositional logic with short-circuit evaluation: a non-commutative and a commutative variant.
\\
\newblock \url{https://doi.org/10.48550/arXiv.1810.02142} [cs.LO,math.LO]. 

\bibitem{BPS21}
Bergstra, J.A., Ponse, A., and Staudt, D.J.C. (2021). 
\newblock Non-commutative propositional logic with short-circuit evaluation. 
\newblock \emph{Journal of Applied Non-Classical Logics}, 31(3-4):234-278 (online available).
\\
\url{https://doi.org/10.1080/11663081.2021.2010954}

\bibitem{Blo11}
Blok, A. (2011). \emph{{Side-Effecting Logic}}, 
[Unpublished manuscript], University of Amsterdam.

\bibitem{Boc38}
Bochvar, D.A. (1938). 
On a 3-valued logical calculus and its application to the analysis of contradictions
(in Russian). Matematicheskii Sbornik, 4(2):287-308.

\bibitem{Burris}
Burris, S.N. and Sankappanavar, H.P.  (2012).
\newblock \emph{A Course in Universal Algebra - 
The Millennium Edition}.
\newblock Available at \url{http://www.math.uwaterloo.ca/~snburris/htdocs/ualg.html}.

\bibitem{Chu}
Church, A. (1956).
\emph{Introduction to Mathematical Logic}. Princeton University Press.

\bibitem{CdG20}
Cornets de Groot, S.H. (2012).
\newblock Logical systems with left-sequential versions of NAND and XOR.
\newblock MSc.~thesis Logic, University of Amsterdam (June 2020).
\\
\url{https://eprints.illc.uva.nl/id/eprint/1743/1/MoL-2020-02.text.pdf}.

\bibitem{GS90}
Guzmán, F. and Squier, C.C. (1990). 
The algebra of conditional logic. 
\emph{Algebra Universalis}, 27:88-110. 
\url{https://doi.org/10.1007/BF01190256}

\bibitem{Hoa85}
Hoare, C.A.R. (1985).
A couple of novelties in the propositional calculus. 
\newblock \emph{Zeitschrift für Mathematische Logik und 
Grundlagen der Mathematik}, 31(2):173-178.
\url{https://doi.org/10.1002/malq.19850310905}

\bibitem{Kle38}
Kleene, S.C. (1938).
\newblock On a notation for ordinal numbers. 
\emph{Journal of Symbolic Logic}, 3:150-155.

\bibitem{McC63}
McCarthy, J. (1963).
\newblock A basis for a mathematical theory of computation.
In P.~Braffort and D.~Hirschberg (eds.),
\emph{Computer Programming and Formal Systems},
pages~33--70, North-Holland, Amsterdam.

\bibitem{Prover9}
McCune W. (2008). The GUI: Prover9 and Mace4 with a graphical user interface. 
Prover9-Mace4 Version 0.5B.
(Prover9-Mace4-v05B.zip, March 14, 2008). 
\\
\url{https://www.cs.unm.edu/~mccune/prover9/gui/v05.html}.

\bibitem{PaPo22}
Papuc, D. and Ponse, A. (2022). 
Non-commutative propositional logic with short-circuited biconditional and NAND. 
\url{https://doi.org/10.48550/arXiv.2203.09321} [cs.LO].

\bibitem{PS18}
Ponse, A. and Staudt, D.J.C. (2018). 
\newblock An independent axiomatisation for free short-circuit logic. 
\newblock \emph{Journal of Applied Non-Classical Logics}, 28(1), 35-71 (online available).
\\
\url{https://doi.org/10.1080/11663081.2018.1448637} 

\bibitem{Stau}
Staudt, D.J.C. (2012).
\newblock Completeness for two left-sequential logics.
\newblock MSc. thesis Logic, University of Amsterdam (May 2012).
\url{https://doi.org/10.48550/arXiv.1206.1936} [cs.LO].

\end{thebibliography}
\end{document}